\documentclass[11pt]{article}

\usepackage[margin=1in]{geometry}
\usepackage{authblk}
\usepackage{microtype}
\usepackage{amsmath,amssymb,amsthm,mathtools}
\usepackage{aliascnt}
\usepackage{physics}
\usepackage{booktabs}
\usepackage{array}
\usepackage{algorithm}
\usepackage{algpseudocode}
\usepackage{graphicx}
\usepackage{flafter}
\usepackage{placeins}
\usepackage{tikz}
\usepackage{xcolor}
\usepackage{hyperref}
\usepackage[capitalize,nameinlink]{cleveref}
\RemoveFromHook{label}[firstaid/cleveref] 
\hypersetup{
    colorlinks=true,linkcolor=blue!55!black,
    citecolor=blue!55!black,urlcolor=blue!55!black}
\allowdisplaybreaks
\newtheorem{theorem}{Theorem}[section]
\newtheorem{proposition}[theorem]{Proposition}
\newtheorem{lemma}[theorem]{Lemma}

\newtheorem{corollary}[theorem]{Corollary}

\theoremstyle{definition}
\newtheorem{definition}[theorem]{Definition}
\newtheorem{example}[theorem]{Example}
\newtheorem{remark}[theorem]{Remark}

\newcommand{\FF}{\mathbb{F}}
\newcommand{\CC}{\mathbb{C}}
\newcommand{\RR}{\mathbb{R}}
\newcommand{\cc}{\mathrm{cc}}
\newcommand{\cl}{\mathrm{cl}}

\newcommand{\calB}{\mathcal{B}}
\newcommand{\calC}{\mathcal{C}}
\newcommand{\calH}{\mathcal{H}}
\newcommand{\calL}{\mathcal{L}}
\newcommand{\calP}{\mathcal{P}}
\newcommand{\calQ}{\mathcal{Q}}
\newcommand{\calR}{\mathcal{R}}
\newcommand{\calN}{\mathcal{N}}

\newcommand{\calS}{\mathcal{S}}
\newcommand{\calT}{\mathcal{T}}

\newcommand{\sfP}{\mathsf{P}}
\newcommand{\sfS}{\mathsf{S}}
\newcommand{\sfN}{\mathsf{N}}
\newcommand{\sfT}{\mathsf{T}}

\newcommand{\im}{\operatorname{im}}
\newcommand{\leaf}{\operatorname{leaf}}
\newcommand{\rw}{\operatorname{rw}}
\newcommand{\tw}{\operatorname{tw}}
\newcommand{\poly}{\operatorname{poly}}
\newcommand{\syn}{\operatorname{syn}}
\newcommand{\supp}{\operatorname{supp}}
\newcommand{\Span}{\operatorname{span}}

\newcommand{\argmax}{\operatorname*{arg\,max}}

\newcommand{\rowsp}{\operatorname{rowsp}}
\newcommand{\sigmoid}{\operatorname{sigmoid}}

\title{Exact Maximum Likelihood Decoding beyond Treewidth via Rank-Decomposition Dynamic Programming}
\author[1]{Bin Cheng\thanks{\href{mailto:bincheng@nus.edu.sg}{bincheng@nus.edu.sg}}}
\author[2]{Feng Pan\thanks{\href{mailto:feng_pan@sutd.edu.sg}{feng\_pan@sutd.edu.sg}}}
\affil[1]{Centre for Quantum Technologies, National University of Singapore, Singapore}
\affil[2]{Science, Mathematics and Technology Cluster, Singapore University of Technology and Design, Singapore}
\date{}
\begin{document}
\maketitle
\begin{abstract}
Maximum-likelihood decoding provides an optimal decoding strategy for quantum error correction under stochastic Pauli noise.
However, computing logical-class probabilities is challenging, and the leading exact tensor-network contraction requires time exponential in treewidth.
In this work, we introduce a new decoding algorithm based on rank-decomposition dynamic programming (Rank DP).
We express decoding partition functions with independent local fault factors as quadratic sums of powers and apply Rank DP.
The resulting exact evaluator has arithmetic complexity polynomial in the input size and exponential in the Tanner graph's rank-width with respect to the fault partition, including decomposition construction.
After Gaussian elimination, it gives polynomial-time decoding for punctured quantum Reed-Muller codes and a Steane-concatenated family with growing distance under independent single-qubit Pauli noise.
Standard contraction of the corresponding tensor networks requires superpolynomial time.
A nonnegative realization gives relative floating-point error bounds under explicit arithmetic assumptions.
Numerical experiments demonstrate runtime advantages over the tested tensor-network implementations on selected code-capacity and circuit-level instances, with full likelihood evaluation for Reed-Muller codes up to $1{,}023$ qubits.
We further use these likelihoods to learn circuit noise parameters from syndromes, evaluate rare postselection probabilities, and quantify decoder optimality gaps.
Our work opens new avenues for exploiting algebraic structure in quantum decoding and noise characterization.
\end{abstract}

\section{Introduction}
\label{sec:introduction}

Quantum error correction protects logical quantum information by encoding it across physical qubits and repeatedly measuring error syndromes.
Its effectiveness depends on a classical decoder that uses these observations to determine a recovery operation~\cite{gottesman_1997_stabilizer,cao_2026_maximum}.
Under a correctly specified stochastic Pauli noise model, maximum-likelihood decoding (MLD) minimizes the probability of selecting an incorrect logical class.
It therefore provides both an optimal decoding rule and a benchmark for practical decoders.
The computational cost of implementing this rule is a central problem in quantum error correction.

The defining feature of quantum MLD is the need to account for degeneracy.
For a fixed syndrome, errors in the same logical class differ by a stabilizer and are corrected by the same recovery.
MLD therefore selects the class with the largest total error probability~\cite{pelchat_2013_degenerate,bravyi_2014_efficient}.
Under local stochastic noise, these class probabilities are partition functions, reducing MLD to their evaluation and comparison~\cite{chubb_2021_statistical}.
Exact degenerate MLD is \#P-hard under polynomial-time Turing reductions, motivating algorithms that exploit additional structure~\cite{iyer_2015_hardness}.

Tensor-network (TN) methods are among the leading approaches to evaluating or approximating decoding partition functions.
They have been developed for surface codes and general two-dimensional Pauli codes, and extended to higher-dimensional codes and circuit-level noise~\cite{bravyi_2014_efficient,chubb_2021_general,piveteau_2024_tensor}.
Recent work also uses exact TN likelihoods to estimate circuit-level noise parameters directly from syndrome data~\cite{cao_2026_differentiable}.
These methods enable high-accuracy inference, but exact contraction can require time exponential in the treewidth of the underlying graph~\cite{markov_2008_simulating}.
Approximate contraction uses truncation to reduce this cost, but does not in general preserve exact likelihoods.

The constraints in a decoding problem contain algebraic structure that is not fully described by the connectivity of a tensor network.
In the binary representation of Pauli faults, syndrome and logical constraints are linear equations over $\FF_2$.
Many constraints or partial assignments can therefore carry the same information across a partition.
This suggests that the number of independent parity conditions, rather than the size of a dense tensor boundary, may provide a more succinct description of the computation.
The resulting question is:
\begin{quote}
    \emph{Can exact maximum-likelihood decoding go beyond tensor-network contraction?}
\end{quote}

We answer this question by expressing decoding likelihoods as weighted quadratic sums of powers (SOPs).
Applying the rank-decomposition dynamic program~\cite{decolnet_2026_quadratic}, we obtain an exact evaluator whose complexity is exponential in the rank-width of the underlying Tanner graph.
To our knowledge, this is the first application of this quadratic SOP/Rank DP framework to exact quantum MLD.
Our contribution is the decoding formulation, the resulting complexity guarantees and explicit tractable code families, and numerical evidence identifying instances where the approach is useful.
Gaussian elimination further reduces the graph on which this evaluator operates.
For two families based on punctured quantum Reed-Muller codes, the reduced representation gives polynomial-time MLD, whereas dense pairwise contraction of the specified original tensor networks requires superpolynomial time.

\subsection{Main results}
\label{subsec:intro-main-results}

The decoding likelihoods under independent local Pauli noise can be expressed as partition functions
\begin{align}
    Z_{\vb{M}}(\vb{t};\vb{w})
    =\sum_{\substack{\vb{e}\in\FF_2^d\\\vb{M}\vb{e}=\vb{t}}}
        \prod_{B\in\calB}w_B(\vb{e}_B),
    \label{eq:intro-partition-function}
\end{align}
where $\vb{M}\in\FF_2^{m\times d}$ and $\vb{t}$ specify the syndrome and logical constraints, and each $w_B$ gives the joint distribution of a local fault block $B\in\calB$.
This formulation covers both code-capacity and circuit-level decoding.
By expressing $Z_{\vb{M}}(\vb{t};\vb{w})$ as a quadratic SOP and applying the rank-decomposition dynamic programming (Rank DP)~\cite{decolnet_2026_quadratic}, we obtain the following complexity bound.

\begin{theorem}[Informal version of \cref{thm:sop-dp}]
\label{thm:intro-likelihood}
Let $N:=m+d+\sum_{B\in\calB}2^{|B|}$.
For explicitly given finite local weight tables, the partition function in \cref{eq:intro-partition-function} can be evaluated exactly using $\poly(N)2^{O(\rw(T(\vb{M});\calC))}$ arithmetic operations, where $\rw(T(\vb{M});\calC)$ is the rank-width of the Tanner graph $T(\vb{M})$ with respect to the partition $\calC$ into fault blocks and singleton checks.
\end{theorem}

Thus, logarithmic rank-width and polynomially many logical classes give polynomial-time MLD in the arithmetic model.
Our second main result compares rank-decomposition dynamic program for MLD with standard TN contraction.

\begin{theorem}[Informal version of \cref{thm:sop-tn-separation}]
\label{thm:intro-tn-separation}
Assume that the fault blocks have bounded size, and let $G$ be the decoding TN graph defined in \cref{prop:decoding-tensor-network}.
Given an optimal rank decomposition of $G$, the Rank DP evaluates the likelihood using $\poly(N)2^{O(\rw(G))}$ arithmetic operations.
Dense pairwise contraction of $G$ requires $2^{\Omega(\tw(G))}$ time, where $\rw(G)$ and $\tw(G)$ denote its rank-width and treewidth, respectively.
\end{theorem}

To obtain explicit separations, we further reduce the SOP graph by Gaussian elimination without increasing the decomposition width (\cref{cor:systematic-sop}).
For punctured quantum Reed-Muller codes and a suitable Steane concatenation~\cite{anderson_2014_conversion,hastings_2018_distillation,jochym_2014_concatenated,chamberland_2016_thresholds}, the reduced representation gives polynomial-time MLD under independent single-qubit Pauli noise (\cref{cor:rm-tn-separation}).
Dense pairwise contraction of the corresponding TN representations requires superpolynomial time $2^{\Omega(\sqrt{n}/(\log n)^{3/2})}$.

We also give a nonnegative realization of Rank DP that preserves the table sizes and arithmetic complexity (\cref{prop:nonnegative-realization}).
It uses only nonnegative sums and products, avoiding cancellation and admitting relative roundoff-error bounds under explicit floating-point assumptions (Supplemental \cref{app:numerical-accuracy}).

We complement these theoretical results with numerical experiments assessing the practical performance of Rank DP for exact MLD.
The Reed-Muller experiments illustrate the gap between decomposition widths after reduction and treewidth lower bounds for the original graphs, with full likelihood evaluation up to $1{,}023$ qubits.
Rank DP achieves shorter runtimes than the tested TN implementations on selected code-capacity and circuit-level decoding instances.
We allow additional algebraic optimizations for the TN implementations, whereas Rank DP uses the general algorithm without code-specific algebraic optimization.
On color-code memory circuits, it computes both logical-observable sector probabilities with certified relative accuracy and determines the maximum-likelihood label.
We further use these likelihoods to learn circuit noise models from syndrome data, evaluate rare-event probabilities under postselection, and quantify decoder optimality gaps.

\subsection{Related work}
\label{subsec:intro-related-work}

\paragraph{Maximum-likelihood decoding and partition functions.}
The distinction between the most probable physical error and the most probable logical class is central to degenerate decoding~\cite{pelchat_2013_degenerate}.
Iyer and Poulin~\cite{iyer_2015_hardness} establish the counting complexity of optimal stabilizer decoding.
Chubb and Flammia~\cite{chubb_2021_statistical} develop statistical-mechanical mappings that account for correlations in Pauli noise.
These results motivate treating logical-class probabilities as partition functions and seeking structural conditions for their efficient evaluation.
Exact decoding also has an established connection to graphical representations of codes, including optimal decoding of concatenated quantum codes~\cite{poulin_2006_concatenated} and classical normal and tree realizations~\cite{forney_2001_normal,kashyap_2009_tree}.
Quantum stabilizer trellises provide minimal sequential state representations supporting error estimation and posterior inference~\cite{ollivier_2006_trellises}, while degenerate Viterbi decoding accounts for stabilizer-equivalent errors~\cite{pelchat_2013_degenerate}.
Polynomial exact decoding is also available in special circuit-level settings, including repetition-code noise models admitting a planar partition-function representation~\cite{cao_2025_exact}.
Our nonnegative realization uses the state spaces of linear-code tree realizations and relates the resulting dynamic program to the SOP representation while preserving the local fault blocks.

\paragraph{Tensor-network decoding.}
Bravyi, Suchara, and Vargo~\cite{bravyi_2014_efficient} give an exact surface-code decoder for independent bit-flip and phase-flip noise and an approximate MPS decoder for more general noise.
Chubb~\cite{chubb_2021_general} develops a general TN decoder for two-dimensional stabilizer and subsystem codes under Pauli noise.
Piveteau, Chubb, and Renes~\cite{piveteau_2024_tensor} extend TN decoding to three-dimensional codes and noisy syndrome extraction, including circuit-level noise.
These algorithms use the tensor structure to evaluate or approximate the same logical-class sums considered here.

\paragraph{Learned decoding and noise estimation.}
Learning can improve decoding either by predicting logical labels or by estimating the noise parameters supplied to a decoder.
AlphaQubit, introduced by Bausch et al.~\cite{bausch_2024_learning}, is a recurrent transformer decoder that outperforms the TN baseline reported in the same Sycamore experiments.
Sivak, Newman, and Klimov~\cite{sivak_2024_optimization} instead use reinforcement-learning-inspired optimization to calibrate decoder priors against logical error rates.
Differentiable exact-likelihood methods estimate noise parameters by optimizing syndrome likelihoods or conditional logical-outcome likelihoods~\cite{cao_2026_differentiable,pancotti_2026_exact_learning}.
Complementary methods estimate decoding weights and fault priors from syndrome correlations, including higher-order correlations for hypergraph models~\cite{spitz_2018_adaptive,takou_2025_estimating,remm_2025_correlations}.
Our noise-learning example uses Rank DP as an alternative exact likelihood evaluator for the established syndrome-only estimation objective.

\paragraph{Sums of powers and quantum-circuit simulation.}
SOP represents circuit amplitudes as sums with polynomial phases, providing an algebraic approach to classical simulation~\cite{dawson_2005_polynomial}.
This representation has also been applied to complexity theory~\cite{montanaro_2017_quantum} and quantum circuit analysis~\cite{amy_2019_verification}.
To evaluate these sums, a novel simulation paradigm called FeynmanDD combines SOP representations with decision diagrams~\cite{wang_2025_feynmandd}.
The complexity of this approach is governed by linear rank-width, which permits separations from standard TN contraction~\cite{cheng_2025_breaking}.
Extending this computation from linear layouts to rank-decomposition trees gives a dynamic program parameterized by rank-width~\cite{decolnet_2026_quadratic}.
Here, we apply this framework to the exact evaluation of decoding likelihoods.

\subsection{Organization}

The remaining paper is organized as follows.
In \cref{sec:preliminaries}, we review the stabilizer formalism, binary representation and rank-width.
Then, we specify the noise models and formulate MLD as partition-function evaluation in \cref{sec:decoding}.
This formulation leads to the SOP representation and Rank DP in \cref{sec:sop-alg}, where we analyze the computational complexity and give a nonnegative realization.
To compare this approach with TN methods, \cref{sec:tn-separation} establishes an exponential separation from dense pairwise contraction of the specified representations using Reed-Muller code families.
Numerical comparisons for code-capacity and circuit-level decoding follow in \cref{sec:numerical-experiments}.
We then apply the resulting likelihoods to circuit noise learning, postselection, and decoder assessment in \cref{sec:applications}.
The paper concludes with a discussion of limitations and future directions in \cref{sec:discussion}.
The Supplemental Material contains additional proofs, numerical error analysis, a surface-code example, and experimental details.

\section{Preliminaries}
\label{sec:preliminaries}

In this section, we review stabilizer codes and the classification of Pauli errors by their syndromes and logical actions.
We then introduce the binary symplectic representation, which expresses these classes as affine spaces.
Finally, we recall rank-width and the decomposition convention used in the algorithms below.

\subsection{Stabilizer formalism}

Let $\calH_n=(\CC^2)^{\otimes n}$ be the $n$-qubit Hilbert space and $\sfP_n$ the group of $n$-qubit Pauli operators.
Each $f\in\sfP_n$ has the form $f=cf_1\otimes\cdots\otimes f_n$, where $f_i\in\{I,X,Y,Z\}$ are single-qubit Pauli operators and $c\in\{\pm1,\pm\mathrm{i}\}$ is an overall phase.
An $[[n,k]]$ quantum code is a $2^k$-dimensional subspace of $\calH_n$.
A \emph{stabilizer code} is specified by an Abelian subgroup $\sfS=\langle g_1,\ldots,g_{n-k}\rangle\subseteq\sfP_n$ with $n-k$ independent generators and $-I\notin\sfS$~\cite{gottesman_1997_stabilizer}.
Its code space is the simultaneous $+1$ eigenspace of all stabilizers
\begin{align}
    \calQ=\{\ket{\psi}\in\calH_n:g\ket{\psi}=\ket{\psi}\text{ for all }g\in\sfS\}.
    \label{eq:stabilizer-code}
\end{align}
The \emph{normalizer} of $\sfS$ in $\sfP_n$ consists of Pauli operators commuting with every stabilizer: $\sfN(\sfS) := \{f\in\sfP_n:fg=gf\text{ for all }g\in\sfS\}$.
Its elements preserve the code space $\calQ$.

Since stabilizers fix every code state, Pauli errors in the same coset $f\sfS := \{fg:g\in\sfS\}$ have identical action on $\calQ$.
The syndrome records the $+1$ or $-1$ measurement outcome of each stabilizer generator after an error acts on a code state.
Two errors $f,h\in\sfP_n$ have the same syndrome if they give the same measurement outcome for each stabilizer generator.
That means, the two errors differ by an operator that commutes with every stabilizer, i.e., $f^{\dagger}h\in\sfN(\sfS)$.
Thus, all errors with a fixed syndrome form a coset $t\sfN(\sfS)$, where $t$ is any error with that syndrome.
Choosing a representative $\ell \in \sfN(\sfS)$ for each coset in $\sfN(\sfS)/\sfS$, the coset $\ell \sfS$ labels a \emph{logical Pauli operator}.
Multiplying these cosets by $t$ gives the disjoint union
\begin{align}
    t\sfN(\sfS)=\bigsqcup_{\ell\sfS\in\sfN(\sfS)/\sfS}t\ell\sfS.
\end{align}
Consequently, any Pauli error $f\in\sfP_n$ can be decomposed as $f=t\ell g$, where $t$ represents its syndrome, $\ell$ represents its logical Pauli class relative to $t$, and $g\in\sfS$ is a stabilizer.

\subsection{Binary symplectic representation}

We now translate the operator description into binary linear algebra, so that commutation, syndromes, and logical classes can be expressed using vectors and matrices.

Ignoring scalar phases, represent a Pauli error $f$ by a binary column vector $\vb{e}=(\vb{x}^{\sfT},\vb{z}^{\sfT})^{\sfT}\in\FF_2^{2n}$ through $f\simeq E(\vb{e}):=\bigotimes_{i=1}^n X^{x_i}Z^{z_i}$, where $\simeq$ denotes equality up to phase.
The local pair $\vb{e}_i=(x_i,z_i)^{\sfT}$ represents $I,X,Z,Y$ for $(x_i,z_i)=(0,0),(1,0),(0,1),(1,1)$, respectively.
All binary vectors below are column vectors, and Pauli multiplication becomes vector addition.
For two errors with binary vectors $\vb{e},\vb{e}'$, define the symplectic product $\langle\vb{e},\vb{e}'\rangle_s:=\vb{e}^{\sfT}\vb{\Lambda}\vb{e}'$, where $\vb{\Lambda}:=\left(\begin{smallmatrix}\vb{0}&\vb{I}_n\\\vb{I}_n&\vb{0}\end{smallmatrix}\right)$.
The errors commute when this product is zero and anticommute when it is one.
Let $\vb{H}\in\FF_2^{(n-k)\times2n}$ have the transposed binary vectors of $g_1,\ldots,g_{n-k}$ as its rows.
Their independence and commutativity give $\rank\vb{H}=n-k$ and $\vb{H}\vb{\Lambda}\vb{H}^{\sfT}=\vb{0}$.
The stabilizer group $\sfS$ and its normalizer $\sfN(\sfS)$ are represented by the binary subspaces
\begin{align}
    S=\rowsp(\vb{H}),\qquad N=S^{\perp_s}=\ker(\vb{H}\vb{\Lambda}),
\end{align}
with $S\subseteq N$ and dimensions $n-k$ and $n+k$, respectively.
Here, $S^{\perp_s}$ denotes the orthogonal complement of $S$ with respect to the symplectic product.

The stabilizer coset $f\sfS$ is represented by the affine space $\vb{e}+S$.
Encode the measured eigenvalue of $g_a$ as $(-1)^{s_a}$; the syndrome vector is $\vb{s}=\vb{H}\vb{\Lambda}\vb{e}\in\FF_2^{n-k}$.
Two errors have the same syndrome exactly when their binary vectors differ by an element of $N$.
Let $\vb{r}_{\vb{s}}$ be the binary vector of the chosen representative $t \in \sfP_n$, so $\vb{H}\vb{\Lambda}\vb{r}_{\vb{s}}=\vb{s}$ and the syndrome class is $\vb{r}_{\vb{s}}+N$.
Choose the logical representatives with binary vectors in a linear complement $L$ satisfying $N=S\oplus L$ and $\dim L=2k$.
The coset partition becomes
\begin{align}
    \vb{r}_{\vb{s}}+N=\bigsqcup_{\vb*{\ell}\in L}\bigl(\vb{r}_{\vb{s}}+\vb*{\ell}+S\bigr).
\end{align}
Consequently, the operator decomposition $f=t\ell g$ becomes $\vb{e}=\vb{r}_{\vb{s}}+\vb*{\ell}+\vb{u}$, where $\vb{r}_{\vb{s}}$, $\vb*{\ell}\in L$, and $\vb{u}\in S$ represent $t$, $\ell$, and $g$, respectively.

\subsection{Rank decomposition and rank-width}
\label{subsec:rank-width}

We recall rank decomposition and rank-width~\cite{oum_2005_rankwidth,oum_seymour_2006_approximating}.
Let $G=(V,E)$ be a finite simple undirected graph with binary adjacency matrix $\vb{A}_G\in\FF_2^{|V|\times|V|}$.
For $Y\subseteq V$, write $\overline Y:=V\setminus Y$ and let $\vb{A}_G[Y,\overline Y]$ be the $|Y|\times|\overline Y|$ submatrix whose rows are indexed by $Y$ and columns by $\overline Y$.
The \emph{cut-rank function} of $G$ is
\begin{align}
    \rho_G(Y):=\rank_{\FF_2}\vb{A}_G[Y,\overline Y].
    \label{eq:prelim-cutrank}
\end{align}
This rank is the dimension of the space spanned by the binary neighborhood vectors across the cut.
Since $\vb{A}_G$ is symmetric, $\rho_G(Y)=\rho_G(\overline Y)$, so either side determines the same cut-rank.

A \emph{rank decomposition} of $G$ is a pair $(\calT,\mu)$, where $\calT$ is a rooted binary tree in which every internal node has exactly two children, and $\mu:\leaf(\calT)\to V$ is a bijection.
For each node $u$ of $\calT$, let $Y_u\subseteq V$ be the set of vertices assigned to the leaves of the subtree rooted at $u$.
The node $u$ thus displays the cut $Y_u\mid\overline{Y_u}$.
The width of the decomposition is the largest cut-rank displayed by its nodes:
\begin{align}
    w(\calT,\mu):=\max_{u\in V(\calT)}\rho_G(Y_u).
    \label{eq:rank-decomposition-width}
\end{align}
The \emph{rank-width} of $G$ is the minimum over all its rank decompositions,
\begin{align}
    \rw(G):=\min_{(\calT,\mu)}w(\calT,\mu).
    \label{eq:rank-width}
\end{align}
We set $\rw(G)=0$ when $G$ has at most one vertex.

We also use rank decompositions over a fixed partition $\calP$ of $V$ into nonempty parts.
Here, $\mu:\leaf(\calT)\to\calP$ is a bijection.
For each node $u$ of $\calT$, let $Y_u\subseteq V$ be the union of the parts assigned to the leaves of the subtree rooted at $u$.
The width is also given by \cref{eq:rank-decomposition-width}.
Thus, every displayed cut separates whole parts of $\calP$.
The rank-width of $G$ with respect to $\calP$, denoted by $\rw(G;\calP)$, is the minimum width over all such decompositions.
The singleton partition recovers $\rw(G)$.
This is the graph cut-rank version of grouping ground-set elements into parts in the partitioned-matroid and subspace-arrangement frameworks~\cite{jeong_2017_trellis}.

\section{Maximum likelihood decoding}
\label{sec:decoding}

Maximum-likelihood decoding selects the most probable logical label conditioned on an observed syndrome or detector-deviation vector.
We first specify the code-capacity and circuit-level noise models, then formulate their likelihood computations as a common partition-function evaluation problem.

\subsection{Noise model}
\label{subsec:noise-model}

\paragraph{Local stochastic noise.}
We consider stochastic Pauli channels with independent local factors.
For a single-qubit factor,
\begin{align}
    \calN_i(\rho)&=\sum_{f\in\{I,X,Y,Z\}}p_i(f)f\rho f,
    &p_i(f)&\geq0,\qquad\sum_f p_i(f)=1.
\end{align}
In particular, depolarizing noise has $p_i(I)=1-p$ and $p_i(X)=p_i(Y)=p_i(Z)=p/3$.
The Pauli noise channels may differ between qubits.
Writing a local error as $(x_i,z_i)$ gives the joint weight $p_i(x_i,z_i)$ and the product model $p(\vb{e})=\prod_i p_i(x_i,z_i)$.

More generally, let $d$ denote the total number of binary fault coordinates and write $\vb{e}\in\FF_2^d$ for a fault configuration.
The single-qubit Pauli model above has $d=2n$.
Let $\calB$ partition the coordinate index set $\{1,\ldots,d\}$ into disjoint local blocks.
For each block $B\in\calB$, write $\vb{e}_B\in\FF_2^{|B|}$ for the restriction of $\vb{e}$ to $B$ and assign a nonnegative joint weight $w_B:\FF_2^{|B|}\to\RR_{\geq0}$.
We assume independence between blocks, so a normalized probabilistic model has
\begin{align}
    p(\vb{e})&=\prod_{B\in\calB}w_B(\vb{e}_B),
    &\sum_{\vb{a}\in\FF_2^{|B|}}w_B(\vb{a})&=1\qquad(B\in\calB).
\end{align}
Within each block, $w_B$ is a joint distribution that need not factor over the individual coordinates.
For example, a single-qubit Pauli factor uses $\vb{e}_B=(x,z)^{\sfT}\in\FF_2^2$, whose four values represent $I,X,Z,Y$.
A stochastic two-qubit Pauli location uses $\vb{e}_B=(x_1,z_1,x_2,z_2)^{\sfT}\in\FF_2^4$, whose $2^4=16$ values represent the possible two-qubit Pauli errors.

\paragraph{Code-capacity noise.}
In the code-capacity model, Pauli noise acts on the data qubits, while syndrome extraction and recovery are noiseless.
Define $\vb{M}_s=\vb{H}\vb{\Lambda}$ and $\syn(\vb{e})=\vb{M}_s\vb{e}$, so the measured syndrome is $\vb{s}=\syn(\vb{e})$.
The joint likelihood of syndrome $\vb{s}$ and logical class $\vb*{\ell}\in L$ is
\begin{align}
    P^{\cc}(\vb{s},\vb*{\ell})=\sum_{\vb{u}\in S}p(\vb{r}_{\vb{s}}+\vb*{\ell}+\vb{u}).
    \label{eq:logical-partition}
\end{align}

\paragraph{Circuit-level noise.}
Circuit-level noise includes gate, reset, and measurement faults.
A \emph{detector} checks a measurement parity that is deterministic without faults~\cite{mcewen_2023_relaxing}.
A detector error model (DEM) assigns probabilities to independent error mechanisms and records their detector and logical-observable flips~\cite{stim}.
We now use $\vb{e}$ for the binary fault vector of the entire measurement circuit, rather than the data-qubit Pauli vector of the code-capacity model.
For an independent DEM, $e_j=1$ if error mechanism $j$ occurs, and $e_j=0$ otherwise.
Define the binary detector matrix $\vb{D}$ by $D_{aj}=1$ if fault coordinate $j$ flips detector $a$, and $D_{aj}=0$ otherwise.
Define the binary observable matrix $\vb{O}$ by $O_{bj}=1$ if fault coordinate $j$ flips measured logical observable $b$, and $O_{bj}=0$ otherwise.
Given a Clifford measurement circuit, specified detector and logical-observable parities, and an explicit local Pauli fault model, these matrices can be constructed classically in polynomial time by propagating each fault through the circuit~\cite{stim}.
Let $\vb{s}$ denote the detector deviations from noiseless values computed from the measurement record, and $\vb*{\ell}$ the logical flips to be inferred from $\vb{s}$.
Then, we have
\begin{align}
    \vb{s}=\vb{D}\vb{e},\qquad \vb*{\ell}=\vb{O}\vb{e}.
    \label{eq:circuit-maps}
\end{align}
When the ideal logical outcomes are known, $\vb*{\ell}$ can also be obtained from the measurement record for decoder evaluation.
Summing the local-block weights over fault configurations with these outputs gives the joint likelihood:
\begin{align}
    P^{\cl}(\vb{s},\vb*{\ell})=\sum_{\substack{\vb{e}\in\FF_2^d\\\vb{D}\vb{e}=\vb{s},\;\vb{O}\vb{e}=\vb*{\ell}}}
    \prod_{B\in\calB}w_B(\vb{e}_B).
    \label{eq:circuit-likelihood}
\end{align}
Singleton blocks give the independent DEM case; larger blocks retain local fault correlations.
The detector rows need not be independent.
Exactness is relative to the given stochastic fault model.

\subsection{Decoding via partition functions}
\label{subsec:decoding}

We formulate the two decoding rules and identify the partition-function evaluation problem underlying their likelihoods.

\paragraph{Code capacity.}
For an observed syndrome $\vb{s}$, the compatible logical classes are $C_{\vb*{\ell}}(\vb{s})=\vb{r}_{\vb{s}}+\vb*{\ell}+S$, with $\vb*{\ell}\in L$.
A Pauli recovery from $C_{\vb*{\ell}}(\vb{s})$ corrects the error up to a stabilizer if and only if the actual error belongs to that class.
Maximum-likelihood decoding therefore selects the class with the greatest total probability in \cref{eq:logical-partition}~\cite{bravyi_2014_efficient,cao_2026_maximum}:
\begin{align}
    \widehat{\vb*{\ell}}_{\mathrm{MLD}}
    \in\argmax_{\vb*{\ell}\in L}P^{\cc}(\vb{s},\vb*{\ell}).
    \label{eq:mld}
\end{align}

\paragraph{Circuit level.}
Let $q$ be the number of rows of $\vb{O}$, so $\vb*{\ell}\in\FF_2^q$ records the specified logical-observable flips.
Using the joint likelihood in \cref{eq:circuit-likelihood}, MLD selects
\begin{align}
    \widehat{\vb*{\ell}}_{\mathrm{MLD}}
    \in\argmax_{\vb*{\ell}\in\FF_2^q}P^{\cl}(\vb{s},\vb*{\ell}).
    \label{eq:circuit-mld}
\end{align}
This rule maximizes the probability of correctly identifying the entire flip vector specified by $\vb{O}$.
If $\vb{O}$ records only one final measured memory observable, its two labels need not specify the entire residual logical Pauli error.
For $\nu\in\{\cc,\cl\}$, the marginal probability of the observed vector is $P^\nu(\vb{s})=\sum_{\vb*{\ell}}P^\nu(\vb{s},\vb*{\ell})$, where the sum runs over $L$ for $\nu=\cc$ and $\FF_2^q$ for $\nu=\cl$.
For $P^\nu(\vb{s})>0$, the logical posterior is $\pi^\nu(\vb*{\ell}\mid\vb{s})=P^\nu(\vb{s},\vb*{\ell})/P^\nu(\vb{s})$, so the common denominator does not affect the maximizing label.

\paragraph{Decoding failure probability.}
For code-capacity decoding, let $\vb{s}$ have positive probability under the true noise model, and let a decoder $\delta$ select a logical label $\delta(\vb{s})\in L$.
Define
\begin{align}
    R_*(\vb{s})&=1-\max_{\vb*{\ell}\in L}\pi^{\cc}(\vb*{\ell}\mid\vb{s}),
    &R_\delta(\vb{s})&=1-\pi^{\cc}(\delta(\vb{s})\mid\vb{s}),\nonumber\\
    \Delta_\delta(\vb{s})&=R_\delta(\vb{s})-R_*(\vb{s})
        =\max_{\vb*{\ell}\in L}\pi^{\cc}(\vb*{\ell}\mid\vb{s})
         -\pi^{\cc}(\delta(\vb{s})\mid\vb{s}).
    \label{eq:decoder-regret}
\end{align}
Here, $R_\delta(\vb{s})$ is the conditional failure probability of $\delta$, and $R_*(\vb{s})$ is the minimum conditional failure probability, attained by MLD.
Thus, $\Delta_\delta(\vb{s})\geq0$ is the decoder's excess conditional failure probability relative to MLD.
Averaging $R_\delta$ and $\Delta_\delta$ over the true syndrome distribution gives the overall failure probability and excess failure probability relative to MLD, respectively.

\paragraph{Likelihood evaluation.}
We express the likelihood computations above as the following partition-function evaluation problem with binary linear constraints.
This partition-function formulation of maximum-likelihood decoding was considered by Chubb and Flammia for correlated Pauli noise~\cite{chubb_2021_statistical}.

\begin{definition}[Partition function evaluation]
\label{def:partition-function}
Given a binary matrix $\vb{M}\in\FF_2^{m\times d}$, a target $\vb{t}\in\FF_2^m$, a partition $\calB$ of $\{1,\ldots,d\}$, and finite nonnegative weights $w_B:\FF_2^{|B|}\to\RR_{\geq0}$ for each $B\in\calB$, compute the partition function
\begin{align}
    Z_{\vb{M}}(\vb{t};\vb{w})
    :=\sum_{\substack{\vb{e}\in\FF_2^d\\\vb{M}\vb{e}=\vb{t}}}
    \prod_{B\in\calB}w_B(\vb{e}_B),
    \label{eq:inference-primitive}
\end{align}
where $\vb{w}=(w_B)_{B\in\calB}$ and $\vb{e}_B$ is the restriction of $\vb{e}$ to $B$.
\end{definition}

When the weights depend on noise parameters $\theta$, we also write $Z_{\vb{M}}(\vb{t};\theta):=Z_{\vb{M}}(\vb{t};\vb{w}(\theta))$.
For code capacity, choose a full-row-rank matrix $\vb{P}\in\FF_2^{(n+k)\times2n}$ with $\ker\vb{P}=S$, and set $\vb{t}_{\vb*{\ell}}=\vb{P}(\vb{r}_{\vb{s}}+\vb*{\ell})$.
The constraint $\vb{P}\vb{e}=\vb{t}_{\vb*{\ell}}$ selects exactly $C_{\vb*{\ell}}(\vb{s})$, giving $P^{\cc}(\vb{s},\vb*{\ell})=Z_{\vb{P}}(\vb{t}_{\vb*{\ell}};\theta)$.
For the circuit model, take
\begin{align*}
    \vb{M}&=\begin{pmatrix}\vb{D}\\\vb{O}\end{pmatrix},
    &\vb{t}&=\begin{pmatrix}\vb{s}\\\vb*{\ell}\end{pmatrix},
    &P^{\cl}(\vb{s},\vb*{\ell})&=Z_{\vb{M}}(\vb{t};\theta).
\end{align*}

MLD can thus be implemented by evaluating the candidate likelihoods and selecting a maximizer.
Direct enumeration uses $4^k$ candidate classes for full-Pauli code-capacity decoding ($2^k$ for a binary component) and at most $2^q$ flip vectors for the circuit task.
A single likelihood value does not in general determine the maximizing label.
Approximate likelihoods suffice when their error bounds establish a maximizing label; otherwise, the largest estimate need not give an MLD decision (\cref{subsec:output-guarantees}).

\section{Rank DP for maximum-likelihood decoding}
\label{sec:sop-alg}

In this section, we reduce the partition-function evaluation problem in \cref{def:partition-function} to a quadratic sum-of-powers (SOP) for both code-capacity and circuit-level noise.
We then review the rank-decomposition dynamic program of de Colnet et al.~\cite{decolnet_2026_quadratic} in our notation, adapting its initialization to retain the local joint weights.
Together, the reduction and the dynamic program give a common arithmetic complexity bound for likelihood evaluation and, by comparing logical labels, maximum-likelihood decoding.

\subsection{Reduction to quadratic sums of powers}

For a simple graph $G$, a partition $\calC$ of $V(G)$, and local weights $g_C:\FF_2^C\to\CC$, define the quadratic SOP
\begin{align}
    \calS_{\calC}(G;g)
    :=\sum_{\vb{q}\in\FF_2^{V(G)}}
        (-1)^{\sum_{(u,v)\in E(G)}q_uq_v}
        \prod_{C\in\calC}g_C(\vb{q}_C).
    \label{eq:quadratic-sop}
\end{align}
For singleton blocks, we abbreviate $\calS_{\calC}(G;g)$ as $\calS(G;g)$ and write $g_v:=g_{\{v\}}$.
For $\vb{M}\in\FF_2^{m\times d}$, let $T(\vb{M})$ be the Tanner graph with check nodes $c_1,\ldots,c_m$, variable nodes $v_1,\ldots,v_d$, and edge $(c_a,v_j)$ exactly when $M_{aj}=1$.

\begin{lemma}[Partition-function reduction to quadratic SOPs]
\label{lem:partition-sop}
For every instance $\vb{M},\vb{t},\calB,\vb{w}$ of \cref{def:partition-function}, let $\calC = \calB \cup \{ \{c_1\}, \ldots, \{c_m\} \}$ be a partition of the $m+d$ vertices of $T(\vb{M})$ into the given coordinate blocks and singleton check nodes.
Define the local weights $g_B(\vb{e}_B):=w_B(\vb{e}_B)$ and $g_{c_a}(\lambda_a):=(-1)^{t_a\lambda_a}$.
Then,
\begin{align}
    Z_{\vb{M}}(\vb{t};\vb{w})
    =2^{-m}\calS_{\calC}\bigl(T(\vb{M});g\bigr).
    \label{eq:partition-sop}
\end{align}
\end{lemma}

\begin{proof}
For each $\vb{e}\in\FF_2^d$, the binary Fourier identity gives
\begin{align}
    \mathbf{1}[\vb{M}\vb{e}=\vb{t}]
    =2^{-m}\sum_{\vb*{\lambda}\in\FF_2^m}
        (-1)^{\vb*{\lambda}^{\sfT}(\vb{M}\vb{e}+\vb{t})}.
    \label{eq:character-identity}
\end{align}
Multiplying \cref{eq:character-identity} by the local weights and summing over $\vb{e}$ yields
\begin{align}
    Z_{\vb{M}}(\vb{t};\vb{w})
    =2^{-m}\sum_{\substack{\vb*{\lambda}\in\FF_2^m\\\vb{e}\in\FF_2^d}}
        (-1)^{\vb*{\lambda}^{\sfT}\vb{M}\vb{e}}
        \prod_{a=1}^m g_{c_a}(\lambda_a)
        \prod_{B\in\calB}w_B(\vb{e}_B).
    \label{eq:partition-fourier}
\end{align}
The bilinear exponent is the edge polynomial of $T(\vb{M})$, with variables $\lambda_a$ at $c_a$ and $e_j$ at $v_j$.
The remaining factors are exactly the local weights $g_C$, proving \cref{eq:partition-sop}.
Constructing the graph and assigning these weights takes polynomial time in the explicit input size.
\end{proof}

Below, we give the SOP representation of the likelihood for both code-capacity and circuit-level noise.

\paragraph{Code-capacity noise.}
Take $\vb{M}=\vb{P}$ and $\vb{t}=\vb{t}_{\vb*{\ell}}$ from \cref{subsec:decoding}.
Applying \cref{lem:partition-sop} gives
\begin{align}
    P^{\cc}(\vb{s},\vb*{\ell})
    =2^{-(n+k)}\calS_{\calC}\bigl(T(\vb{P});g^{(\vb{s},\vb*{\ell})}\bigr).
    \label{eq:capacity-scalar-sop}
\end{align}
For independent single-qubit Pauli noise, let $B_i:=\{i,n+i\}$ be the coordinate block of qubit $i$.
The weights in \cref{eq:capacity-scalar-sop} are
\begin{align*}
    g_{B_i}^{(\vb{s},\vb*{\ell})}(x_i,z_i)
    &=p_i(x_i,z_i),\qquad 1\leq i\leq n,\\
    g_{c_a}^{(\vb{s},\vb*{\ell})}(\lambda_a)
    &=(-1)^{[\vb{P}(\vb{r}_{\vb{s}}+\vb*{\ell})]_a\lambda_a},
        \qquad 1\leq a\leq n+k.
\end{align*}
The graph $T(\vb{P})$ has $3n+k$ vertices.

\paragraph{Circuit-level noise.}
Take the stacked detector/observable matrix and target from \cref{subsec:decoding}.
If $\vb{D}$ has $m_{\mathrm{det}}$ rows and $\vb{O}$ has $q$ rows, then $m=m_{\mathrm{det}}+q$ and \cref{lem:partition-sop} gives
\begin{align}
    P^{\cl}(\vb{s},\vb*{\ell})
    =2^{-m}\calS_{\calC}\bigl(T(\vb{M});g^{(\vb{s},\vb*{\ell})}\bigr).
    \label{eq:circuit-sop}
\end{align}
The graph $T(\vb{M})$ has $m+d$ vertices.
For an independent DEM, all blocks are singletons, and the check weights are $g_{c_a}^{(\vb{s},\vb*{\ell})}(\lambda_a)=(-1)^{t_a\lambda_a}$ for $1\leq a\leq m$, where $\vb{t}=(\vb{s}^{\sfT},\vb*{\ell}^{\sfT})^{\sfT}$.
The coordinate weights are
\begin{align}
    g_{v_j}^{(\vb{s},\vb*{\ell})}(e_j)=
        \begin{cases}
            1-p_j(\theta),&e_j=0,\\
            p_j(\theta),&e_j=1,
        \end{cases}
        \qquad 1\leq j\leq d,
    \label{eq:dem-sop-weights}
\end{align}
where $d$ is the number of independent error mechanisms in the DEM, one per column of $\vb{M}$, and $p_j(\theta)$ is the occurrence probability of mechanism $j$ for noise parameters $\theta$.
Thus, the singleton specialization is
\begin{align}
    P^{\cl}(\vb{s},\vb*{\ell})
    =2^{-m}\calS\bigl(T(\vb{M});g^{(\vb{s},\vb*{\ell})}\bigr).
    \label{eq:dem-scalar-sop}
\end{align}

\subsection{Rank-decomposition dynamic programming}
\label{sec:dp-sop}

De Colnet et al.~\cite{decolnet_2026_quadratic} give a rank-decomposition dynamic programming (Rank DP) algorithm to evaluate a quadratic SOP and proved that its complexity is exponential on the rank-width of the underlying graph.
Here, we adapt their algorithm to the partition function evaluation problem in \cref{def:partition-function}.

Recall that we are given a graph $G$, a partition $\calC$ of its vertices, and local weights $g_C:\FF_2^C\to\CC$, which define a quadratic SOP $\calS_{\calC}(G;g)$ in \cref{eq:quadratic-sop}.
Let $\mathcal T$ be a rooted binary tree whose leaves are the blocks $C\in\calC$, and let $Y_u\subseteq V(G)$ be the union of the blocks below a node $u$.
For $Y\subseteq V(G)$, write $\overline Y:=V(G)\setminus Y$ and let $\vb{A}_G$ be the binary adjacency matrix of $G$.
The boundary signature and cut-rank are
\begin{align}
    \vb*{\sigma}_Y(\vb{q}_Y)
        &:=\vb{A}_G[Y,\overline Y]^{\sfT}\vb{q}_Y,
    \label{eq:sop-signature}\\
    \rho_G(Y)&:=\rank_{\FF_2}\vb{A}_G[Y,\overline Y].
    \label{eq:cutrank} 
\end{align}
The signature determines the quadratic sign contributed by all edges crossing the cut, so configurations with the same signature can be summed.
There are $2^{\rho_G(Y)}$ attainable signatures, and we set $\omega:=\max_u\rho_G(Y_u)$.
For singleton blocks this is the width of a rooted rank-decomposition; minimizing over such decompositions gives the usual rank-width~\cite{oum_2005_rankwidth}.
For general blocks, the given tree keeps each joint factor intact.

Let $\phi_Y(\vb{q}_Y):=\sum_{(v,v')\in E(G[Y])}q_vq_{v'}\in\FF_2$.
The table at $u$ is the partial SOP
\begin{align}
    F_u(\vb*{\sigma})
    :=\sum_{\substack{\vb{q}_{Y_u}\in\FF_2^{Y_u}\\
                       \vb*{\sigma}_{Y_u}(\vb{q}_{Y_u})=\vb*{\sigma}}}
        (-1)^{\phi_{Y_u}(\vb{q}_{Y_u})}
        \prod_{\substack{C\in\calC\\C\subseteq Y_u}}g_C(\vb{q}_C).
    \label{eq:sop-dp-table}
\end{align}
As we traverse the tree in postorder (children before parents), we compute the table at each internal node from the tables at its children.
For a node $u$ with children $u_1,u_2$, representatives $\vb{a},\vb{b}$ of signatures $\vb*{\alpha},\vb*{\beta}$ determine the parent signature and crossing parity by
\begin{align}
    \Gamma_u(\vb*{\alpha},\vb*{\beta})
        &:=\vb*{\alpha}|_{\overline{Y_u}}+
            \vb*{\beta}|_{\overline{Y_u}},
    \label{eq:sop-dp-parent-signature}\\*
    \chi_u(\vb*{\alpha},\vb*{\beta})
        &:=\vb{a}^{\sfT}\vb{A}_G[Y_{u_1},Y_{u_2}]\vb{b}.
    \label{eq:sop-dp-maps}
\end{align}
As proved in \cite{decolnet_2026_quadratic}, $\chi_u$ is independent of the chosen representatives because each child signature fixes its parity contribution to the other child.
With these notations, \cref{alg:sop-dp} presents the dynamic program for evaluating $\calS_{\calC}(G;g)$.

\begin{algorithm}[htbp]
\caption{Block-weighted quadratic SOP evaluation, adapted from~\cite{decolnet_2026_quadratic}.}
\label{alg:sop-dp}
\begin{algorithmic}[1]
\Require $G$, $\calC$, weights $g_C$, and a rooted block tree $\mathcal T$ with root $u_\star$
\Ensure $\calS_{\calC}(G;g)$
\State Precompute signature coordinates, leaf parities, and the maps $\Gamma_u,\chi_u$
\For{each leaf $u$ corresponding to a block $C$}
    \State Initialize $F_u(\vb*{\sigma})\gets0$ for every attainable signature
    \For{each $\vb{q}_C\in\FF_2^C$}
        \State $\vb*{\sigma}\gets\vb*{\sigma}_C(\vb{q}_C)$
        \State $F_u(\vb*{\sigma})\mathrel{+}=(-1)^{\phi_C(\vb{q}_C)}g_C(\vb{q}_C)$
    \EndFor
\EndFor
\For{each internal node $u$ in postorder, with children $u_1,u_2$}
    \State Initialize $F_u(\vb*{\sigma})\gets0$ for every attainable signature
    \For{each pair of attainable child signatures $(\vb*{\alpha},\vb*{\beta})$}
        \State $\vb*{\sigma}\gets\Gamma_u(\vb*{\alpha},\vb*{\beta})$
        \State $F_u(\vb*{\sigma})\mathrel{+}=(-1)^{\chi_u(\vb*{\alpha},\vb*{\beta})}F_{u_1}(\vb*{\alpha})F_{u_2}(\vb*{\beta})$
    \EndFor
\EndFor
\State \Return $F_{u_\star}(\varnothing)$
\end{algorithmic}
\end{algorithm}

\subsection{Application to decoding and complexity}
\label{subsec:sop-complexity}

We now bound the evaluation cost in terms of the rank-width of the Tanner graph with respect to its block partition.
The following consequence of the Oum-Seymour approximation algorithm includes the construction of a decomposition that respects this partition.

\begin{proposition}[Constructing a rank decomposition over a partition]
\label{prop:partition-rank-decomposition}
Let $G$ be a graph with $n\geq1$ vertices, and let $\calP$ be a partition of $V(G)$ into nonempty parts.
Set $r:=\rw(G;\calP)$ and $c:=\max_{P\in\calP}\rho_G(P)$.
A rank decomposition of $G$ over $\calP$ of width at most $3r+c\leq4r$ can be constructed deterministically in time $\poly(n)2^{O(r)}$.
\end{proposition}

The proof is given in \cref{proof:partition-rank-decomposition}.
Combining \cref{prop:partition-rank-decomposition} with the dynamic program gives the following complexity bound, including decomposition construction and structural preprocessing.
It applies to both code-capacity and circuit-level likelihoods through the choices of $\vb{M}$ and $\vb{t}$ in \cref{subsec:decoding}.

\begin{theorem}[Partition-function evaluation complexity]
\label{thm:sop-dp}
Given $\vb{M}\in\FF_2^{m\times d}$, $\vb{t}\in\FF_2^m$, a partition $\calB$ of $\{1,\ldots,d\}$, and finite nonnegative local weight tables $\vb{w}=(w_B)_{B\in\calB}$, consider the partition function $Z_{\vb{M}}(\vb{t};\vb{w})$ in \cref{def:partition-function}.
Set $\calC:=\calB\cup\{\{c_1\},\ldots,\{c_m\}\}$, $r:=\rw(T(\vb{M});\calC)$, and $N:=m+d+\sum_{B\in\calB}2^{|B|}$.
Then, $Z_{\vb{M}}(\vb{t};\vb{w})$ can be computed exactly using \cref{alg:sop-dp} in $\poly(N)2^{O(r)}$ time, including decomposition construction and structural preprocessing.
The numeric DP tables contain $O((m+|\calB|)2^{O(r)})$ entries in total, excluding input weights, structural data, and workspace.
\end{theorem}

\begin{proof}
\cref{prop:partition-rank-decomposition} constructs a rank-decomposition tree $\calT$ of $T(\vb{M})$ over $\calC$ of width $\omega\leq4r$ in time $\poly(N)2^{O(r)}$.
By \cref{lem:partition-sop}, $Z_{\vb{M}}(\vb{t};\vb{w})=2^{-m}\calS_{\calC}(T(\vb{M});g)$, where $g_B=w_B$ and $g_{c_a}(\lambda_a)=(-1)^{t_a\lambda_a}$.
Leaf initialization gives \cref{eq:sop-dp-table}.
At an internal node, each block below it belongs to one child subtree, so \cref{eq:sop-dp-parent-signature,eq:sop-dp-maps} allow us to group \cref{eq:sop-dp-table} by attainable child signatures:
\begin{align}
    F_u(\vb*{\sigma})
    =\sum_{\substack{\vb*{\alpha},\vb*{\beta}\\
        \Gamma_u(\vb*{\alpha},\vb*{\beta})=\vb*{\sigma}}}
        (-1)^{\chi_u(\vb*{\alpha},\vb*{\beta})}
        F_{u_1}(\vb*{\alpha})F_{u_2}(\vb*{\beta}).
    \label{eq:sop-dp-recurrence}
\end{align}
This is the update in \cref{alg:sop-dp}, so induction gives $F_{u_\star}(\varnothing)=\calS_{\calC}(T(\vb{M});g)$.
Multiplication by $2^{-m}$ then gives $Z_{\vb{M}}(\vb{t};\vb{w})$.

For a graph $G$ with block partition $\calC$ and a given tree of width $\omega$, each leaf enumerates $2^{|C|}$ configurations and each internal node considers $2^{\rho_G(Y_{u_1})+\rho_G(Y_{u_2})}$ pairs of child signatures.
Hence, the dynamic program uses
\begin{align}
    O\!\left(\sum_{C\in\calC}2^{|C|}
        +\sum_{u\text{ internal}}
            2^{\rho_G(Y_{u_1})+\rho_G(Y_{u_2})}\right)
    \subseteq O\!\left(\sum_{C\in\calC}2^{|C|}+|\calC|4^\omega\right)
    \label{eq:sop-dp-cost}
\end{align}
scalar operations.
Here, the tree has $|\calC|$ leaves and hence $|\calC|-1$ internal nodes.
Each summand over internal nodes is at most $2^{2\omega}=4^\omega$, which gives the second bound.
For $G=T(\vb{M})$, we have $|\calC|=m+|\calB|$ and $\sum_{C\in\calC}2^{|C|}=2m+\sum_{B\in\calB}2^{|B|}$.
Substituting these identities into \cref{eq:sop-dp-cost} and including the final normalization gives
\begin{align}
    O\!\left(\sum_{B\in\calB}2^{|B|}
        +(m+|\calB|)4^\omega\right)
    \label{eq:decoding-dp-given-tree}
\end{align}
scalar operations on any given rank-decomposition tree of width $\omega$, after structural preprocessing.

To construct these structural data, we first compute the signature bases and compact descriptions of the maps in \cref{alg:sop-dp} by polynomial-time binary linear algebra.
We then enumerate the leaf configurations and child-signature pairs, giving a total preprocessing time of at most $\poly(N)(\sum_{B\in\calB}2^{|B|}+4^\omega)$.
Since $\sum_{B\in\calB}2^{|B|}\leq N$ and $\omega\leq4r$, combining this cost with decomposition construction and the evaluation cost in \cref{eq:decoding-dp-given-tree} gives the claimed total time $\poly(N)2^{O(r)}$.
Finally, each of the $2(m+|\calB|)-1$ nodes stores at most $2^\omega$ numeric entries, so the numeric DP tables contain $O((m+|\calB|)2^{O(r)})$ entries in total.
\end{proof}

Once a rank-decomposition tree of width $\omega$ and its structural maps have been constructed, \cref{eq:decoding-dp-given-tree} bounds the arithmetic cost of each likelihood evaluation.
For blocks of bounded size, $\sum_{B\in\calB}2^{|B|}=O(|\calB|)$, so each likelihood evaluation uses $O((m+|\calB|)4^\omega)$ arithmetic operations.

For fixed $\vb{M}$, $\calB$, and $\mathcal T$, the graph and precomputed maps are independent of $\vb{t}$ and $\vb{w}$ and can be reused across likelihood evaluations.
After this one-time preprocessing, evaluating $K$ candidate logical labels costs at most $K$ times the bound in \cref{eq:decoding-dp-given-tree}, followed by $K-1$ comparisons to select an MLD label.
The candidate counts for the two noise models are given in \cref{subsec:decoding}.

\subsection{Reducing the Tanner graph by Gaussian elimination}
\label{subsec:systematic-sop}

We now use Gaussian elimination to evaluate the same partition function on a reduced Tanner graph with only $d$ vertices.
Choose a column basis of $\vb{M} \in \FF_2^{m\times d}$ with index set $B_0$, and let $D_0:=\{1,\ldots,d\}\setminus B_0$.
Then, columns of $\vb{M}_{D_0}$ are linear combinations of the columns of $\vb{M}_{B_0}$.
This allows us to uniquely define a matrix $\vb{A}\in\FF_2^{B_0\times D_0}$ such that $\vb{M}_{D_0}=\vb{M}_{B_0}\vb{A}$, where columns of $\vb{A}$ are the coefficients of the linear combinations.
The Tanner graph $T(\vb{A})$ has bipartition $(B_0,D_0)$ and is also known as the fundamental graph of the binary matroid represented by $\vb{M}$ with respect to the basis $B_0$~\cite{oum_2005_rankwidth}.

The corollary below shows that the partition function can be evaluated on the reduced Tanner graph $T(\vb{A})$.

\begin{corollary}[Check elimination]
\label{cor:systematic-sop}
Under the assumptions of \cref{thm:sop-dp}, suppose that $\vb{M}\vb{e}=\vb{t}$ has at least one solution.
Let $\calT$ be a given rank-decomposition tree of $T(\vb{M})$ over $\calC$ of width $\omega$.
Deleting the check leaves from $\calT$ and suppressing nodes with only one child yields a rank-decomposition tree of $T(\vb{A})$ over $\calB$ with width $\omega_A\leq\omega$.
On this tree, \cref{alg:sop-dp} evaluates $Z_{\vb{M}}(\vb{t};\vb{w})$ with time complexity $O\!\left(\sum_{B\in\calB}2^{|B|}+|\calB|4^{\omega_A}\right)$ after local weight transformation and structural preprocessing.
\end{corollary}

\begin{proof}
Recall that $Z_{\vb{M}}(\vb{t};\vb{w})=\sum_{\vb{e}\in\FF_2^d} \mathbf{1}[\vb{M}\vb{e}=\vb{t}] \prod_{B\in\calB}w_B(\vb{e}_B)$.
Let $r:=|B_0|=\rank_{\FF_2}\vb{M}$.
Gaussian elimination gives an invertible matrix $\vb{U}\in\FF_2^{m\times m}$ such that
\begin{align}
    \vb{U}\begin{pmatrix}\vb{M}_{B_0}&\vb{M}_{D_0}\end{pmatrix}
        &=\begin{pmatrix}\vb{I}_r&\vb{A}\\\vb{0}&\vb{0}\end{pmatrix},
    &\vb{U}\vb{t}&=\begin{pmatrix}\vb*{\tau}\\\vb{0}\end{pmatrix}.
    \label{eq:systematic-row-reduction}
\end{align}
Here, $\vb*{\tau}\in\FF_2^{B_0}$, and the last $m-r$ entries of $\vb{U}\vb{t}$ vanish because $\vb{M} \vb{e} = \vb{t}$ has at least one solution.
Since $\vb{U}$ is invertible,
\begin{align}
    \vb{M}\vb{e}=\vb{t}
    \quad\Longleftrightarrow\quad
    \vb{U}\vb{M}\vb{e}=\vb{U}\vb{t}
    \quad\Longleftrightarrow\quad
    \vb{e}_{B_0}+\vb{A}\vb{e}_{D_0}=\vb*{\tau}.
    \label{eq:systematic-form}
\end{align}
Using \cref{eq:systematic-form} and the binary Fourier identity in \cref{eq:character-identity}, we obtain
\begin{align}
    Z_{\vb{M}}(\vb{t};\vb{w})
    &=\sum_{\vb{e}\in\FF_2^d}
        \mathbf{1}[\vb{e}_{B_0}+\vb{A}\vb{e}_{D_0}=\vb*{\tau}]
        \prod_{B\in\calB}w_B(\vb{e}_B)
        \label{eq:systematic-reduced-indicator}\\
    &=2^{-r}\sum_{\substack{\vb{q}_{B_0}\in\FF_2^{B_0}\\\vb{e}\in\FF_2^d}}
        (-1)^{\vb{q}_{B_0}^{\sfT}
            (\vb{e}_{B_0}+\vb{A}\vb{e}_{D_0}+\vb*{\tau})}
        \prod_{B\in\calB}w_B(\vb{e}_B).
    \label{eq:systematic-fourier-expansion}
\end{align}
Set $\vb{q}_{D_0}:=\vb{e}_{D_0}$ and separate the phase term $\vb{q}_{B_0}^{\sfT}\vb{A}\vb{q}_{D_0}$, which is independent of $\vb{e}_{B_0}$.
Then,
\begin{equation}
    \begin{aligned}
        Z_{\vb{M}}(\vb{t};\vb{w})
        &=2^{-r}\sum_{\vb{q}\in\FF_2^d}
            (-1)^{\vb{q}_{B_0}^{\sfT}\vb{A}\vb{q}_{D_0}}
            \sum_{\vb{e}_{B_0}\in\FF_2^{B_0}}
                (-1)^{\vb{q}_{B_0}^{\sfT}(\vb{e}_{B_0}+\vb*{\tau})}\\
        &\qquad\times\prod_{B\in\calB}
            w_B(\vb{e}_{B\cap B_0},\vb{q}_{B\cap D_0}).
    \end{aligned}
    \label{eq:systematic-partial-sum}
\end{equation}
Since the blocks partition the error coordinates, the inner sum factors over $B\in\calB$.
Renaming the original summation variable $\vb{e}_{B\cap B_0}$ as $\vb{a}$ within each block defines the transformed weight
\begin{align}
    g_B(\vb{q}_B)
    :=\sum_{\vb{a}\in\FF_2^{B\cap B_0}}
        (-1)^{\vb{q}_{B\cap B_0}^{\sfT}
            (\vb{a}+\vb*{\tau}_{B\cap B_0})}
        w_B(\vb{a},\vb{q}_{B\cap D_0}).
    \label{eq:systematic-sop-weights}
\end{align}
Thus, each $g_B$ is a partial Walsh--Hadamard transform of $w_B$, with the target-dependent sign included.
Substituting these local sums into the partition function yields
\begin{align}
    Z_{\vb{M}}(\vb{t};\vb{w})
    &=2^{-r}\sum_{\vb{q}\in\FF_2^d}
        (-1)^{\vb{q}_{B_0}^{\sfT}\vb{A}\vb{q}_{D_0}}
        \prod_{B\in\calB}g_B(\vb{q}_B)
        \label{eq:systematic-weighted-sop}\\
    &=2^{-r}\calS_{\calB}(T(\vb{A});g),
    \label{eq:block-fundamental-sop}
\end{align}
because $\vb{q}_{B_0}^{\sfT}\vb{A}\vb{q}_{D_0}$ is the edge polynomial of $T(\vb{A})$.

Each cut $Y$ of the reduced tree comes from a cut $X=Y\sqcup S$ of $\calT$, where $Y$ and $S$ consist of variable and check nodes, respectively.
The complements $\overline{Y}$ and $\overline{S}$ are taken within the variable and check node sets, respectively.
By the definition of cut-rank,
\begin{align}
    \rho_{T(\vb{A})}(Y)
    &=\rank\vb{A}[B_0\cap\overline{Y},D_0\cap Y]
      +\rank\vb{A}[B_0\cap Y,D_0\cap\overline{Y}]\nonumber\\
    &=\bigl(\rank\vb{M}_Y-|B_0\cap Y|\bigr)
      +\bigl(\rank\vb{M}_{\overline{Y}}-|B_0\cap\overline{Y}|\bigr)\nonumber\\
    &=\rank\vb{M}_Y+\rank\vb{M}_{\overline{Y}}-r,
    \label{eq:block-fundamental-rank}
\end{align}
where the second equality follows from \cref{eq:systematic-row-reduction} and the last uses $|B_0|=r$.
Select $r$ independent rows of $\vb{M}$, of which $s$ are indexed by $S$.
Deleting the other rows preserves all column-submatrix ranks, so $\rank\vb{M}_Y\leq s+\rank\vb{M}[\overline{S},Y]$ and $\rank\vb{M}_{\overline{Y}}\leq r-s+\rank\vb{M}[S,\overline{Y}]$.
Combining these inequalities with \cref{eq:block-fundamental-rank} gives
\begin{align}
    \rho_{T(\vb{A})}(Y)
    \leq\rank\vb{M}[\overline{S},Y]+\rank\vb{M}[S,\overline{Y}]
        =\rho_{T(\vb{M})}(X).
    \label{eq:systematic-cut-bound}
\end{align}
Taking maxima over cuts of the reduced tree gives $\omega_A\leq\omega$; in particular, $\rw(T(\vb{A}))\leq\rw(T(\vb{M}))$.
Applying \cref{eq:sop-dp-cost} with $\calC=\calB$ gives the claimed complexity bound.
The blockwise Walsh--Hadamard transforms cost an additional $O\!\left(\sum_{B\in\calB}(1+|B\cap B_0|)2^{|B|}\right)$ operations.
\end{proof}

A surface-code example illustrating the reduced Tanner graph and Rank DP is given in Supplemental \cref{app:worked-example}.

\subsection{Nonnegative realization}
\label{subsec:nonnegative-accuracy}

The transformed weights in \cref{eq:systematic-sop-weights} can have either sign, so cancellation in the Rank DP update \cref{eq:sop-dp-recurrence} can cause a loss of relative accuracy in floating-point arithmetic.
Here, we give a Rank DP algorithm that evaluates the partition function using only nonnegative sums and products.
The construction uses the state spaces of linear-code tree realizations~\cite{forney_2001_normal,kashyap_2009_tree} and the sum-product algorithm on the resulting tree~\cite{kschischang_2001_sumproduct}.
Its table sizes are also determined by the cut-ranks of the reduced Tanner graph.

Let $G:=T(\vb{A})$, and fix a rank-decomposition tree $\calT$ of $G$ over $\calB$ with width $\omega_A$.
Write $N:=|\calB|\geq1$.
We again suppose that the linear system $\vb{M} \vb{e} = \vb{t}$ has a solution; otherwise, the partition function will be zero.
Let $\vb{e}_0$ be a solution and write $\vb{e}=\vb{e}_0+\vb{x}$, where $\vb{x}\in\ker\vb{M}$.
Define $\widetilde{w}_B(\vb{x}_B):=w_B(\vb{e}_{0,B}+\vb{x}_B)$.
This shift permutes the outcomes within each block and preserves their nonnegative weights.

We now define the DP tables and the update rule for a nonnegative implementation.
For a node $u$ of $\calT$, let $Y_u\subseteq[d]$ be the union of the blocks assigned to the leaves of the subtree rooted at $u$.
Let $\overline{Y_u}:=[d]\setminus Y_u$ and define the state space
\begin{align}
    I_{Y_u}:=\im\vb{M}_{Y_u}\cap\im\vb{M}_{\overline{Y_u}}.
    \label{eq:block-interface}
\end{align}
Here, $\vb{M}_{Y_u}$ contains the columns indexed by $Y_u$.
From \cref{eq:block-fundamental-rank}, we have $\dim I_{Y_u}=\rank\vb{M}_{Y_u}+\rank\vb{M}_{\overline{Y_u}}-\rank\vb{M}=\rho_G(Y_u)$.
For an assignment $\vb{x}_{Y_u}$ to the variables in $Y_u$, let $\vb{h}:=\vb{M}_{Y_u}\vb{x}_{Y_u}$ be its partial syndrome.
To satisfy $\vb{M}\vb{x}=\vb{0}$, the remaining variables must satisfy $\vb{M}_{\overline{Y_u}}\vb{x}_{\overline{Y_u}}=\vb{h}$.
Since $\vb{h}\in\im\vb{M}_{Y_u}$ already, such an assignment exists exactly when $\vb{h}\in I_{Y_u}$.
At an internal node $u$, define the DP table by
\begin{align}
    F_{Y_u}(\vb{h})
        :=\sum_{\substack{\vb{x}_{Y_u}\in\FF_2^{Y_u}\\
                          \vb{M}_{Y_u}\vb{x}_{Y_u}=\vb{h}}}
            \prod_{\substack{B\in\calB\\B\subseteq Y_u}}
                \widetilde{w}_B(\vb{x}_B),
    \qquad \vb{h}\in I_{Y_u}.
    \label{eq:block-table}
\end{align}
This table has $2^{\rho_G(Y_u)}$ entries.

At a leaf $u$ with $Y_u=B$, the product in \cref{eq:block-table} consists of the single weight $\widetilde{w}_B(\vb{x}_B)$.
Thus, we initialize the leaf table by enumerating the local assignments and summing their weights for each partial syndrome in $I_B$.

Finally, at an internal node $u$ with children $u_1,u_2$, we have $Y_u=Y_{u_1}\sqcup Y_{u_2}$.
Once the two child tables have been computed, combine them using the update rule
\begin{align}
    F_{Y_u}(\vb{h})
        =\sum_{\substack{\vb{h}_1\in I_{Y_{u_1}},\,
                          \vb{h}_2\in I_{Y_{u_2}}\\
                          \vb{h}_1+\vb{h}_2=\vb{h}}}
            F_{Y_{u_1}}(\vb{h}_1)F_{Y_{u_2}}(\vb{h}_2),
    \qquad \vb{h}\in I_{Y_u}.
    \label{eq:block-join}
\end{align}

The following proposition shows that this dynamic program evaluates the same partition function as \cref{cor:systematic-sop}, with the same arithmetic complexity bound.

\begin{proposition}[Nonnegative realization]
\label{prop:nonnegative-realization}
Suppose that $\vb{M}\vb{e}=\vb{t}$ has a solution.
On the given rank-decomposition tree $\calT$, the leaf initialization and \cref{eq:block-join} compute the tables in \cref{eq:block-table} using only nonnegative sums and products.
Each entry at an internal node $u$ sums exactly $2^{\mu_{Y_u}}$ products, where $\mu_{Y_u}:=\dim(I_{Y_{u_1}}\cap I_{Y_{u_2}})$.
The root value is $F_{[d]}(\vb{0})=Z_{\vb{M}}(\vb{t};\vb{w})$.
The arithmetic complexity is $O\!\left(\sum_{B\in\calB}2^{|B|}+N4^{\omega_A}\right)$.
The numeric tables contain at most $(2N-1)2^{\omega_A}$ entries.
These counts exclude structural preprocessing, input storage, indices, and workspace.
\end{proposition}

The proof is given in Supplemental \cref{proof:block-inference}.
For the shifted instance $\vb{M}\vb{x}=\vb{0}$, the two dynamic programs are related by an invertible Fourier change of basis between their tables at each node of $\calT$; see Supplemental \cref{proof:sop-positive-basis}.
Thus, the nonnegative realization preserves the information and table sizes of the original dynamic program.
The nonnegative DP uses the original shifted weights and requires no numerical Fourier transform.
Fourier duality of code partition functions on normal factor graphs provides the broader context for this correspondence~\cite{forney_2011_duality}.

The nonnegative realization avoids cancellation from signed intermediate values and admits a relative roundoff-error bound under explicit floating-point assumptions.
Supplemental \cref{app:numerical-accuracy} gives the analysis, its consequences for posterior probabilities and decoding decisions, and the implementation details.

\section{Exponential separation from tensor-network contraction}
\label{sec:tn-separation}

Tensor-network methods are a state-of-the-art approach to MLD~\cite{bravyi_2014_efficient,piveteau_2024_tensor,shutty_2024_ensembling}.
For a fixed tensor-network graph, pairwise contraction with dense intermediate tensors has a cost lower bound exponential in its treewidth~\cite{markov_2008_simulating}.
We first compare this lower bound with the arithmetic complexity of Rank DP, which is governed by rank-width.
We then give two code families for which Rank DP evaluates likelihoods in polynomial time after check elimination, while dense pairwise contraction of the corresponding tensor networks requires superpolynomial time.

\subsection{Rank-width versus treewidth}
\label{subsec:rank-width-vs-treewidth}

We first recall the tensor-network representation of the partition function in the detector picture, including the independent-DEM case~\cite{piveteau_2024_tensor,shutty_2024_ensembling}.

\begin{proposition}[Tensor-network for MLD]
\label{prop:decoding-tensor-network}
Let $\vb{M}\in\FF_2^{m\times d}$, $\vb{t}\in\FF_2^m$, and let $\calB$ be a partition of $[d]$, with local weights $w_B:\FF_2^B\to\RR_{\geq0}$.
The partition function $Z_{\vb{M}}(\vb{t};\vb{w})=\sum_{\vb{e}\in\FF_2^d:\,\vb{M}\vb{e}=\vb{t}}\prod_{B\in\calB}w_B(\vb{e}_B)$ has a binary tensor-network representation with a variable node $v_j$ carrying a COPY tensor for each coordinate, a check node $c_a$ carrying a parity tensor for each constraint, and a weight node $u_B$ for each block.
Its underlying graph $G$ has edge set
\begin{align}
    E(G)=\{(c_a,v_j):M_{aj}=1\}
        \cup\{(u_B,v_j):B\in\calB,\ j\in B\}.
    \label{eq:decoding-tn-graph}
\end{align}
\end{proposition}

\begin{proof}
Use the standard COPY and parity tensors~\cite{piveteau_2024_tensor}, with components
\begin{align}
    C_r(\vb{x})&:=\mathbf{1}[x_1=\cdots=x_r],
    &P_{r,t}(\vb{x})&:=\mathbf{1}\!\left[\sum_{i=1}^r x_i=t\right],
\end{align}
where $\vb{x}\in\FF_2^r$ and the sum is in $\FF_2$.
On $G$, place $C_{\deg_G(v_j)}$ at $v_j$, $P_{\deg_G(c_a),t_a}$ at $c_a$, and the weight tensor $W_B(\vb{e}_B):=w_B(\vb{e}_B)$ at $u_B$.
Contraction sums the product of all tensor entries over binary assignments to the edges.
The assignments satisfying all COPY constraints correspond bijectively to $\vb{e}\in\FF_2^d$, with every edge incident to $v_j$ assigned $e_j$.
For such an assignment, each COPY tensor equals $1$, and the parity tensor at $c_a$ equals $\mathbf{1}[\sum_{j:M_{aj}=1}e_j=t_a]=\mathbf{1}[(\vb{M}\vb{e})_a=t_a]$.
Hence, the contraction equals
\begin{align}
    \sum_{\vb{e}\in\FF_2^d}
        \prod_{a=1}^m\mathbf{1}[(\vb{M}\vb{e})_a=t_a]
        \prod_{B\in\calB}w_B(\vb{e}_B)
    =Z_{\vb{M}}(\vb{t};\vb{w}).
\end{align}
\end{proof}

Denote by $\tw(G)$ the treewidth of the graph $G$.
For TN contraction, we use the standard model of pairwise contractions with dense intermediate tensors~\cite{markov_2008_simulating}.

\begin{theorem}[Rank DP-TN separation]
\label{thm:sop-tn-separation}
Given a partition function $Z_{\vb{M}}(\vb{t};\vb{w})$ as in \cref{def:partition-function}, let $G$ be its tensor-network graph from \cref{prop:decoding-tensor-network}.
Let $\calC:=\calB\cup\{\{c_1\},\ldots,\{c_m\}\}$, $b:=\max_{B\in\calB}|B|$, and $N:=m+d+\sum_{B\in\calB}2^{|B|}$.
Given an optimal rank decomposition of $G$, Rank DP evaluates $Z_{\vb{M}}(\vb{t};\vb{w})$ using $\poly(N)4^{b\cdot\rw(G)}$ arithmetic operations.
Every TN contraction order of $G$ has cost $2^{\Omega(\tw(G))}$.
Consequently, for any family with bounded $b$ and $\rw(G)+\log N=o(\tw(G))$, Rank DP takes $2^{o(\tw(G))}$ time, whereas TN contraction requires $2^{\Omega(\tw(G))}$ time.
Moreover, $\rw(T(\vb{M});\calC)\leq b\cdot\rw(G)$ and $\rw(G)\leq\tw(G)+1$.
\end{theorem}

Independent single-qubit Pauli factors have $b=2$, and two-qubit Pauli factors have $b=4$.
In both cases, Rank DP uses $\poly(N)2^{O(\rw(G))}$ arithmetic operations.
Splitting each parity (XOR) tensor into a connected binary tensor network with the same external indices preserves $G$ as a minor of the resulting graph, so the contraction lower bound still applies.
Further algebraic preprocessing used in the numerical comparisons is discussed in \cref{sec:numerical-experiments}.

\begin{proof}
Let $\calR$ be the tree in the given rank decomposition of $G$, of width $\rw(G)$.
Delete its variable leaves and any resulting empty branches, and suppress nodes with only one child.
Relabel each remaining weight leaf $u_B$ by the entire block $B$, retaining every check leaf as a singleton.
This gives a rank-decomposition tree $\calT$ of $T(\vb{M})$ over $\calC$.

Fix a node $u$ of $\calR$, and let $X_u\subseteq V(G)$ be the vertices assigned to the leaves of the subtree rooted at $u$.
The node $u$ represents the cut $X_u\mid(V(G)\setminus X_u)$.
Let $Y_u\subseteq V(T(\vb{M}))$ consist of the check nodes in $X_u$ and all variable nodes in every block $B$ whose weight node $u_B$ belongs to $X_u$, i.e., $Y_u=\{c_a:c_a\in X_u\}\cup\bigcup_{B\in\calB:\,u_B\in X_u}\{v_j:j\in B\}$.
This defines the corresponding cut $Y_u\mid(V(T(\vb{M}))\setminus Y_u)$ in the Tanner graph.
Every cut displayed by $\calT$ is obtained from some node $u$ in this way.
It therefore suffices to prove $\rho_{T(\vb{M})}(Y_u)\leq b\rho_G(X_u)$ for every $u$.

Write $X=X_u$, $Y=Y_u$, $\overline X=V(G)\setminus X$, and $\overline Y=V(T(\vb{M}))\setminus Y$.
Partition the variable vertices $V_{\mathrm{var}}:=\{v_j:j\in[d]\}$ into
\begin{align}
    V_{\mathrm{in}}&:=V_{\mathrm{var}}\cap X\cap Y,
    &O&:=V_{\mathrm{var}}\setminus(X\cup Y),\\
    D_-&:=V_{\mathrm{var}}\cap(X\setminus Y),
    &D_+&:=V_{\mathrm{var}}\cap(Y\setminus X).
\end{align}
Let $C_X,U_X$ be the check and weight vertices in $X$, and let $C_{\overline X},U_{\overline X}$ be those in $\overline X$.
Identify row $a$ of $\vb{M}$ with $c_a$ and column $j$ with $v_j$ when specifying its submatrices.
Let $\vb{W}$ have rows indexed by weight vertices and columns by variable vertices, with $W_{u_B,v_j}:=\mathbf{1}[j\in B]$.
By the definition of $Y_u$, we have $v_j\in Y$ if and only if $u_B\in X$ for $j\in B$, which gives $\vb{W}[U_X, O] = \vb{0}$ and $\vb{W}[U_{\overline X},V_{\mathrm{in}}] = \vb{0}$.
Ordering the rows as $(C_X,U_X,D_-,V_{\mathrm{in}})$ and the columns as $(O,D_+,U_{\overline X},C_{\overline X})$ gives
\begin{align}
    \vb{A}_G[X,\overline X]&=
        \begin{pmatrix}
            \vb{M}[C_X,O]&\vb{M}[C_X,D_+]&\vb{0}&\vb{0}\\
            \vb{0}&\vb{W}[U_X,D_+]&\vb{0}&\vb{0}\\
            \vb{0}&\vb{0}&\vb{W}[U_{\overline X},D_-]^{\sfT}&\vb{M}[C_{\overline X},D_-]^{\sfT}\\
            \vb{0}&\vb{0}&\vb{0}&\vb{M}[C_{\overline X},V_{\mathrm{in}}]^{\sfT}
        \end{pmatrix}.
    \label{eq:tn-cut-block-matrices}
\end{align}
For the Tanner cut, ordering the rows as $(C_X,V_{\mathrm{in}},D_+)$ and the columns as $(O,C_{\overline X},D_-)$ gives
\begin{align}
    \vb{A}_{T(\vb{M})}[Y,\overline Y]&=
        \begin{pmatrix}
            \vb{M}[C_X,O]&\vb{0}&\vb{M}[C_X,D_-]\\
            \vb{0}&\vb{M}[C_{\overline X},V_{\mathrm{in}}]^{\sfT}&\vb{0}\\
            \vb{0}&\vb{M}[C_{\overline X},D_+]^{\sfT}&\vb{0}
        \end{pmatrix}.
    \label{eq:tanner-cut-block-matrix}
\end{align}

The block upper-triangular form in \cref{eq:tn-cut-block-matrices} gives
\begin{align}
    \rho_G(X)\geq\rank\vb{M}[C_X,O]+\rank\vb{M}[C_{\overline X},V_{\mathrm{in}}]
        +\rank\vb{W}[U_X,D_+]+\rank\vb{W}[U_{\overline X},D_-].
\end{align}
Every column of either $\vb{W}[U_X,D_+]$ or $\vb{W}[U_{\overline X},D_-]$ has exactly one $1$, since each variable belongs to exactly one block.
The row indexed by $u_B$ has at most $|B|\leq b$ ones.
No two rows have a $1$ in the same column, so the nonzero rows are linearly independent.
Thus, $|D_+|\leq b\cdot\rank\vb{W}[U_X,D_+]$ and $|D_-|\leq b\cdot\rank\vb{W}[U_{\overline X},D_-]$.

Since $\vb{M}[C_X,D_-]$ has $|D_-|$ columns and $\vb{M}[C_{\overline X},D_+]^{\sfT}$ has $|D_+|$ rows, \cref{eq:tanner-cut-block-matrix} yields
\begin{align}
    \rho_{T(\vb{M})}(Y)
    &\leq\rank\vb{M}[C_X,O]+\rank\vb{M}[C_{\overline X},V_{\mathrm{in}}]+|D_-|+|D_+|\\
    &\leq\rank\vb{M}[C_X,O]+\rank\vb{M}[C_{\overline X},V_{\mathrm{in}}]
        +b\bigl(\rank\vb{W}[U_{\overline X},D_-]+\rank\vb{W}[U_X,D_+]\bigr)\\
    &\leq b\bigl(\rank\vb{M}[C_X,O]+\rank\vb{M}[C_{\overline X},V_{\mathrm{in}}]
        +\rank\vb{W}[U_{\overline X},D_-]+\rank\vb{W}[U_X,D_+]\bigr)\\
    &\leq b\rho_G(X)
    \leq b\cdot\rw(G).
    \label{eq:tn-to-sop-cut-rank}
\end{align}
Taking the maximum over the cuts of $\calT$ proves that its width is at most $b\cdot\rw(G)$.
Then, applying \cref{eq:decoding-dp-given-tree} to this rank-decomposition tree gives the claimed arithmetic bound, with structural preprocessing absorbed in the polynomial factor.
The tensor-network lower bound follows from the standard contraction-complexity result~\cite{markov_2008_simulating}, and the stated separation follows by comparing the exponents.
Finally, $\rw(G)\leq\tw(G)+1$ is the standard rank-width bound~\cite{oum_2008_branchwidth}.
\end{proof}

\subsection{Reed-Muller code families}
\label{subsec:rm-code-families}

We consider code-capacity noise with independent single-qubit Pauli factors.
For a code on $n$ qubits, let $\calB=\{B_i:1\leq i\leq n\}$, where $B_i=\{i,n+i\}$ contains the two error coordinates of qubit $i$.
The local weight is $w_{B_i}(\vb{e}_{B_i})=p_i(x_i,z_i)$, the noise distribution for the $i$-th qubit.
We specify a full-row-rank matrix $\vb{P}$ with $\ker\vb{P}=S$, as in \cref{subsec:decoding}.
Let $G$ be its tensor-network graph from \cref{prop:decoding-tensor-network}, and let $T(\vb{A})$ be the reduced Tanner graph from \cref{subsec:systematic-sop}.
Thus, we compare the rank-width $\rw(T(\vb{A});\calB)$ of the reduced Tanner graph with respect to $\calB$ with the treewidth $\tw(G)$ of the TN graph.

\paragraph{Punctured Reed-Muller codes.}
Fix $m\geq3$ and let $n_m:=2^m-1$.
Index the physical qubits by $V_m:=\FF_2^m\setminus\{\vb{0}\}$.
For a nonempty set $J\subseteq[m]$, define the monomial evaluation vector $\vb{h}_J:=\left(\prod_{j\in J}v_j\right)_{\vb{v}\in V_m}\in\FF_2^{n_m}$.
Let $\calR_m$ be the CSS code whose $X$- and $Z$-stabilizer spaces are
\begin{align}
    C_X&:=\Span\{\vb{h}_J:|J|=1\},
    &C_Z&:=\Span\{\vb{h}_J:1\leq|J|\leq m-2\}.
    \label{eq:rm-stabilizer-spaces}
\end{align}
Here, a vector $\vb{h}$ specifies $X(\vb{h}):=\bigotimes_i X^{h_i}$ or $Z(\vb{h}):=\bigotimes_i Z^{h_i}$, respectively.
The monomial evaluation vectors $\vb{h}_J$ are linearly independent.
The spaces $C_X$ and $C_Z$ have dimensions $m$ and $n_m-m-1$, respectively, and satisfy $C_X\subseteq C_Z^\perp$.
This is the first-order punctured quantum Reed-Muller code with parameters $[[n_m,1,3]]$~\cite{anderson_2014_conversion,hastings_2018_distillation}.

We now construct the constraint matrix $\vb{P}_m$ required by the partition-function formulation in \cref{subsec:decoding}.
Let $\vb{H}_X$ and $\vb{H}_Z$ have rows $\vb{h}_J^{\sfT}$ indexed by $|J|=1$ and $1\leq|J|\leq m-2$, respectively.
The all-ones vector $\vb{1}$ gives logical representatives $\overline{X}=X(\vb{1})$ and $\overline{Z}=Z(\vb{1})$.
Appending their commutation constraints to the syndrome constraints gives
\begin{align}
    \vb{P}_m:=
    \begin{pmatrix}
        \vb{H}_Z&\vb{0}\\
        \vb{0}&\vb{H}_X\\
        \vb{1}^{\sfT}&\vb{0}^{\sfT}\\
        \vb{0}^{\sfT}&\vb{1}^{\sfT}
    \end{pmatrix}
    \in\FF_2^{(n_m+1)\times2n_m}.
    \label{eq:rm-sector-matrix}
\end{align}
The first two block rows form the syndrome matrix $\vb{M}_s=\vb{H}\vb{\Lambda}$, up to row ordering.
The last two rows record commutation with $\overline{Z}$ and $\overline{X}$, respectively.
Since $\vb{H}_X$ and $\vb{H}_Z$ have independent even-weight rows and $\vb{1}$ has odd weight, $\rank\vb{P}_m=n_m+1$.
Moreover, $\vb{P}_m$ annihilates the binary stabilizer space $S=(C_X\times\{\vb{0}\})\oplus(\{\vb{0}\}\times C_Z)$, which has dimension $n_m-1$.
Thus, $\ker\vb{P}_m=S$, as required.

Let $\calB_m:=\{\{i,n_m+i\}:1\leq i\leq n_m\}$ be the physical-qubit partition.
Define $G_m$ to be the TN graph for $\vb{P}_m$ and $\calB_m$ in \cref{prop:decoding-tensor-network}, including the weight nodes.
Let $\vb{A}_m$ be the reduced matrix obtained from $\vb{P}_m$ as in \cref{subsec:systematic-sop}.
\Cref{cor:rm-tn-separation} shows that $T(\vb{A}_m)$ has rank-width $O(\log n_m)$ with respect to $\calB_m$.
This gives polynomial-time MLD by Rank DP, while dense pairwise contraction of $G_m$ has superpolynomial cost.

\paragraph{Steane-Reed-Muller concatenation.}
To extend this separation to codes with growing distance, we use the Steane $[[7,1,3]]$ code~\cite{steane_1996_multiple} as an outer code.
For $L\geq1$, take $L$ levels of Steane concatenation and encode each of the resulting $7^L$ physical qubits in the punctured quantum Reed-Muller code $\calR_m$ defined above.
Writing $d_{\min}$ for the minimum distance, the resulting code $\calQ_{m,L}$ has parameters $[[n,1,d_{\min}]]=[[7^L n_m,1,3^{L+1}]]$.
The case $m=4,L=1$ is the $[[105,1,9]]$ Steane-Reed-Muller code~\cite{jochym_2014_concatenated,chamberland_2016_thresholds}.

For the width comparison, we specify the constraint matrix $\vb{P}_{m,L}$ using the monomial stabilizer generators in each inner block and the six encoded Steane generators at every outer concatenation node.
At each level, encode the Pauli factors using the fixed logical representatives $\overline{X}$ and $\overline{Z}$ of the child code.
Appending the two root logical commutation constraints gives a full-row-rank matrix $\vb{P}_{m,L}$ whose kernel is the concatenated stabilizer space.
Let $G_{m,L}$ be the corresponding TN graph, $\vb{A}_{m,L}$ the reduced matrix, and $\calB_{m,L}$ the partition into physical-qubit blocks.

\Cref{cor:rm-tn-separation} shows that $\rw(T(\vb{A}_{m,L});\calB_{m,L})=O(m)$, independently of $L$.
Choosing $L=\lfloor\log_7 n_m\rfloor$ gives a family whose distance grows polynomially with its length.
Rank DP still gives polynomial-time MLD, while dense pairwise contraction of $G_{m,L}$ has superpolynomial cost.

\begin{corollary}[Width bounds for Reed-Muller code families]
\label{cor:rm-tn-separation}
For the matrices and partitions defined above, there is a constant $c>0$, independent of $m$ and $L$, such that
\begin{align}
    \rw(T(\vb{A}_m);\calB_m)&\leq2m+1,
    &\frac{cn_m}{m^{3/2}}&\leq\tw(G_m)\leq n_m+3,
    \label{eq:rm-reduced-widths}\\
    \rw(T(\vb{A}_{m,L});\calB_{m,L})&\leq\max\{8,2m+3\},
    &\frac{cn_m}{m^{3/2}}&\leq\tw(G_{m,L})\leq4n_m+6L.
    \label{eq:steane-rm-reduced-widths}
\end{align}
Rank decompositions over the physical-qubit partitions with the stated width bounds can be constructed in polynomial time.
Under independent single-qubit Pauli noise, Rank DP gives polynomial-time likelihood evaluation and MLD for both families in the arithmetic model, including decomposition construction and structural preprocessing.
Every dense pairwise contraction order for $G_m$ or $G_{m,L}$ has cost $2^{\Omega(n_m/m^{3/2})}$.
In particular, choosing $L=\lfloor\log_7 n_m\rfloor$ gives $n=\Theta(n_m^2)$, $d_{\min}=\Theta\!\left(n^{(\log_7 3)/2}\right)$, and $\rw(T(\vb{A}_{m,L});\calB_{m,L})=O(\log n)$, while dense contraction of $G_{m,L}$ requires $2^{\Omega(\sqrt{n}/(\log n)^{3/2})}$ time.
\end{corollary}

The proof is given in Supplemental \cref{proof:rm-tn-separation}.

\begin{remark}
\label{rem:rm-original-width}
Let $\calC_m$ and $\calC_{m,L}$ extend the corresponding physical-qubit partitions by singleton check nodes.
Before elimination, the same matrices satisfy
\begin{align}
    \Omega(n_m/m^{5/2})&\leq\rw(T(\vb{P}_m);\calC_m)\leq n_m+1,
    \label{eq:rm-original-width}\\
    \Omega(n_m/m^{5/2})&\leq\rw(T(\vb{P}_{m,L});\calC_{m,L})\leq\max\{8,n_m+1\}.
    \label{eq:steane-rm-original-width}
\end{align}
The same respective bounds hold for $\rw(G_m)$ and $\rw(G_{m,L})$.
Their proofs are given in Supplemental \cref{app:rm-original-width}.
Thus, the polynomial-time guarantee for Rank DP uses check elimination.
\end{remark}

\section{Numerical experiments on decoding}
\label{sec:numerical-experiments}

We test the nonnegative realization of Rank DP on code-capacity and circuit-level decoding problems.
In this section, Rank DP denotes the recurrence in \cref{prop:nonnegative-realization}, which evaluates the same partition function as the reduced SOP representation in \cref{subsec:systematic-sop}.
The experiments examine how the rank structure controls evaluation cost
and how numerical accuracy affects the information obtained from a
likelihood. Exact inference evaluates the complete finite sum for the
specified noise model; numerical error bounds quantify its floating-point
evaluation.

For code capacity, the output contains all $4^k$ masses
$P^{\cc}(\vb{s},\vb*{\ell})$, so its largest entry determines the maximum-likelihood (ML) label.
For the memory circuits below, the output comprises
$P^{\cl}(\vb{s},L_0=0)$ and $P^{\cl}(\vb{s},L_0=1)$ for one measured logical
observable. Their maximizing label solves this measured-observable
decoding task.
Writing these masses as $v_\ell(s)=P(s,\ell)$, their normalization gives
the posterior $\pi_\ell(s)=v_\ell(s)/\sum_jv_j(s)$.
The capacity comparison controls the full-vector error. The circuit
comparison additionally resolves each sector individually and certifies
a maximizing label.

\paragraph{Comparison and timing conventions.}
All main performance measurements use single-thread central processing
unit (CPU) workers on a shared host. We write FP64 for standard double
precision, with $53$ significand bits, and MP$b$ for multiprecision
arithmetic with $b$ significand bits; MP96 and MP128 therefore use $96$
and $128$ significand bits, respectively.
Within a paired comparison, both methods use the same stored FP64 weights
and syndrome queries. Capacity evaluation uses FP64 unless an extended
arithmetic format is explicitly stated; circuit precision is given below.
For independent and identically distributed (IID) sampling, repeated
syndromes are retained.
\emph{Planning} searches for a tree or contraction path;
\emph{compilation} builds its reusable maps and kernels; and
\emph{evaluation} computes the requested probabilities.
For fixed constraints and block partition, structural compilation is
shared across queries and noise-weight updates. Capacity times measure
reused evaluation, with independent references and accuracy checks
recorded separately. Circuit batch times include compilation and the
methods' own numerical checks. Planning is recorded separately throughout.
Core pinning and repeated paired measurements reduce timing variability
on the shared host.

The TN controls include nonnegative coset and categorical representations,
alternative equivalent bases, and searched contraction orders
\cite{cotengra,opt_einsum}; the concatenated-code control also uses its
inner/outer hierarchy. These are comparisons of achieved implementations
and finite search portfolios.
The numerical controls may apply algebraic preprocessing beyond local XOR-tensor decomposition, including changes of constraint basis and elimination of variables.
Such reformulations can change the graph underlying \cref{thm:sop-tn-separation}, so its lower bound does not automatically apply to these controls.
Nevertheless, these comparisons complement the theorem by assessing the practical performance of Rank DP against algebraically preprocessed TN implementations under matched output and accuracy requirements.
Detailed models, search outcomes, and validation are given in
Supplemental \cref{app:protocols,app:capacity-revision,app:circuit-mld-revision}.

\subsection{Code-capacity decoding and structural scaling}
\label{subsec:capacity-scaling}

For code-capacity noise, both Rank DP and the TN controls sum
nonnegative weights. The capacity examples are
Calderbank--Shor--Steane (CSS) codes.
We use independent physical-qubit Pauli blocks,
retaining the joint X/Z distribution within each qubit. The main plots use
depolarizing noise with total Pauli probability $p=.03$; a correlated
channel $(p_I,p_X,p_Y,p_Z)=(.945,.005,.045,.005)$ supplies a separate
robustness check. Every complete output is accepted at relative
$\ell_1$ error $10^{-10}$ against a positive reference with a numerical
error bound. For the vector $\vb{v}$ of sector masses, this criterion is
\begin{equation}
 \frac{\|\widehat{\vb{v}}-\vb{v}\|_1}{\|\vb{v}\|_1}\leq10^{-10}.
 \label{eq:experiment-vector-accuracy}
\end{equation}
The denominator is the syndrome likelihood, so the criterion controls
aggregate absolute error relative to that likelihood.

\paragraph{The Reed--Muller separation family.}
We first instantiate the punctured quantum Reed--Muller codes
$\mathcal R_m=[[2^m-1,1,3]]$ of \cref{cor:rm-tn-separation}, retaining
the specified monomial checks and all-ones logical representatives.
The experiment distinguishes the reduced graph used by Rank DP from the
original graph used in the dense-TN lower bound.
Recursively bisecting the ordered nonzero binary labels gives an explicit
physical-qubit decomposition. We compute every displayed cut rank for
$m=3,\ldots,17$; its achieved width is $2m-2$ throughout.
The prefix-subspace identities in Supplemental \cref{app:rm-capacity}
independently verify these cut ranks.

For the original graph, we combine the theorem's middle-levels separator
bound with an explicit complete-bipartite subgraph bound.
At $m=17$, the reduced decomposition has width $32$, whereas the original
graph has treewidth at least $1012$.
\Cref{fig:reed-muller-capacity}(a) shows these structural bounds, computed
by binary rank calculations through $m=17$.

For full likelihoods, the CPU comparison covers $m=3,\ldots,10$
($n=7,\ldots,1023$), using $12$ IID physical-error draws per size and three
repetitions of every query. Repeated syndromes are retained.
The TN portfolio includes a positive conditional parity-chain
representation: after solving the X constraints, it sums over
X-stabilizer choices and contracts the remaining Z-parity chain.
This algebraic preprocessing changes the constraint representation, so the original-graph lower bound does not apply to this TN control.
Nevertheless, this comparison tests whether Rank DP offers a practical advantage when the TN method also exploits the code's algebraic structure.

At $n=1023$, the median Rank DP/TN times are $9.240/26.137$\,s;
Rank DP is about $2.83$ times faster in paired evaluation.
Both timed methods use FP64, with every four-sector vector meeting the
relative-$\ell_1$ target $10^{-10}$.
At $1023$ qubits, the reference intervals identify two possible maximizing
labels for each query. An exact binary-algebra check establishes a symmetry
between those sectors for all twelve queries, certifying the twofold ML
ties. The symmetry argument, structural checks, TN plans, and accuracy
protocol are given in Supplemental \cref{app:rm-capacity}.

\begin{figure}[!htbp]
\centering
\includegraphics[width=.98\linewidth]{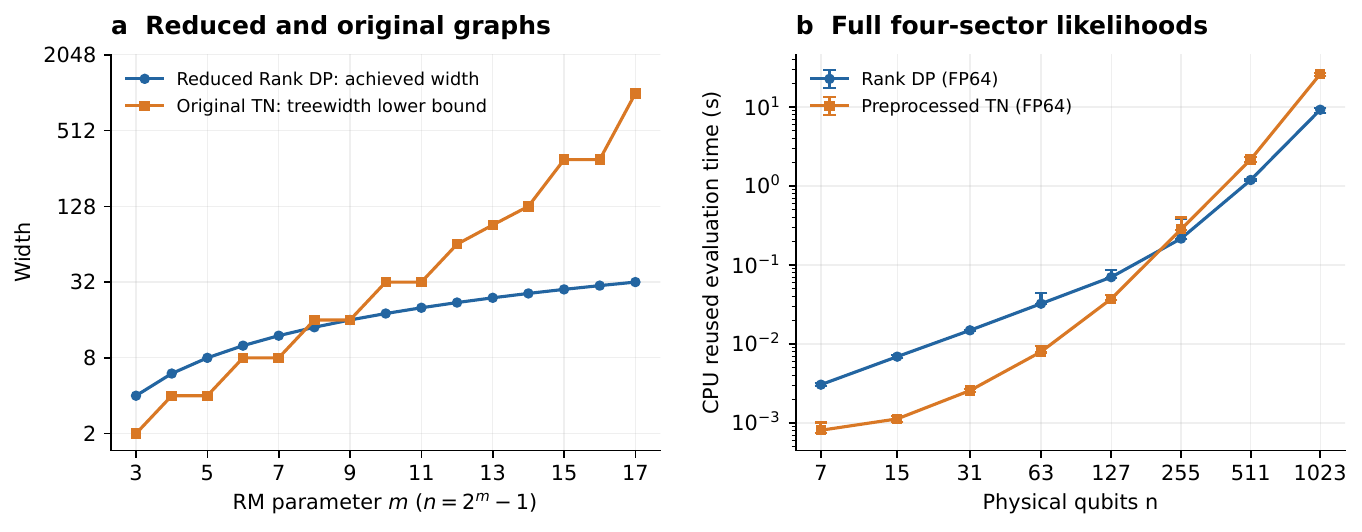}
\caption{Reed--Muller structural separation and practical inference.
(a) Width of the explicit physical-block decomposition after check
elimination and a proved treewidth lower bound for the specified original
graph. Every cut is checked through $m=17$.
(b) CPU FP64 evaluation of all four logical-sector likelihoods for
$7$--$1023$ qubits under depolarizing noise $p=.03$, with
relative-$\ell_1$ error at most $10^{-10}$.
The TN controls include algebraically preprocessed representations.
Each size uses $12$ IID queries with three repetitions; markers are
medians of per-query medians and bars show their observed range.
Times measure reused evaluation under the conventions stated above.
The code distance is three throughout this series.}
\label{fig:reed-muller-capacity}
\end{figure}

\paragraph{Rectangular three-dimensional surface codes.}
The Reed--Muller series connects the separation theorem to inference
costs at fixed distance three. To examine performance as both code size
and distance increase, we next consider a family of topological codes.
The open three-dimensional surface codes in \cref{fig:capacity-scaling} have
dimensions $3\times3\times L_z$, with $L_z=5,7,9$, and respectively
$93,135,177$ physical qubits. Their logical distances are
$(d_X,d_Z)=(9,L_z)$, so the minimum distance increases from five to nine.
These are surface codes with qubits on edges, vertex X checks, face Z
checks, and rough caps in the longitudinal direction. For cubic boxes,
rectifying this lattice gives the $\mathrm{SC}_g$ component of the
Vasmer--Browne construction~\cite{vasmer_browne_2019_surface}; the present
experiments evaluate individual-code likelihoods.
Supplemental \cref{app:cubic-capacity} extends the comparison to growing
cubic boxes: the $d=3$ confirmation favors TN, the $d=4$ paired comparison
remains unresolved within its time limit, and the $d=5$ case supplies
structural forecasts.
The cohorts draw IID physical errors after freezing the plans,
retaining repeated syndromes. TN evaluation uses opt\_einsum with the
selected contraction paths. Times are medians of per-query repeated
medians. The median paired TN/Rank DP ratios are $3.06$, $3.51$, and
$2.64$ at lengths five, seven, and nine, respectively. All vectors meet
the accuracy criterion.

The structural explanation uses the actual Rank DP decomposition.
All three selected rank decompositions have width $18$.
Their forecast scalar work for a four-sector vector grows from
$77.1$ to $144.2$ to $207.2$ million operations, while forecast live
numeric arrays remain about $18$\,MiB. The detailed bound sums the
compatible join work across the rank decomposition; layer-linear
trees of width $17$ have greater total work on these instances.
These are widths of the selected rank decompositions of the reduced Tanner graph $T(\vb{A})$ over $\calB$.
They give upper bounds on $\rw(T(\vb{A});\calB)$.
For this fixed $3\times3$ cross-section, the minimum distance is
$d=\min(9,L_z)$.

\begin{figure}[!htbp]
\centering
\includegraphics[width=.98\linewidth]{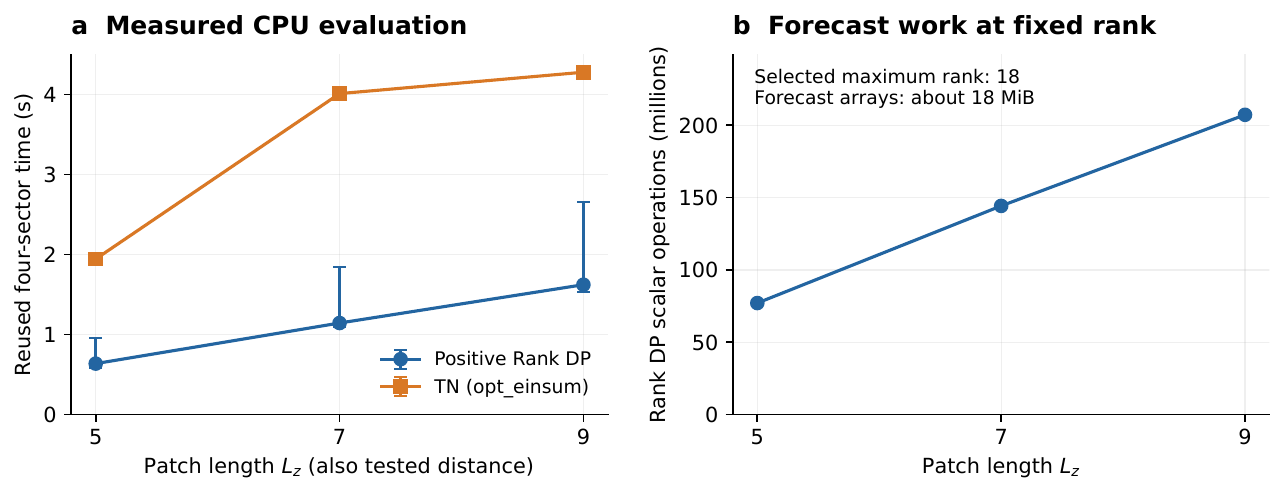}
\caption{CPU code-capacity inference on open $3\times3\times L_z$ surface codes
under depolarizing noise $p=.03$. (a) Rank DP and nonnegative TN both return
all four logical-sector masses in FP64; TN evaluation uses opt\_einsum.
Markers show median reused
evaluation time; bars show the observed range of per-query median times.
(b) Forecast Rank DP scalar work on the selected rank decompositions,
whose width remains $18$ and whose live numeric-array forecast stays
near $18$\,MiB. Work and memory forecasts exclude implementation overhead.
The tested distances are five, seven, and nine.}
\label{fig:capacity-scaling}
\end{figure}

\Cref{tab:capacity-complementary} gives complementary structures.
Golay23, the doubled CSS code, and the Golay--Steane concatenation favor
Rank DP under the stated protocol.
The concatenated TN contracts each inner Steane block exactly.
TB30 encodes four logical qubits, so both methods return all $256$ sectors.
Rank DP is about $18\%$ slower in paired evaluation, while the numeric-array
forecasts are $4.42$\,MB for Rank DP and $228.74$\,MB for TN.

\begin{table}[!htbp]
\centering
\caption{Complementary CPU code-capacity structures, depolarizing $p=.03$. Both methods compute all $4^k$ sectors in FP64. Times measure reused evaluation; TN/Rank DP is the median paired ratio. The first three rows use 12 queries each, and TB30 uses six. TN evaluation uses opt\_einsum, with nested contractions for Golay--Steane.}
\label{tab:capacity-complementary}
\begin{tabular}{llrrr}
\toprule
Code & $[[n,k,d]]$ & Rank DP (s) & TN (s) & TN/Rank DP \\
\midrule
Golay & $[[23,1,7]]$ & 0.0516 & 0.4552 & 8.828 \\
Doubled CSS & $[[49,1,5]]$ & 0.0346 & 0.1928 & 5.558 \\
Golay--Steane & $[[161,1,21]]$ & 0.1347 & 0.2404 & 1.781 \\
TB30 & $[[30,4,5]]$ & 3.7851 & 3.1988 & 0.850 \\
\bottomrule
\end{tabular}
\end{table}

Representation controls materially affect the conclusion. Equivalent
lower-weight and degree-balanced bases strengthen the doubled-code TN,
and exact contraction of each inner Steane block strengthens the
concatenated TN. The displayed ratios use these representations and
their selected decompositions.
The correlated-channel results, selection limits, and larger feasibility
pilots are detailed in Supplemental \cref{app:capacity-revision}.

\subsection{Circuit-level decoding and numerical precision}
\label{subsec:circuit-decoding}
\label{subsec:output-guarantees}

For noisy syndrome extraction, the requested numerical guarantee matters
as much as the contraction cost. A full vector can be accurate while a
small sector has poor relative accuracy. For example, replacing
$(0.1,10^{-13})$ by $(0.1,0)$ gives relative $\ell_1$ error about
$10^{-12}$ and preserves the largest entry while losing the rare mass.
A decoding decision and an accurate rare probability therefore require
different checks.

\paragraph{Nonnegative arithmetic and signed cancellation.}
The nonnegative Rank DP recurrence uses only nonnegative sums and products.
Under the arithmetic and range assumptions of Supplemental \cref{prop:roundoff-error}, its returned value satisfies $|\widehat Z-Z|\leq\gamma_\delta Z$, where $\gamma_\delta=\delta\varepsilon_{\mathrm{mach}}/(1-\delta\varepsilon_{\mathrm{mach}})$, $\varepsilon_{\mathrm{mach}}$ is unit roundoff, and $\delta$ is determined by the compiled reductions.
For $\gamma_\delta<1$ this gives the interval
$[\widehat Z/(1+\gamma_\delta),\widehat Z/(1-\gamma_\delta)]$.
The bound quantifies arithmetic error for the exact stored weights under
the stated assumptions.

The detector-Fourier TN enforces parity by a signed sum.
Cancellation can make its absolute error large relative to a small
likelihood, motivating a separate roundoff analysis and, where needed,
more precision (Supplemental \cref{app:fourier-certificate}).
Categorical and coset TNs sum nonnegative terms. The bounded searches of
these representations, including their planning outcomes, are described
in Supplemental \cref{app:circuit-ledger}.

\paragraph{From likelihood intervals to an ML label.}
Given intervals $[L_i,U_i]$ for the requested sector masses, a unique
winner $i$ is established when $L_i>\max_{j\ne i}U_j$.
When the intervals overlap, we retain the possible winners.
The same intervals bound posterior and excess failure probabilities
(Supplemental \cref{app:block-output-bounds}).
Decision certification depends on interval separation, while relative
accuracy depends on the size of each mass. We therefore report vector
accuracy, individual-sector accuracy, and the ML decision separately in
the precision diagnostic.

On four predeclared queries per color circuit, native FP64 Fourier TN
passes the reference-certified vector target on $8/8$ queries, while
$6/8$ provably fail the individual-sector target. Its own intervals
certify the ML label on $8/8$ queries. Fourier MP96, the reused MP128
values, and positive Rank DP pass the sector and ML checks on all eight
queries. These selected queries illustrate how precision affects the
three accuracy criteria. The reference comparison, each method's
standalone certificate, and sufficient-precision bounds for the full
cohorts are reported in Supplemental \cref{app:circuit-mld-revision}.

\begin{figure}[!htbp]
\centering
\includegraphics[width=.98\linewidth]{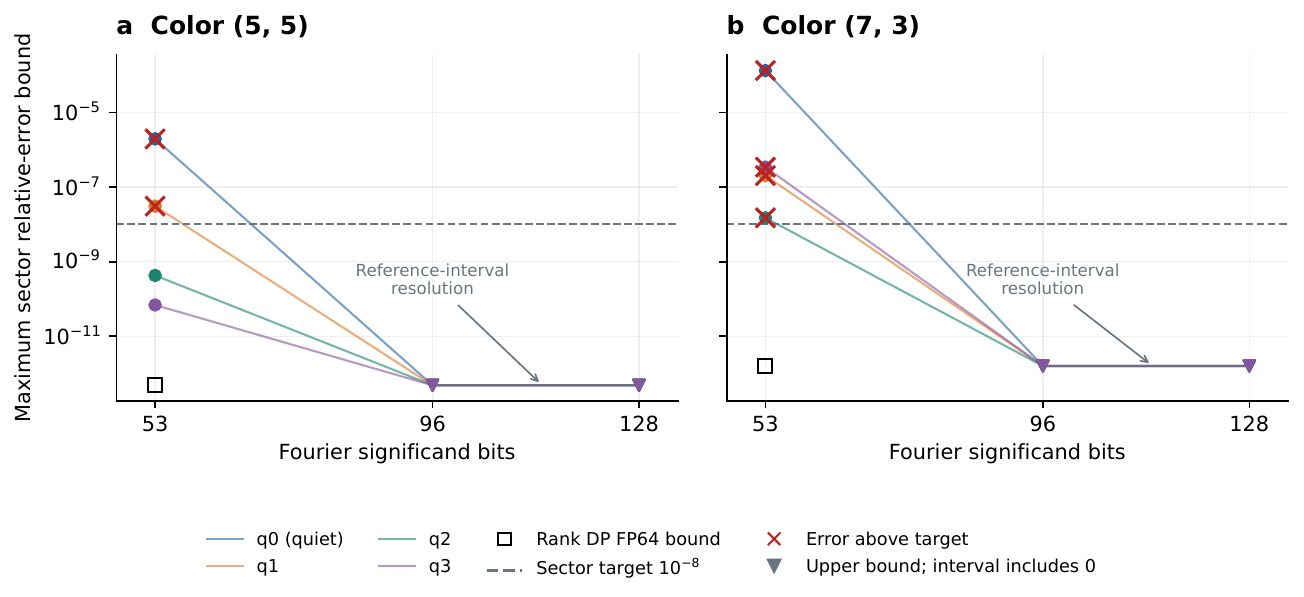}
\caption{Signed Fourier precision on four fixed queries per color circuit.
Values enclose the larger of the two sector relative errors, using independent
positive Rank DP reference intervals. Crosses establish error above $10^{-8}$;
downward triangles mark upper bounds with intervals that include zero.
Open squares show positive Rank DP FP64 bound resolution.
The MP96 and MP128 upper bounds reach the resolution of the reference intervals.
All methods use the same stored models; Fourier precisions share a saved TN
path and Rank DP uses its saved positive tree. MP128 values come from
the completed batch comparison.}
\label{fig:circuit-precision}
\end{figure}

\paragraph{Complete color-circuit batches.}
The performance comparison uses Stim color-code memory circuits with
$(d,r)=(5,5)$ and $(7,3)$, gate/reset/measurement rates $0.001$, and no
additional idle-noise channel \cite{stim}. Supported exact DEM conversion
retains complete event supports without disjoint-error approximation.
Each query returns both measured-observable sectors and a certified ML
label. In addition, each mass must meet relative error $10^{-8}$.
Rank DP uses FP64; Fourier TN uses fixed $128$-significand-bit arithmetic
with its own certificate. The comparison combines ML decision
certification with a relative-accuracy guarantee for each sector.
Lower sufficient precisions for the complete cohorts are bounded in
Supplemental \cref{app:circuit-mld-revision}; the timing comparison retains
its common fixed MP128 policy.

Each case uses the quiet syndrome and the first $31$ distinct sampled
detector patterns, selected by detector outcomes alone. These fixed cohorts
provide query diversity for the timing comparison. Three adjacent Rank DP/TN batch pairs
alternate execution order; every batch starts a fresh worker, compiles
once, and evaluates all $32$ queries. The table includes startup,
compilation, evaluation, certification, and recording, while planning and
the independent terminal audit are separate.
Every vector and component interval passes that audit.

\begin{table}[!htbp]
\centering
\caption{Certified color-circuit observable decoding: cold batches of 32 queries, including compilation. Rank DP uses FP64; Fourier TN uses 128 significand bits. Both sectors satisfy relative error $10^{-8}$ and yield a certified ML label. Times are separate medians of three batches; the ratio and bracket are the median and range of paired TN/Rank DP ratios. Planning and independent terminal validation are separate.}
\label{tab:circuit-mld}
\begin{tabular}{lrrr}
\toprule
Circuit $(d,r)$ & Rank DP batch (s) & TN batch (s) & Paired TN/Rank DP \\
\midrule
Color (5, 5) & 1662.0 & 2479.6 & 1.492 [1.080, 1.496] \\
Color (7, 3) & 63.1 & 489.8 & 7.758 [6.855, 8.414] \\
\bottomrule
\end{tabular}
\end{table}

The selected rank structure helps explain the different costs of these
circuits. Color $(5,5)$ has $1104$ DEM events and a selected maximum
interface rank of $22$; color $(7,3)$ has $782$ events and rank $16$.
Their two-sector work counts are $4.23\times10^9$ and
$1.80\times10^8$ terms, respectively. The shorter round count and smaller
interfaces make the latter computation cheaper despite its larger nominal
code distance. Supplemental \cref{app:circuit-mld-revision} gives the
structural and resource details. The timing ratios combine the selected
representations, plans, arithmetic precisions, and execution kernels.
\FloatBarrier

\section{Applications of likelihood evaluation}
\label{sec:applications}

Likelihood evaluation supports several tasks beyond maximum-likelihood decoding.
We demonstrate these applications of Rank DP by computing syndrome likelihoods for learning circuit noise parameters, event probabilities for evaluating postselection, and logical posterior probabilities for quantifying decoder optimality gaps.
All three applications use the partition function in \cref{eq:inference-primitive}.
The numerical error bounds in Supplemental \cref{app:numerical-accuracy} provide enclosures for the computed probabilities.

\subsection{Circuit noise learning from syndromes}
\label{subsec:noise-learning}

Noise-prior estimation from syndrome correlations includes
pairwise and higher-order detector-event models~\cite{spitz_2018_adaptive,takou_2025_estimating,remm_2025_correlations}.
Differentiable exact TN learning has also been developed for syndrome
likelihoods~\cite{cao_2026_differentiable} and conditional logical-label
cross-entropy~\cite{pancotti_2026_exact_learning}.
Here, Rank DP provides an alternative exact likelihood evaluator for the
established syndrome-only objective.
For physical rates $\theta$, marginalizing the logical outcomes gives
$P_\theta^{\cl}(s)=Z_{\vb{D}}(s;\theta)$ and the mean negative
log-likelihood (NLL)
\begin{equation}
 \calL_N(\theta)=-\frac1N\sum_{i=1}^N\log P_\theta^{\cl}(s_i).
 \label{eq:learning-objective}
\end{equation}
The detector constraints, state spaces, and transition maps remain fixed
as the local weights vary. Differentiating the joins in
\cref{eq:block-join} therefore reuses the structural compilation for both
likelihoods and gradients.

We use the color $(7,3)$ circuit from \cref{subsec:circuit-decoding},
with physical gate, reset, and measurement rates
$\theta=(g,r,m)$. A supported exact map converts these three rates into
the $782$ DEM event priors $p_j(\theta)$, retaining their complete
detector and observable supports. The generating rates are
$(.0012,.0007,.0016)$ and initialization is $(.0005,.0018,.0006)$.
Three forward sensitivities propagate alongside each likelihood using
the product rule. Supplemental \cref{app:circuit-learning-revision}
gives the rate map, derivative recurrence, and independent numerical
checks.

Three independent datasets each contain $4096$ training, $512$
validation, and $2048$ test circuit syndromes. Each fit uses $64$ Adam
updates, a minibatch of $64$ training positions per update, and learning
rate $.03$. Batches are sampled independently from the training pool;
repeated syndromes retain their multiplicities. Training and checkpoint
selection use syndrome likelihoods. Validation NLL selects among
initialization and every eighth update, and the selected fit is then
evaluated on the independent test split.

To examine recovery of the generating model, we plot
\begin{equation}
 E_{\mathrm{DEM}}(t)=\frac{1}{782}\sum_{j=1}^{782}
 \frac{|p_j(\theta_t)-p_j^\star|}{p_j^\star},
 \label{eq:learning-dem-error}
\end{equation}
where $p_j^\star>0$ is the generating prior of event $j$.
This diagnostic weights each event equally and uses the known rates of
the synthetic experiment. The optimized objective is syndrome NLL,
so the two quantities can attain their minima at different updates.
\Cref{fig:noise-learning}(a) shows all $65$ saved iterates and marks the
validation-selected updates $48$, $64$, and $64$.
Mean relative DEM-prior error falls from $54.2\%$ to
$5.2\%$--$7.0\%$ at those checkpoints. Held-out NLL improves by
$0.093$--$0.100$ nats per syndrome, with the paired intervals in
\cref{fig:noise-learning}(b).

\begin{figure}[!htbp]
\centering
\includegraphics[width=.98\linewidth]{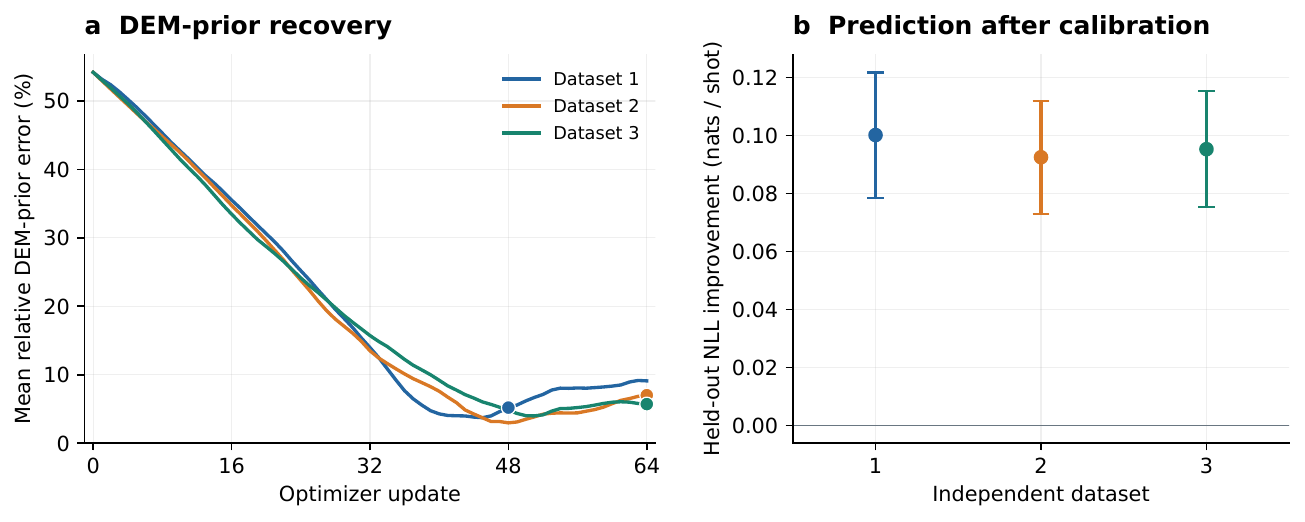}
\caption{Syndrome-only calibration of the color $(7,3)$ circuit.
(a) Mean absolute relative error of the $782$ DEM event priors against
the generating model, in percent. Curves show every optimizer update;
filled markers identify checkpoints selected by validation NLL.
(b) Held-out NLL improvement from initialization to each selected fit.
The approximate $95\%$ paired intervals quantify test-sample uncertainty
conditional on that fitted model. Each point uses its own independent
test set.}
\label{fig:noise-learning}
\end{figure}

The individual physical-rate estimates provide a complementary view:
the selected reset rates are $1.4$--$2.2$ times the generating value.
A detector-moment calculation at the generating rates finds full local
rank and lower statistical resolution for the reset parameter.
Supplemental \cref{app:circuit-learning-revision} gives the fitted rates,
local resolution calculation, and a paired downstream decoding check.
On $1024$ fresh circuit shots, the initial model, all three selected fits,
and the generating model yield identical certified ML decisions, with a
common sample mean conditional failure probability of approximately
$9.76\times10^{-6}$ under the generating-model posterior.
Independent validation checks the physical-rate map against fresh circuit
conversions, likelihoods against positive Rank DP and Fourier TN, and
the three signed derivatives against finite differences at two step sizes.

\subsection{Direct rare-event probabilities and postselection}
\label{subsec:rare-events}

For a noise model $\nu\in\{\cc,\cl\}$, let $A$ be the acceptance event
and let $B$ be a bad logical event after a declared recovery. Their joint
probability determines the conditional failure probability:
\begin{equation}
 P^\nu(B\mid A)=\frac{P^\nu(A\cap B)}{P^\nu(A)},
 \qquad P^\nu(A)>0.
 \label{eq:postselection-probability}
\end{equation}
Both $P^\nu(A)$ and $P^\nu(A\cap B)$ are evaluated using partition
functions of the form in \cref{eq:inference-primitive}, summing disjoint
logical sectors where needed. Thus the accepted-bad probability is evaluated as a sum of
nonnegative contributions. If the two masses lie in $[L_A,U_A]$ and
$[L_B,U_B]$, with $L_A>0$, then
\[
 P^\nu(B\mid A)\in[L_B/U_A,U_B/L_A]\cap[0,1].
\]
Separate bounds on numerator and denominator propagate to the ratio.

Static postselection accepts zero syndrome and applies identity recovery;
$B$ is the union of all nontrivial logical Pauli classes.
The smallest accepted-bad errors are logical operators of weight $d$.
Under independent depolarizing noise their probability therefore starts
at order $p^d$, while acceptance tends to one as $p$ tends to zero.
This explains the distance-dependent slopes in
\cref{fig:rare-events}(b).

At depolarizing $p=10^{-5}$, surface-$5$ has acceptance approximately
$0.99975003$ and conditional logical failure $6.58\times10^{-26}$.
\Cref{fig:rare-events} displays the suppression together with the
probability of retaining a run for surface-$3$, surface-$5$, and the
$[[17,1,5]]$ square-octagon color code~\cite{bravyi_2015_doubled}.
Positive Rank DP bounds both masses throughout the $126$-case grid,
including the biased and correlated channels defined in
Supplemental \cref{app:postselection-details}.
These finite sums resolve the extreme tail directly.

\begin{figure}[!htbp]
\centering
\includegraphics[width=.98\linewidth]{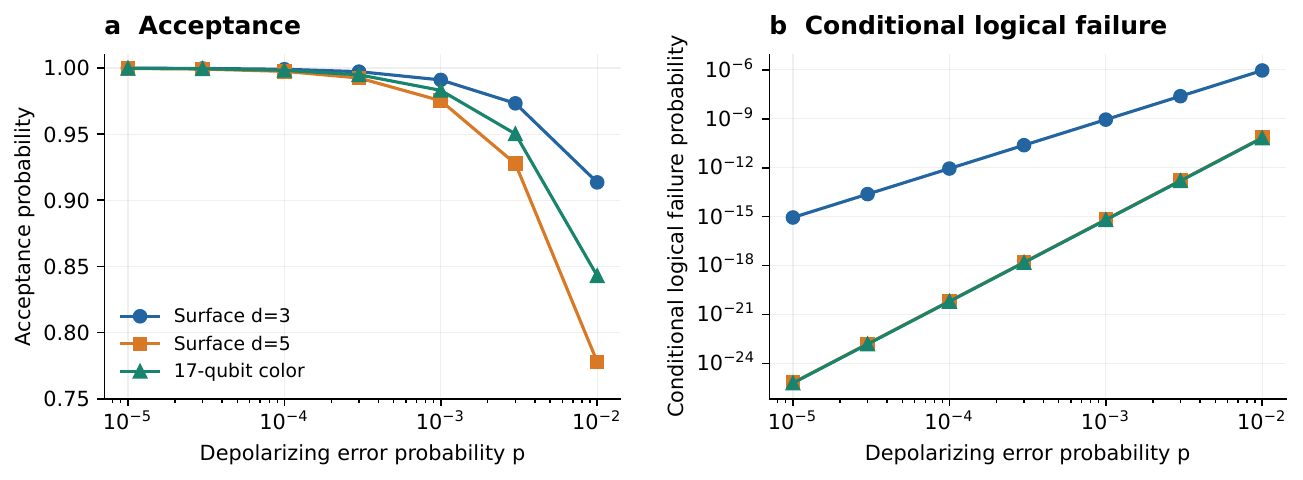}
\caption{Static postselection under code-capacity depolarizing noise.
(a) Zero-syndrome acceptance probability.
(b) Conditional probability of a nontrivial logical class under identity
recovery. Positive summation separately evaluates acceptance and
accepted-bad probability; panel (b) shows their ratio.
Supplemental \cref{app:postselection-details} gives the channel grid,
interval construction, and representative numerical enclosures.}
\label{fig:rare-events}
\end{figure}

For circuit postselection, specify the constrained detector set $J$ and
write $A_J=\{s_j=0\text{ for all }j\in J\}$.
The other detector outcomes are marginalized.
For rotated Z memory with identity recovery, we take $B=\{L_0=1\}$,
the final measured memory-observable flip.
Constraining all detectors gives stronger suppression and lower acceptance
than constraining only the first round, as shown for distance-$3$ memories
in Supplemental \cref{tab:partial-postselection}.
The circuit precision study in \cref{subsec:circuit-decoding}
motivates resolving each small mass individually before forming this ratio.

\subsection{Measuring decoder optimality gaps}
\label{subsec:decoder-gap}

The exact logical posterior also measures the quality of a practical
decoder. Writing $\pi_\ell(s)=P^{\cc}(s,\ell)/P^{\cc}(s)$, a decoder
$\delta$ has conditional failure probability and excess risk
\begin{align*}
 R_\delta(s)&=1-\pi_{\delta(s)}(s),&
 R_*(s)&=1-\max_\ell\pi_\ell(s),\\
 \Delta_\delta(s)&=R_\delta(s)-R_*(s)
                 =\max_\ell\pi_\ell(s)-\pi_{\delta(s)}(s).
\end{align*}
Averaging these quantities over \emph{i.i.d.} sampled syndromes estimates population
failure probabilities and their gaps. Conditioning on the syndrome
integrates over the remaining logical-outcome randomness.
Numerical enclosures and uncertainty from syndrome sampling are computed
separately.

We compare joint Bayes inference under the true Pauli model, exact
inference under its factorized X/Z approximation, marginal belief
propagation with ordered-statistics decoding (BP+OSD)~\cite{panteleev_osd}, and a conditional
BP+OSD heuristic based on the conditional X/Z decoding
principle~\cite{delfosse_2014_correlations}. The conditional heuristic alternates X and Z updates
using the joint channel's conditional probabilities. It retains
syndrome-valid corrections and ranks them by joint physical-error
probability; exact logical-sector sums then measure its remaining excess
risk. Supplemental \cref{app:decoder-details} gives the update rule and
fixed decoder settings.

The distance-$5$ surface-code study fixes both binary marginals at $.05$
and varies $y=p_Y$ over $.0025,.025,.045$:
\[
 (p_I,p_X,p_Y,p_Z)=(.9+y,.05-y,y,.05-y).
\]
The first point has independent X/Z components; the total Pauli error
probabilities are $.0975$, $.075$, and $.055$, respectively.
At each channel, all decoders use the same $8192$ IID physical draws,
retaining repeated syndromes.

At $p_Y=.045$, the estimated failure probabilities are $4.81\%$ for
marginal BP+OSD, $1.47\%$ for conditional-low, and $0.44\%$ for joint Bayes
inference. The paired reduction is $3.34$ percentage points, with a
conservative $95\%$ interval $[2.49,4.19]$ points. Conditional-low retains
a $1.03$-point gap to Bayes, with interval $[0.60,1.46]$ points.
Exact factorized inference and marginal BP+OSD have the same reported
failure probability at both correlated points. These comparisons quantify
the effect of the noise model and the benefit of using its local
correlations in a practical decoder.

\begin{figure}[!htbp]
\centering
\includegraphics[width=.72\linewidth]{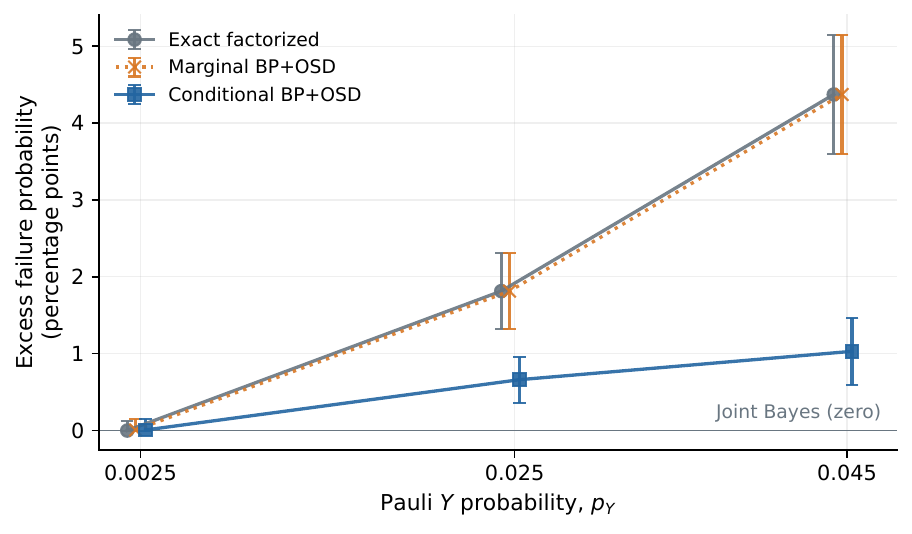}
\caption{Excess decoding failure probability above joint Bayes inference
for the distance-$5$ surface code.
Each channel uses $8192$ IID physical draws with perfect syndrome
measurements. The binary marginals are $.05$; the total Pauli error
probability is $.1-p_Y$.
Each interval has $95\%$ coverage for its stated estimand, with the
numerical radius added separately. Markers are offset horizontally for
visibility. The conditional curve uses the fixed low-effort heuristic.}
\label{fig:decoder-gap}
\end{figure}

All draws contribute to the estimates, and every posterior and decoding
decision is resolved. Supplemental \cref{app:decoder-details} gives the
distance-$3$ population calculation, additional heuristic settings,
realized failure counts, and independent enumeration checks.
\FloatBarrier

\section{Discussion}
\label{sec:discussion}

In this work, we formulated maximum-likelihood decoding as partition-function evaluation and applied Rank DP through a quadratic SOP representation.
The formulation covers code-capacity and circuit-level noise with independent local fault factors.
The arithmetic complexity is polynomial in the input size and exponential in the Tanner graph's rank-width with respect to the partition into fault blocks and singleton checks, including decomposition construction.
Gaussian elimination yields a reduced Tanner graph without increasing the decomposition width.
For punctured quantum Reed-Muller codes and a Steane-concatenated family with growing distance, this reduction gives polynomial-time MLD under independent single-qubit Pauli noise.
Dense pairwise contraction of the corresponding TN graphs has superpolynomial cost.

The nonnegative realization preserves the table sizes and arithmetic complexity while avoiding cancellation from signed weights.
Under explicit arithmetic assumptions, it gives relative floating-point error bounds.
Numerical comparisons under matched output and accuracy requirements show runtime advantages on selected code-capacity and circuit-level instances over the tested TN implementations, including algebraically preprocessed controls.
These comparisons assess the practical performance of Rank DP beyond the graph representations covered by the separation theorem.
We further used likelihood evaluation to learn circuit noise parameters from syndromes, evaluate rare-event probabilities under postselection, and quantify decoder optimality gaps.

Further work could identify structural conditions on codes and syndrome-extraction circuits that yield small rank-width after Gaussian elimination.
This would extend the range of instances for which Rank DP offers efficient exact likelihood evaluation.
A second direction is to develop practical rank-decomposition algorithms that jointly reduce total join work and memory use.
The width alone does not determine these costs.
Finally, likelihood-based noise learning could be extended to larger parameter spaces and experimental syndrome data.
Quantifying how parameter uncertainty affects predicted decoding failure and postselection probabilities would connect noise estimation more closely to error-correction performance.

\paragraph{Statement of AI use.}
We used GPT 5.6 Sol and GPT 6 Astra to assist with literature searches, proof development, mathematical checks, numerical experiments, and manuscript preparation.
The authors take full responsibility for every claim, proof, and citation in it.

\bibliographystyle{unsrt}
\bibliography{ref}
\clearpage
\appendix
\part*{Supplemental material}
\phantomsection\label{page:supplemental-material}
\setcounter{figure}{0}
\setcounter{table}{0}
\setcounter{equation}{0}
\renewcommand{\thefigure}{S\arabic{figure}}
\renewcommand{\thetable}{S\arabic{table}}
\renewcommand{\theequation}{S\arabic{equation}}
\renewcommand{\theHfigure}{supp.\arabic{figure}}
\renewcommand{\theHtable}{supp.\arabic{table}}
\renewcommand{\theHequation}{supp.\arabic{equation}}
\section{Proofs for Rank DP}
\label{app:proofs}

This section provides the supporting proofs for \cref{sec:sop-alg}.
We establish the rank-decomposition construction over a partition, prove the nonnegative realization, and derive its Fourier relation to the SOP dynamic program.

\subsection{Rank-decomposition construction over a partition}
\label{proof:partition-rank-decomposition}

Here, we prove \cref{prop:partition-rank-decomposition}.
\begin{proof}
Every part in the partition is represented by a leaf, so $c\leq r$.
If $c=0$, no edge joins distinct parts of $\calP$, and any rooted binary tree on $\calP$ has width zero.
We may therefore assume $c>0$, which also implies $|\calP|\geq2$.

Let $V:=V(G)$ and $Y_{\calS}:=\bigcup_{P\in\calS}P$ for $\calS\subseteq\calP$.
Define
\begin{align}
    f(\calS):=\rho_G(Y_{\calS}).
\end{align}
The function $f$ is nonnegative and integer-valued, with $f(\varnothing)=0$.
It is \emph{symmetric}, meaning that $f(\calS)=f(\calP\setminus\calS)$, because $Y_{\calP\setminus\calS}=V\setminus Y_{\calS}$ and cut-rank is symmetric.
A set function is \emph{submodular} if, for all $\calS,\calQ\subseteq\calP$,
\begin{align*}
    f(\calS)+f(\calQ)\geq f(\calS\cap\calQ)+f(\calS\cup\calQ).
\end{align*}
Since the parts are disjoint, $Y_{\calS\cap\calQ}=Y_{\calS}\cap Y_{\calQ}$ and $Y_{\calS\cup\calQ}=Y_{\calS}\cup Y_{\calQ}$.
Hence, submodularity of the cut-rank function~\cite[Corollary~6.2]{oum_seymour_2006_approximating} implies submodularity of $f$.
An integer-valued symmetric submodular function that vanishes on the empty set is called a \emph{connectivity function}.

A \emph{branch-decomposition} of a connectivity function $f$ is a pair $(\calT,\mu)$, where $\calT$ is a rooted binary tree in which every internal node has exactly two children, and $\mu:\leaf(\calT)\to\calP$ is a bijection.
Its width is $\max_{u\in V(\calT)} f(\calS_u)$, where $\calS_u\subseteq\calP$ is the set of labels assigned by $\mu$ to the leaves of the subtree rooted at $u$.
The minimum width over all such pairs is the \emph{branch-width} of $f$.
Since $f(\calS_u)=\rho_G(Y_{\calS_u})$, this minimum equals $r$.

We now apply the Oum-Seymour approximation algorithm~\cite{oum_seymour_2006_approximating} in the connectivity-function form~\cite[Section~5]{korhonen_oum_2026_connectivity}.
For an integer $k\geq c$, it either constructs a branch-decomposition of $f$ of width at most $3k+c$ or certifies that its branch-width exceeds $k$.
Here, $|\calP|\leq n$, $0\leq f(\calS)\leq n$, and each value of $f$ is computable by binary Gaussian elimination in polynomial time.
The running time in this instance is therefore $\poly(n)2^{O(k)}$.
Trying $k=c,c+1,\ldots$ until a decomposition is returned stops no later than $k=r$.
The total running time is $\poly(n)2^{O(r)}$, and the returned width is at most $3r+c\leq4r$.
\end{proof}

\subsection{Proof of the nonnegative realization}
\label{proof:block-inference}

Here, we prove \cref{prop:nonnegative-realization}.
\begin{proof}
At a leaf, enumerate its local binary assignments, retain those whose
parities lie in $I_B$, and sum their shifted weights by parity. This
gives \cref{eq:block-table} directly.

For a join, abbreviate $U_L:=\im\vb{M}_{Y_L}$, $U_R:=\im\vb{M}_{Y_R}$, and $U_O:=\im\vb{M}_{\overline Y}$.
Let $\vb{h}\in I_Y$.
Since
$\vb{h}\in(U_L+U_R)\cap U_O$, there are
$\vb{h}_L\in U_L$ and $\vb{h}_R\in U_R$ with
$\vb{h}_L+\vb{h}_R=\vb{h}$. Moreover,
$\vb{h}_L=\vb{h}+\vb{h}_R\in U_O+U_R$, so
$\vb{h}_L\in I_{Y_L}$; the symmetric argument gives
$\vb{h}_R\in I_{Y_R}$. Thus every parent state has a compatible pair.
The kernel of the addition map
\[
    I_{Y_L}\oplus I_{Y_R}\longrightarrow\FF_2^m,
    \qquad (\vb{a},\vb{b})\longmapsto\vb{a}+\vb{b},
\]
is $\{(\vb{a},\vb{a}):\vb{a}\in I_{Y_L}\cap I_{Y_R}\}$.
Every nonempty fiber therefore has $2^{\mu_Y}$ elements. These fibers
and a particular pair for each parent basis vector can be represented by
exact binary linear maps.

Every assignment on $Y$ appearing in \cref{eq:block-table} restricts
uniquely to assignments on the two children. The preceding argument
shows that their parities lie in the child state spaces.
Conversely,
two child assignments contributing to a compatible pair concatenate to
an assignment of parent parity $\vb{h}$. Since no block is divided by
the join, its weight is the product of the two child weights.
Expanding \cref{eq:block-join} therefore gives every summand in
\cref{eq:block-table} exactly once. Induction proves the table invariant.
At the root $I_{[d]}=\{0\}$, so its entry sums exactly over
$\ker\vb{M}$ and hence gives \cref{eq:inference-primitive} after the shift.

Leaf enumeration contributes at most a constant times $\sum_B2^{|B|}$ scalar operations.
A join has $2^{\rho_G(Y)}$ entries, each with $2^{\mu_Y}$ products and a nonnegative sum of those products.
The set of all compatible pairs used at that join is a subset of $I_{Y_L}\times I_{Y_R}$.
Consequently,
\begin{align*}
    2^{\rho_G(Y)+\mu_Y}
    \leq2^{\rho_G(Y_L)+\rho_G(Y_R)}
    \leq4^{\omega_A}.
\end{align*}
There are $N-1$ joins and $2N-1$ nodes.
These observations give the operation and table bounds.
\end{proof}

\subsection{Fourier relation between the two dynamic programs}
\label{proof:sop-positive-basis}

The nonnegative recurrence avoids cancellation from signed weights while retaining the table sizes of the SOP recurrence.
This comparison is related to Fourier duality of code partition functions on normal factor graphs~\cite{forney_2011_duality}; here we derive the explicit boundary transformation for the block-preserving tables.
To explain the relation between the two formulations, we compare their tables for the shifted instance $\vb{M}\vb{x}=\vb{0}$ on the same rank-decomposition tree.
The following lemma shows that the tables contain the same information in different bases.

\begin{lemma}[Fourier change of basis]
\label{lem:sop-positive-basis}
Use the assumptions and notation of \cref{subsec:nonnegative-accuracy}.
Construct the quadratic SOP for the shifted weights $\widetilde{w}_B$ and zero target as in \cref{cor:systematic-sop}.
At each node $u$ of $\calT$, the tables in \cref{eq:sop-dp-table,eq:block-table} are related by an invertible Fourier change of basis.
Both tables have $2^{\rho_G(Y_u)}$ entries, and the local weight transformations preserve the blocks in $\calB$.
\end{lemma}

\begin{proof}
Use the column basis $B_0$ from \cref{subsec:systematic-sop} and write $D_0=[d]\setminus B_0$ and $r=|B_0|$.
Apply \cref{eq:systematic-row-reduction} and delete the zero rows to work in systematic coordinates $\vb{M}=(\vb{I}_{B_0}\mid\vb{A})$.
These row operations preserve the solution set and identify the parity labels in \cref{eq:block-table} bijectively with their transformed labels.
Use the solution $\vb{e}_0$ and shifted weights $\widetilde{w}_B$ from \cref{subsec:nonnegative-accuracy}, so the constraint is $\vb{M}\vb{x}=\vb{0}$.
For each physical block $B\in\calB$, Fourier transform only its basis coordinates:
\begin{align}
    \widehat{\widetilde w}_B
        (\vb*{\lambda}_{B\cap B_0},\vb{z}_{B\cap D_0})
    :=\sum_{\vb{x}_{B\cap B_0}\in\FF_2^{B\cap B_0}}
        (-1)^{\vb*{\lambda}_{B\cap B_0}^{\sfT}\vb{x}_{B\cap B_0}}
        \widetilde w_B(\vb{x}_{B\cap B_0},\vb{z}_{B\cap D_0}).
    \label{eq:block-partial-fourier}
\end{align}
Applying \cref{cor:systematic-sop} to the shifted weights and zero target gives \cref{eq:block-fundamental-sop} with $G=T(\vb{A})$ and $g_B=\widehat{\widetilde w}_B$.
The transform acts within each complete block, so no coordinate-wise factorization of $w_B$ is needed.

Fix a union $Y$ of blocks and abbreviate
\begin{align*}
    \vb{H}_Y&:=\vb{A}[B_0\cap Y,D_0\cap\overline Y],
    &\mathcal V_Y&:=\im\vb{H}_Y,\\
    \vb{K}_Y&:=\vb{A}[B_0\cap\overline Y,D_0\cap Y],
    &\mathcal U_Y&:=\im\vb{K}_Y.
\end{align*}
In systematic row coordinates, the state space $I_Y$ is $\mathcal V_Y\oplus\mathcal U_Y$.
Indeed, its component on $B_0\cap Y$ must be given by the complementary free coordinates through $\vb{H}_Y$, while its component on $B_0\cap\overline Y$ is given by the free coordinates in $Y$ through $\vb{K}_Y$; the basis columns allow these components to vary independently.
Write the nonnegative table as $F_Y(\vb{h},\vb{v})$, with $\vb{h}\in\mathcal V_Y$ and $\vb{v}\in\mathcal U_Y$.
The SOP signature instead records $\vb{v}=\vb{K}_Y\vb{z}_{D_0\cap Y}$ and $\vb{H}_Y^{\sfT}\vb*{\lambda}_{B_0\cap Y}$.
The latter is equivalently a linear functional $\xi\in\mathcal V_Y^*$, acting by $\xi(\vb{h})=\vb*{\lambda}^{\sfT}\vb{h}$.
In these coordinates, write the DP table in \cref{eq:sop-dp-table} as $\widehat F_Y(\xi,\vb{v})$.
Then
\begin{align}
    \widehat F_Y(\xi,\vb{v})
    =2^{|B_0\cap Y|-\dim\mathcal V_Y}
        \sum_{\vb{h}\in\mathcal V_Y}
            (-1)^{\xi(\vb{h})}F_Y(\vb{h},\vb{v}).
    \label{eq:sop-positive-basis}
\end{align}
To see this, expand \cref{eq:block-partial-fourier} in \cref{eq:sop-dp-table} and combine its Fourier sign with the internal-edge sign.
The resulting exponent is $\vb*{\lambda}^{\sfT}\vb{h}$, where $\vb{h}=\vb{x}_{B_0\cap Y}+\vb{A}[B_0\cap Y,D_0\cap Y]\vb{z}_{D_0\cap Y}$.
For a representative $\vb*{\lambda}_\xi$ of $\xi$, character orthogonality gives
\begin{align*}
    \sum_{\substack{\vb*{\lambda}\in\FF_2^{B_0\cap Y}\\
        \vb{H}_Y^{\sfT}\vb*{\lambda}=\vb{H}_Y^{\sfT}\vb*{\lambda}_\xi}}
        (-1)^{\vb*{\lambda}^{\sfT}\vb{h}}
    =2^{|B_0\cap Y|-\dim\mathcal V_Y}
        (-1)^{\vb*{\lambda}_\xi^{\sfT}\vb{h}}
        \mathbf 1[\vb{h}\in\mathcal V_Y].
\end{align*}
This leaves exactly the summands of the nonnegative table and proves \cref{eq:sop-positive-basis}.

The boundary transform is invertible, and both tables have $2^{\dim\mathcal V_Y+\dim\mathcal U_Y}=2^{\rho_G(Y)}$ entries.
At the root, \cref{eq:sop-positive-basis} reduces to $\widehat F_{[d]}=2^rF_{[d]}$, consistent with \cref{eq:block-fundamental-sop}.
The nonnegative recurrence in \cref{eq:block-join} considers a subset of the child-state pairs, each contributing to one parent state, so it satisfies the join bound in \cref{eq:sop-dp-cost} on the same rank decomposition over $\calB$.
Its leaf initialization directly sums the original shifted nonnegative weights and requires no numerical Fourier transform.
Thus, the two recurrences are equivalent under the displayed changes of boundary basis, while the nonnegative realization uses only sums and products of nonnegative values.
\end{proof}

\section{Width bounds for Reed-Muller code families}
\label{app:rm-width-bounds}

This section proves the width bounds underlying the computational separation in \cref{subsec:rm-code-families}.
We first prove \cref{cor:rm-tn-separation}, then establish the bounds before check elimination in \cref{rem:rm-original-width}.

\subsection[Proof of the code-family width bounds]{Proof of \cref{cor:rm-tn-separation}}
\label{proof:rm-tn-separation}

\begin{proof}
We first verify the claimed distance of $\calQ_{m,L}$.
Each constituent code has distance three, so concatenation gives $d_{\min}\geq3^{L+1}$.
For distinct $\vb{u},\vb{v}\in V_m$, the three qubits labelled by $\vb{u},\vb{v},\vb{u}+\vb{v}$ support a weight-three logical $Z$ of $\calR_m$.
Their labels sum to zero, so this operator commutes with every $X$ stabilizer.
Its odd weight excludes it from $C_Z$, whose vectors all have even weight.
The same argument applies to the Steane code, which is $\calR_3$.
Recursively encoding these logical operators attains weight $3^{L+1}$, proving the claimed distance.

We now establish the width bounds.
For a set $U$ of physical qubits, write $\overline{U}$ for the complementary set of qubits.
For a binary subspace $C$, let $C_U$ consist of the vectors in $C$ supported in $U$, and set
\begin{align}
    \lambda_C(U):=\dim C-\dim C_U-\dim C_{\overline{U}}.
\end{align}
Rank-nullity gives $\lambda_C(U)=\lambda_{C^\perp}(U)$ and $\lambda_C(U)\leq\dim C$.
For the stabilizer space $S=\ker\vb{P}$ and $Y_U:=\bigcup_{i\in U}B_i$, \cref{eq:block-fundamental-rank} gives
\begin{align}
    \rho_{T(\vb{A})}(Y_U)
    &=\rank\vb{P}_{Y_U}+\rank\vb{P}_{\overline{Y_U}}-\rank\vb{P}\nonumber\\
    &=\dim S-\dim S_U-\dim S_{\overline{U}}.
    \label{eq:stabilizer-cut-connectivity}
\end{align}
For a CSS code, the last expression is $\lambda_{C_X}(U)+\lambda_{C_Z}(U)$.
In \cref{eq:rm-stabilizer-spaces}, $C_Z^\perp=C_X\oplus\Span\{\vb{1}\}$ has dimension $m+1$.
Consequently, every physical-qubit cut of $\calR_m$ satisfies
\begin{align}
    \rho_{T(\vb{A}_m)}(Y_U)
    =\lambda_{C_X}(U)+\lambda_{C_Z^\perp}(U)
    \leq m+(m+1)=2m+1.
    \label{eq:rm-all-cut-bound}
\end{align}
Any rooted binary tree on the physical-qubit blocks is therefore a rank-decomposition tree of $T(\vb{A}_m)$ over $\calB_m$ with width at most $2m+1$.

For the concatenation, first take the outer Steane hierarchy and replace each seven-child node by a binary tree.
At a displayed cut within one such refinement, descendant stabilizers are supported on one side or the other.
The six stabilizer generators at that node and at most two independent restrictions of ancestor generators contain all remaining contributions to \cref{eq:stabilizer-cut-connectivity}.
Indeed, ancestor restrictions lie in the span of the two fixed logical Pauli vectors on the node's encoded qubit.
Changing the ancestor-generator basis leaves at most two generators nonzero on that node.
Thus, the reduced Tanner graph of the outer code has a rank decomposition over its physical-qubit partition of width at most $8$, independently of $L$.

Replace each outer leaf by a binary tree on its inner Reed-Muller block.
A cut through whole inner blocks has the same value in \cref{eq:stabilizer-cut-connectivity} as the corresponding outer-code cut, since the inner stabilizers are supported on individual blocks.
For a cut within one inner block, choose a basis of the encoded outer stabilizers in which at most two generators restrict nontrivially to that block.
All other generators outside the inner stabilizer space are supported in the complement.
Adding one generator increases the right-hand side of \cref{eq:stabilizer-cut-connectivity} by at most one.
Together with \cref{eq:rm-all-cut-bound}, this gives the bound $\max\{8,2m+3\}$.
These basis choices establish cut-rank bounds; the specified matrices $\vb{P}_m$ and $\vb{P}_{m,L}$ are unchanged.
The trees are constructed directly from the concatenation hierarchy.

We next bound the treewidth of the TN graphs $G_m$ and $G_{m,L}$.
Let $q$ be the largest odd integer at most $m$, and write $q=2t+1$.
In one Reed-Muller block, retain only the check nodes for $Z$-stabilizer generators indexed by $t$-subsets $J\subseteq[q]$.
Also retain the variable nodes for the $x$-coordinates of qubits whose labels have support a $(t+1)$-subset $K\subseteq[q]$.
By the definition of $\vb{h}_J$, their adjacency is $J\subset K$.
The induced subgraph is the middle-levels graph $H_q$ of the $q$-dimensional cube.
It is $\Delta_q$-regular on $M_q$ vertices, where
\begin{align}
    M_q=2\binom{q}{t}=\Theta(2^q/\sqrt{q}),
    \qquad \Delta_q=t+1.
    \label{eq:rm-middle-parameters}
\end{align}
Its adjacency eigenvalues are $\pm1,\ldots,\pm(t+1)$, so its combinatorial Laplacian has second-smallest eigenvalue $1$~\cite{qiu_2009_middle_cubes}.
For every $U\subseteq V(H_q)$, the Laplacian variational bound gives
\begin{align}
    |\delta_{H_q}(U)|\geq\frac{|U|(M_q-|U|)}{M_q},
    \label{eq:rm-middle-expansion}
\end{align}
where $\delta_{H_q}(U)$ is the set of edges crossing the cut.

A graph of treewidth $a$ has a vertex separator $W$ of size at most $a+1$ such that every component after removing $W$ has at most $M_q/2$ vertices.
If $|W|<M_q/4$, a union $U$ of these components can be chosen with $M_q/4\leq|U|\leq M_q/2$.
Every edge leaving $U$ ends in $W$, so \cref{eq:rm-middle-expansion} gives $\Delta_q|W|\geq3M_q/16$.
If $|W|\geq M_q/4$, the same lower bound on $|W|$ holds immediately.
Hence,
\begin{align}
    \tw(H_q)+1\geq\frac{3M_q}{16\Delta_q}
    =\Omega(n_m/m^{3/2}).
    \label{eq:rm-middle-treewidth}
\end{align}
Both $G_m$ and $G_{m,L}$ contain this induced subgraph, proving their treewidth lower bounds.

For the upper bound on $\tw(G_m)$, take one bag per physical qubit containing all $n_m+1$ checks, its two variable nodes, and its weight node.
Connecting these bags in any tree gives width $n_m+3$.
For $G_{m,L}$, follow the outer concatenation tree.
At every outer node, put its six checks, all ancestor checks, and the two root logical checks in the bag.
At an inner-block leaf, put its $4n_m-1$ local vertices, the $6L$ ancestor stabilizer checks, and the two root logical checks in the bag.
Every edge is covered, and the bags containing a given vertex are connected.
The maximum bag size is at most $4n_m+6L+1$, proving the remaining upper bound.

Finally, applying Rank DP to the reduced SOP representation in \cref{cor:systematic-sop} gives polynomial-time likelihood evaluation by \cref{eq:sop-dp-cost}, since $m=\log_2(n_m+1)$ and $n_m\leq n$.
Gaussian elimination, the local weight transforms, and the structural preprocessing also take polynomial time for these explicit trees.
Each code encodes one logical qubit, so four likelihood evaluations suffice for MLD.
The tensor-network lower bound follows from \cref{thm:sop-tn-separation}.
Substituting $L=\lfloor\log_7 n_m\rfloor$ gives the asymptotic bounds as claimed.
\end{proof}

\subsection{Width bounds before check elimination}
\label{app:rm-original-width}

We prove the bounds in \cref{rem:rm-original-width} for the matrices and physical-qubit partitions defined in \cref{subsec:rm-code-families}.
The check nodes remain singleton parts.
All cut-ranks below are computed over $\FF_2$.

\paragraph{Lower bounds.}
Let $q=2t+1$ be the largest odd integer at most $m$.
As in the proof of \cref{cor:rm-tn-separation}, restricting the monomial checks to $t$-subsets of $[q]$ and the variable nodes to $x$ coordinates of qubits supported on $(t+1)$-subsets gives an induced middle-levels graph $H_q$.
Its vertex count $M_q$ and maximum degree $\Delta_q$ are given in \cref{eq:rm-middle-parameters}.
The spectral bound \cref{eq:rm-middle-expansion} implies that every cut with $M_q/3\leq|U|\leq2M_q/3$ has at least $2M_q/9$ crossing edges.

Write $r=\rho_{H_q}(U)$ and choose $r$ basis rows of $\vb{A}_{H_q}[U,\overline U]$.
Their supports cover at most $r\Delta_q$ columns, because each row contains at most $\Delta_q$ nonzero entries.
Every nonzero column is covered: a column vanishing in every basis row vanishes in their entire span.
Each covered column contains at most $\Delta_q$ nonzero entries.
Therefore,
\begin{align}
    |\delta_{H_q}(U)|\leq\Delta_q^2\rho_{H_q}(U).
    \label{eq:middle-cut-rank-lower}
\end{align}
Every rooted binary tree with $M_q$ leaves displays a cut with between $M_q/3$ and $2M_q/3$ leaves on one side.
To see this, descend into a child with more than $2M_q/3$ leaves until neither child has that many, and then take the larger child.
Applying \cref{eq:middle-cut-rank-lower} to this cut in any rank decomposition yields
\begin{align}
    \rw(H_q)\geq\frac{2M_q}{9\Delta_q^2}
    =\Omega(2^q/q^{5/2})
    =\Omega(n_m/m^{5/2}).
    \label{eq:middle-rank-width-lower}
\end{align}

Now restrict any rank decomposition of $T(\vb{P}_m)$ over $\calC_m$ to the retained vertices of $H_q$.
Delete leaves with no retained vertex and suppress nodes with only one child.
Each retained check part remains a singleton.
Each retained physical-qubit part contains exactly one remaining $x$ coordinate and therefore also becomes a singleton.
Every displayed cut matrix in the resulting decomposition of $H_q$ is a submatrix of an original cut matrix.
Its rank cannot increase.
Thus, $\rw(T(\vb{P}_m);\calC_m)\geq\rw(H_q)$.
The same restriction inside any one inner block proves the lower bound for $\rw(T(\vb{P}_{m,L});\calC_{m,L})$.
Since $H_q$ is also an induced subgraph of $G_m$ and $G_{m,L}$, ordinary rank-width satisfies the same lower bounds.
The argument uses only inner stabilizer checks and does not depend on the two root logical checks.

\paragraph{Upper bounds.}
For $T(\vb{P}_m)$, place all $n_m+1$ check leaves in one root subtree and all physical-qubit leaves in the other.
Refine each subtree into any rooted binary tree.
Every nonroot node then contains only check parts or only physical-qubit parts.
Its cut matrix is a submatrix of $\vb{P}_m$ or its transpose, up to zero rows and columns.
Hence, its cut-rank is at most $\rank\vb{P}_m=n_m+1$.

For $T(\vb{P}_{m,L})$, follow the concatenation hierarchy.
At each outer node, combine its seven child clusters and six local check nodes using a binary tree.
At the root, also include the two logical check nodes.
Consider a cut within a nonroot refinement that includes $s$ of the six local checks.
Descendant checks in the selected child clusters have no neighbors outside those clusters.
The block of the cut matrix from selected checks to outside coordinates therefore has rank at most $s$.
In the reverse block, the remaining local checks contribute rank at most $6-s$.
All ancestor checks restrict to the span of the two fixed logical Pauli vectors on the current cluster and increase the rank of this block by at most two.
These two rectangular blocks occupy disjoint rows and columns.
Thus, the cut-rank is at most $s+(6-s)+2=8$.
At the root, there are eight local checks and no ancestor checks, giving the same bound.

Within an inner block, use the preceding separated-subtree construction for its $n_m-1$ stabilizer checks and its physical-qubit parts.
The internal cut-rank is at most $n_m-1$.
The restrictions of all outer checks to that block span at most two dimensions, so they increase this rank by at most two.
Consequently, the full decomposition has width at most $\max\{8,n_m+1\}$.

Finally, these decompositions give ordinary rank decompositions of $G_m$ and $G_{m,L}$.
Put each weight node into its physical-qubit part.
Its two incident edges stay within that part, so all displayed cut-ranks are unchanged.
Refine each resulting three-vertex part into singleton leaves.
Every new cut separates at most two vertices from the rest of the graph and has rank at most two.
The stated upper bounds are therefore unchanged.

\section{Numerical accuracy and implementation}
\label{app:numerical-accuracy}

This section develops numerical error bounds for the nonnegative DP in \cref{subsec:nonnegative-accuracy} and describes the implementation and checks used in the experiments.

\subsection{Roundoff-error analysis}
\label{app:roundoff-error}

We now quantify the roundoff error accumulated by the nonnegative DP in \cref{prop:nonnegative-realization}.
Since all summands are nonnegative, relative-error bounds for individual products of input weights also bound the relative error of their sum.

Let $\varepsilon_{\mathrm{mach}}$ be the unit roundoff, and define $\gamma_j:=j\varepsilon_{\mathrm{mach}}/(1-j\varepsilon_{\mathrm{mach}})$ for $j\varepsilon_{\mathrm{mach}}<1$.
We use the standard round-to-nearest model, in which each computed addition or multiplication equals its exact result times $1+\varepsilon$, with $|\varepsilon|\leq\varepsilon_{\mathrm{mach}}$.
Assume that each sum is evaluated by sequential accumulation or by a binary addition tree of no greater depth.

To track error accumulation through the tree, let $\delta_Y$ bound the number of rounding factors multiplying any product of input weights in an expanded entry of $F_Y$.
At a leaf block $B$, each entry sums $g_B$ weights, so we set
\begin{align}
    g_B&:=2^{|B|-\rank\vb{M}_B},
    &\delta_B&:=g_B.
    \label{eq:block-leaf-depth}
\end{align}
This bound allows one rounding factor per addition, including an initial addition into zero.

At an internal node, write $Y_1,Y_2$ for the child sets and $Y=Y_1\sqcup Y_2$.
By \cref{prop:nonnegative-realization}, each parent entry sums $T_Y$ products, and we propagate the error count by
\begin{align}
    T_Y&:=2^{\mu_Y},
    &\delta_Y&:=\delta_{Y_1}+\delta_{Y_2}+T_Y+1.
    \label{eq:block-error-depth}
\end{align}
Each product inherits the rounding factors from both child entries and acquires one additional factor from multiplication.
The sum contributes at most $T_Y$ further factors.
Let $\delta:=\delta_{[d]}$ be the count at the root, obtained from \cref{eq:block-leaf-depth,eq:block-error-depth}.
The following proposition converts this count into an error bound for the computed partition function.

\begin{samepage}
\begin{proposition}[Relative roundoff-error bound]
\label{prop:roundoff-error}
Evaluate the nonnegative DP in \cref{prop:nonnegative-realization} under the floating-point model and summation rule above.
Assume that the input weights are represented exactly and that binary linear algebra and indexing are exact.
Assume also that power-of-two rescaling and exponent bookkeeping are exact, and that no overflow, underflow, or discarded positive terms invalidate this model.
Let $Z:=Z_{\vb{M}}(\vb{t};\vb{w})$, and let $\widehat{Z}$ denote the exact real value represented by the computed output, including any recorded power-of-two scale.
If $\delta\varepsilon_{\mathrm{mach}}<1$, then
\begin{align}
    |\widehat{Z}-Z|\leq\gamma_\delta Z.
    \label{eq:block-forward-bound}
\end{align}
If, in addition, $\gamma_\delta<1$, the exact partition-function value lies in the interval
\begin{align}
    Z\in[L,U],
    \qquad L:=\frac{\widehat{Z}}{1+\gamma_\delta},
    \qquad U:=\frac{\widehat{Z}}{1-\gamma_\delta}.
    \label{eq:block-scalar-interval}
\end{align}
\end{proposition}
\end{samepage}

\begin{proof}
\label{proof:block-certificate}
A product of at most $j$ rounding factors $1+\varepsilon_i$, with $|\varepsilon_i|\leq\varepsilon_{\mathrm{mach}}$ and $j\varepsilon_{\mathrm{mach}}<1$, differs from one by at most $\gamma_j$~\cite{higham_2002_accuracy}.
Expand each computed entry into monomials in the exact input weights and attach the rounding factors introduced by the additions and multiplications.
A leaf monomial receives at most $\delta_B$ factors.
At a join, each product inherits the factors from both child monomials and one factor from multiplication.
The subsequent summation contributes at most $T_Y$ factors, which gives \cref{eq:block-error-depth}.
Exact power-of-two rescaling adds no rounding error.
Since all exact monomials are nonnegative, their relative-error bounds imply the same bound for their sum.
Induction gives \cref{eq:block-forward-bound}, including $Z=0$, for which $\widehat{Z}=0$.
Rearranging $(1-\gamma_\delta)Z\leq\widehat{Z}\leq(1+\gamma_\delta)Z$ gives \cref{eq:block-scalar-interval}.
\end{proof}

The bound concerns the exact real values of the stored weights and excludes error in forming those weights or uncertainty in the noise model.
The interval endpoints must be evaluated exactly or rounded outwards to preserve containment.
These are numerical error bounds, not statistical confidence intervals.
The consequences for posterior probabilities and decoding decisions are given in \cref{app:block-output-bounds}, followed by implementation details in \cref{app:block-arithmetic-policy}.

\subsection{Posterior normalization and excess failure probability bounds}
\label{app:block-output-bounds}

Let $Z_i:=Z_{\vb{M}}(\vb{t}_i;\vb{w})$, $i=1,\ldots,K$, be evaluated with a common rounding-factor bound $\delta$ under the assumptions of \cref{prop:roundoff-error}.
If $S:=\sum_i Z_i>0$, define $\pi_i:=Z_i/S$.
When the queries cover all logical labels at a fixed syndrome or detector vector, these are the corresponding logical posterior probabilities.

The implemented normalization sums the computed joint values and divides each value by that sum.
Set $G:=\gamma_{\delta+K}$ and assume $G<1$, with the same normal-arithmetic conditions also holding for the sum and divisions.
The denominator has relative error at most $G$: expanding its positive sum adds at most $K$ rounding factors to the evaluation of each summand, whose rounding-factor bound is $\delta$.
Each division adds at most $\varepsilon_{\mathrm{mach}}$.
Thus, for every positive $\pi_i$, its computed ratio lies between $(1-G)(1-\varepsilon_{\mathrm{mach}})/(1+G)$ and $(1+G)(1+\varepsilon_{\mathrm{mach}})/(1-G)$ times $\pi_i$.
Zero entries remain zero, giving the conservative bound
\begin{align}
    \|\widehat{\boldsymbol\pi}-\boldsymbol\pi\|_1
       \leq \frac{(1+G)(1+\varepsilon_{\mathrm{mach}})}{1-G}-1.
    \label{eq:block-posterior-bound}
\end{align}
A zero total mass has no conditional posterior and is reported separately.

For the decoding interpretation, assume that the queries include all competing logical labels at the fixed observed vector.
Let $[L_i,U_i]$ be the joint intervals from \cref{eq:block-scalar-interval}.
For a chosen label $a$ and $K\geq2$, its excess conditional failure probability is $\max_i\pi_i-\pi_a$ and satisfies
\begin{align}
    \max_i\pi_i-\pi_a
       \leq\min\!\left\{1,
       \frac{\max\{0,\max_{j\ne a}U_j-L_a\}}{\sum_i L_i}\right\},
    \label{eq:block-regret-bound}
\end{align}
provided $\sum_iL_i>0$.
Indeed, $Z_j-Z_a\leq U_j-L_a$ and $S\geq\sum_iL_i$; the additional bound by one follows because the excess conditional failure probability is a difference of probabilities.
With only one label, the excess conditional failure probability is zero.
Exact rational evaluation of endpoints and outward rounding of displayed bounds avoid an additional, unaccounted interval-rounding error.

\subsection{Implementation and numerical checks}
\label{app:block-arithmetic-policy}

We now describe the guarded CPU evaluator that supplies likelihood intervals under the arithmetic and scaling assumptions of \cref{prop:roundoff-error}.
The FP64 timing policy for the Reed-Muller experiments is specified in \cref{app:rm-capacity}.

\paragraph{Reductions and storage.}
At joins, the CPU implementation either accumulates products into each parent entry or reduces a vector of its products.
Independent parent entries may be processed in chunks, changing workspace without changing their individual sums.
The depth bound permits any binary addition tree with at most the depth of a sequential reduction of the same terms.
The leaf bound $\delta_B=g_B$ deliberately allows one more rounding factor than needed by a sequential sum of exact inputs.
A zero-column block is evaluated by summing its stored categorical entries, rather than replacing their sum by a presumed exact normalization.
Integer-map caches trade storage for reuse; their construction and data movement, like chunk workspace, are additional to the scalar-operation bound.

\paragraph{Arithmetic guards and range.}
This evaluator guards floating-point operations against underflow, overflow, invalid values, and division by zero, and explicitly checks array-level nonnegativity and positive subnormals.
It does not inspect every internal node of a library reduction.
Small rounding and subnormal probes reject some unsupported arithmetic modes; these probes are guards, not a proof that arbitrary hardware or libraries satisfy the proposition's assumptions.
A failed bound, range check, or resource limit causes a recorded rejection or an allowed wider-precision attempt.
The actual significand and exponent range of the platform's format are recorded; in particular, \texttt{longdouble} is not presumed to be binary128.
Exhaustion returns an unresolved result rather than an unchecked likelihood.
A mantissa and integer exponent can represent a value below the format's unscaled range, although a common-exponent vector can still fail if its entries have excessive dynamic range.

\paragraph{Differentiation.}
The earlier capacity-learning implementation applies floating-point reverse-mode differentiation to the positive sums and products.
The chosen integer exponents are detached from differentiation, with their power-of-two scales held fixed and compensated at the output.
For any fixed such scales, the underlying algebraic expression is the same polynomial.
Optional checkpointing recomputes joins during the backward pass and trades work for saved intermediates.
These implementation choices do not extend the likelihood error bound to floating-point derivatives or automatically to graphics processing unit (GPU) arithmetic; gradient checks remain separate validation evidence.
The three-parameter circuit study instead propagates forward sensitivities through the same sums and products, sharing the primal power-of-two scales; its independent checks and resource policy are described in \cref{app:circuit-learning-revision}.

\section{A surface-code example}
\label{app:worked-example}

We illustrate the SOP representation and Rank DP for the $[[9,1,3]]$ rotated surface code in \cref{fig:surface-code} under independent $X$ errors~\cite{tomita_2014_lowdistance}.
Let $\vb{H}_X,\vb{H}_Z\in\FF_2^{4\times9}$ have the supports of its four $X$-type and four $Z$-type generators as rows, in the order shown in the figure.
Both matrices have rank four and satisfy $\vb{H}_X\vb{H}_Z^{\sfT}=\vb{0}$, so the stabilizer matrix $\vb{H}:=\left(\begin{smallmatrix}\vb{H}_X&\vb{0}\\\vb{0}&\vb{H}_Z\end{smallmatrix}\right)$ defines a code with one logical qubit.
We choose logical representatives $\overline{X}:=X_1X_4X_7$ and $\overline{Z}:=Z_1Z_2Z_3$, which commute with every stabilizer and anticommute with each other.
Write $C_X:=\im\vb{H}_X^{\sfT}$ and $C_Z:=\im\vb{H}_Z^{\sfT}$, and let $\overline{\vb{x}},\overline{\vb{z}}\in\FF_2^9$ be the binary supports of $\overline{X},\overline{Z}$.
Both logical cosets $\overline{\vb{x}}+C_X$ and $\overline{\vb{z}}+C_Z$ have minimum Hamming weight three, which implies that the code has distance three.

\begin{figure}[htbp]
\centering
\begin{minipage}[c]{0.40\textwidth}
\centering
\begin{tikzpicture}[x=0.9cm,y=0.9cm,
 q/.style={circle,draw,fill=white,minimum size=5mm,inner sep=0pt},
 xp/.style={draw=blue!55!black,fill=blue!14},
 zp/.style={draw=orange!70!black,fill=orange!18}]
\path[xp] (0,1.2) rectangle (1.2,2.4);
\path[xp] (1.2,0) rectangle (2.4,1.2);
\path[zp] (1.2,1.2) rectangle (2.4,2.4);
\path[zp] (0,0) rectangle (1.2,1.2);
\path[xp] (1.2,2.4)--(1.8,3.05)--(2.4,2.4)--cycle;
\path[xp] (0,0)--(0.6,-0.65)--(1.2,0)--cycle;
\path[zp] (0,2.4)--(-0.65,1.8)--(0,1.2)--cycle;
\path[zp] (2.4,1.2)--(3.05,0.6)--(2.4,0)--cycle;
\foreach \lab/\xx/\yy in {1/0/2.4,2/1.2/2.4,3/2.4/2.4,4/0/1.2,5/1.2/1.2,6/2.4/1.2,7/0/0,8/1.2/0,9/2.4/0}
    \node[q] at (\xx,\yy) {\lab};
\node[blue!55!black] at (0.6,1.8) {$X$};
\node[blue!55!black] at (1.8,0.6) {$X$};
\node[orange!70!black] at (1.8,1.8) {$Z$};
\node[orange!70!black] at (0.6,0.6) {$Z$};
\end{tikzpicture}
\end{minipage}\hfill
\begin{minipage}[c]{0.57\textwidth}
\centering
\begin{tabular}{ccc}
\toprule
Row & $X$-type generator & $Z$-type generator\\
\midrule
1 & $X_2X_3$ & $Z_1Z_4$\\
2 & $X_1X_2X_4X_5$ & $Z_2Z_3Z_5Z_6$\\
3 & $X_5X_6X_8X_9$ & $Z_4Z_5Z_7Z_8$\\
4 & $X_7X_8$ & $Z_6Z_9$\\
\bottomrule
\end{tabular}
\end{minipage}
\caption{The $[[9,1,3]]$ rotated surface code. Data qubits are numbered row by row. Blue and orange regions support $X$- and $Z$-type stabilizers; boundary triangles support weight-two generators. The table fixes the row order of $\vb{H}_X$ and $\vb{H}_Z$.}
\label{fig:surface-code}
\end{figure}

\subsection{Reduced Tanner graph and rank decomposition}

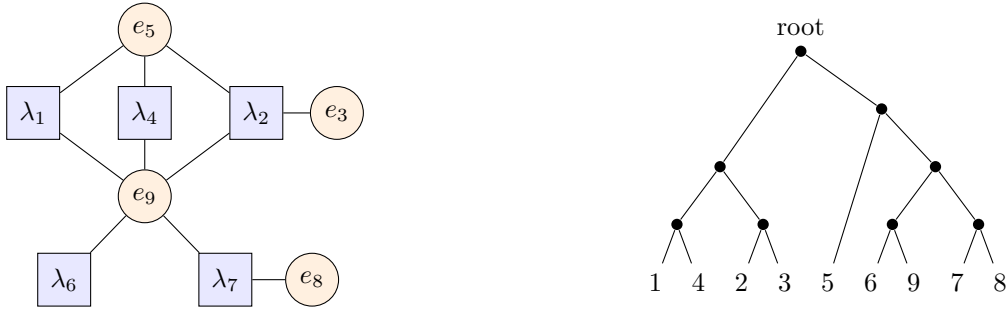
\begin{figure}[tb]
\centering
\begin{minipage}[c]{0.46\textwidth}
\centering
\begin{tikzpicture}[x=0.82cm,y=0.85cm,font=\small,
 b/.style={rectangle,draw,fill=blue!9,minimum size=7mm,inner sep=1pt},
 d/.style={circle,draw,fill=orange!12,minimum size=7mm,inner sep=1pt}]
\node[b] (b1) at (0,1.8) {$\lambda_1$};
\node[b] (b2) at (3.6,1.8) {$\lambda_2$};
\node[b] (b4) at (1.8,1.8) {$\lambda_4$};
\node[b] (b6) at (0.5,-0.8) {$\lambda_6$};
\node[b] (b7) at (3.1,-0.8) {$\lambda_7$};
\node[d] (d3) at (4.9,1.8) {$e_3$};
\node[d] (d5) at (1.8,3.1) {$e_5$};
\node[d] (d8) at (4.5,-0.8) {$e_8$};
\node[d] (d9) at (1.8,0.5) {$e_9$};
\draw (b1)--(d5) (b1)--(d9)
            (b2)--(d3) (b2)--(d5) (b2)--(d9)
            (b4)--(d5) (b4)--(d9)
            (b6)--(d9) (b7)--(d8) (b7)--(d9);
\end{tikzpicture}
\end{minipage}\hfill
\begin{minipage}[c]{0.51\textwidth}
\centering
\begin{tikzpicture}[x=0.57cm,y=0.85cm,font=\small,
 inner/.style={circle,fill,inner sep=1.5pt}]
\foreach \lab/\xx in {1/0,4/1,2/2,3/3,5/4,6/5,9/6,7/7,8/8}
    \node (q\lab) at (\xx,0) {\lab};
\node[inner] (n14) at (0.5,0.9) {};
\node[inner] (n23) at (2.5,0.9) {};
\node[inner] (n69) at (5.5,0.9) {};
\node[inner] (n78) at (7.5,0.9) {};
\node[inner] (nl) at (1.5,1.8) {};
\node[inner] (n6789) at (6.5,1.8) {};
\node[inner] (nr) at (5.25,2.7) {};
\node[inner,label=above:root] (root) at (3.375,3.6) {};
\draw (n14)--(q1) (n14)--(q4) (n23)--(q2) (n23)--(q3)
            (n69)--(q6) (n69)--(q9) (n78)--(q7) (n78)--(q8)
            (nl)--(n14) (nl)--(n23)
            (n6789)--(n69) (n6789)--(n78)
            (nr)--(q5) (nr)--(n6789)
            (root)--(nl) (root)--(nr);
\end{tikzpicture}
\end{minipage}
\caption{Left: the reduced Tanner graph $G=T(\vb{A})$ from \cref{eq:surface-systematic}.
Squares carry Fourier variables $q_b=\lambda_b$ for $b\in B_0$, and circles carry free error bits $q_j=e_j$ for $j\in D_0$.
Right: a rooted rank-decomposition tree of width one, whose nonroot cuts each have a one-bit boundary signature.
Both logical classes use this graph and tree; only the signs of the transformed singleton weights differ.}
\label{fig:surface-graphs}
\end{figure}

\begin{example}[Logical-class likelihoods on the reduced Tanner graph]
\label{ex:surface-sop}
For the surface code in \cref{fig:surface-code}, take $p_I=0.9$, $p_X=0.1$, and $p_Y=p_Z=0$ independently on each qubit.
The $Z$-error coordinates are fixed to zero, so we retain only $\vb{e}=\vb{x}\in\FF_2^9$ and use the singleton partition $\calB=\{\{i\}:i\in[9]\}$.
The local weight tables are $w_i:=w_{\{i\}}=(0.9,0.1)$.
Consider the syndrome $\vb{s}=(0,0,0,0,0,1,1,0)^{\sfT}$, with the four $X$ checks followed by the four $Z$ checks in the order of \cref{fig:surface-code}.
We may choose the compatible representative $t:=X_5$; its binary $X$ component $\vb{r}_X\in\FF_2^9$ is the unit vector at qubit five.
Only the classes represented by $X_5\overline{X}^{a}$, $a\in\{0,1\}$, have nonzero probability.
Let $\vb*{\ell}_a$ be the binary vector of $\overline{X}^{a}$ and write $P_a^{\cc}:=P^{\cc}(\vb{s},\vb*{\ell}_a)$ for the corresponding joint likelihood.
The four $Z$ checks constrain the syndrome, while the parity against $\overline{\vb{z}}$ distinguishes the two possible $X$ logical classes.
Together they give
\begin{align}
    \vb{K}_X:=\begin{pmatrix}\vb{H}_Z\\\overline{\vb{z}}^{\sfT}\end{pmatrix}
    =\begin{pmatrix}
        1&0&0&1&0&0&0&0&0\\
        0&1&1&0&1&1&0&0&0\\
        0&0&0&1&1&0&1&1&0\\
        0&0&0&0&0&1&0&0&1\\
        1&1&1&0&0&0&0&0&0
    \end{pmatrix},
    \qquad \vb{t}_a:=\begin{pmatrix}0\\1\\1\\0\\a\end{pmatrix}.
    \label{eq:surface-K}
\end{align}
The matrix has rank five and kernel $C_X$.
Its last target bit is $a$ because $\overline{\vb{z}}^{\sfT}\vb{r}_X=0$ and $\overline{\vb{z}}^{\sfT}\overline{\vb{x}}=1$.
Consequently, $\vb{K}_X\vb{e}=\vb{t}_a$ selects exactly the class of $X_5\overline{X}^{a}$, giving $P_a^{\cc}=Z_{\vb{K}_X}(\vb{t}_a;\vb{w})$ as in \cref{eq:inference-primitive}.

Choose the column basis $B_0=\{1,2,4,6,7\}$ and its complement $D_0=\{3,5,8,9\}$.
As in \cref{subsec:systematic-sop}, the constraint becomes $\vb{e}_{B_0}+\vb{A}\vb{e}_{D_0}=\vb*{\tau}_a$, where
\begin{align}
    \vb{A}=\begin{pmatrix}
        0&1&0&1\\
        1&1&0&1\\
        0&1&0&1\\
        0&0&0&1\\
        0&0&1&1
    \end{pmatrix},\qquad
    \vb*{\tau}_0=\begin{pmatrix}1\\1\\1\\0\\0\end{pmatrix},\qquad
    \vb*{\tau}_1=\begin{pmatrix}0\\1\\0\\0\\1\end{pmatrix},
    \label{eq:surface-systematic}
\end{align}
with rows and columns ordered increasingly by $B_0$ and $D_0$, respectively.
Let $G:=T(\vb{A})$, with bipartition $(B_0,D_0)$, as in \cref{subsec:systematic-sop}.
There are four free error bits, so each target admits $2^4=16$ errors.
The graph in \cref{fig:surface-graphs} shows the stabilizer changes relating them.
For example, vertex $3$ has the single neighbor $2$: flipping $e_3$ forces $e_2$ to flip, giving the stabilizer $X_2X_3$.
Vertex $5$ has neighbors $1,2,4$, giving $X_1X_2X_4X_5$.
The analogous changes at vertices $8$ and $9$ give $X_7X_8$ and $X_1X_2X_4X_6X_7X_9$, respectively.
These four independent changes generate $C_X$; the last is a product of three generators listed in \cref{fig:surface-code}.
Thus, the graph describes the assignments summed in each logical-class likelihood.

For $w_i=(0.9,0.1)$, the singleton weights in \cref{eq:systematic-sop-weights} are
\begin{align}
    g_b^{(a)}(0)=1,\qquad
    g_b^{(a)}(1)=(-1)^{(\tau_a)_b}\frac{4}{5}\quad(b\in B_0),
    \qquad g_j^{(a)}=w_j\quad(j\in D_0),
    \label{eq:surface-sop-weights}
\end{align}
and
\begin{align}
    P_a^{\cc}=\frac{1}{32}\calS(G;g^{(a)}).
    \label{eq:surface-SOP}
\end{align}
Changing $a$ flips the unary signs at vertices $1,4,7$, the support of $\overline{X}$, while leaving the graph unchanged.
The two SOP evaluations give $32P_0^{\cc}=1.480591872$ and $32P_1^{\cc}=0.621808128$.
Their comparison selects the class $a=0$, giving the MLD recovery $X_5$.
The next example explains which terms in these sums can be combined during their evaluation.
\end{example}

\begin{example}[Rank DP on a width-one decomposition]
\label{ex:surface-width}
For the surface-code graph, take the cut $Y=\{1,2,3,4\}$.
Its crossing-edge sign is
\begin{align}
    (-1)^{(\lambda_1+\lambda_2+\lambda_4)(e_5+e_9)}.
\end{align}
Thus, the DP table in \cref{eq:sop-dp-table} groups assignments to $\lambda_1,\lambda_2,\lambda_4,e_3$ into two entries, indexed by the parity $\lambda_1+\lambda_2+\lambda_4$.
This remains true for both logical classes, since changing the class changes the unary weights but not the crossing edges.

The rooted binary tree
\begin{align}
    \calT=\bigl(((1,4),(2,3)),\ (5,((6,9),(7,8)))\bigr)
    \label{eq:surface-tree}
\end{align}
has nontrivial node cuts, up to complementation, with sides $\{1,4\}$, $\{2,3\}$, $\{1,2,3,4\}$, $\{6,9\}$, $\{7,8\}$, and $\{6,7,8,9\}$.
For each such side, the nonzero rows of the cut matrix are identical, so the cut-rank is one.
Every singleton cut also has rank one, while the root cut has rank zero.
Thus, $\calT$ is a rank-decomposition tree of $G$ over $\calB$ with width one, in the rooted convention of \cref{subsec:rank-width}.
Consequently, Rank DP uses at most two entries per table and combines at most four child-signature pairs at each join.
The root entry equals $32P_a^{\cc}$ by \cref{eq:surface-SOP}.
The same tree supports the nonnegative realization in \cref{subsec:nonnegative-accuracy}, with the same table sizes.
\end{example}

\section{Additional experimental protocols and evidence}
\label{sec:supplemental-experiments}

Reproducing the numerical comparisons in \cref{sec:numerical-experiments} requires the noise models, requested outputs, and computational limits to be specified for each method.
This appendix provides these details, together with the numerical error bounds and independent checks used to validate the results.
The protocols for circuit noise learning, postselection, and decoder assessment in \cref{sec:applications} follow the decoding experiments.
For these applications, we distinguish errors in likelihood evaluation from statistical uncertainty due to finite sampling.

\subsection{Models, outputs, and accounting conventions}
\label{app:protocols}

\paragraph{Noise models and outputs.}
Code-capacity models use independent physical-qubit Pauli blocks, with input probabilities ordered as $(I,X,Y,Z)$.
Within a block, coordinates $(x_i,z_i)$ use the internal order $(00,10,01,11)=(I,X,Z,Y)$; the input routines apply this permutation explicitly.
In each encoded syndrome or logical label, the first constraint row determines the most significant bit.
For a code encoding $k$ qubits, a full output contains $4^k$ logical-class likelihoods $v_\ell(s)=P(s,\ell)$; their sum is $P(s)$, and their normalization gives the logical posterior distribution.
Construction checks verify CSS commutation, check ranks, and logical operators.
Code distances follow from the cited constructions and, for the geometric examples, their string and sheet representatives.

The surface circuits use rotated Z memory, and the color circuits use Stim's \texttt{memory\_xyz} construction.
The pair $(d,r)$ specifies the nominal code distance and number of syndrome-extraction rounds.
Gate, reset, and measurement rates are $.001$ unless a rate sweep or noise-learning setting is stated; the additional idle-noise rate is zero.
The circuit output is the pair $P(s,L_0=0),P(s,L_0=1)$ for the final measured memory observable.
For these circuits, conversion to a detector error model (DEM) is exact and retains each event's complete detector and observable support.
The conversion flattens loops and disables both error decomposition and the disjoint-error approximation.
Thus, an event affecting multiple detectors remains a single joint event.
Rank DP also accepts larger prescribed fault blocks; a $16$-outcome two-qubit block is included in the small-instance checks.

\paragraph{Sampling and arithmetic inputs.}
Each paired performance comparison uses the same stored FP64 weights and queries.
The roundoff bounds treat those stored weights as exact inputs.
The circuit timing cohorts consist of the all-zero detector outcome, called the quiet syndrome, followed by the first $31$ distinct nonzero detector outcomes from seeded sampling.
This selection uses detector outcomes alone.
The IID capacity and application cohorts retain every physical draw,
including repeated syndromes and their multiplicities.

\begin{table}[!htbp]
\centering\small
\caption{Outputs, numerical requirements, and timed phases. The
precision diagnostic reports all three accuracy criteria separately.
Independent reference calculations and checks are recorded separately
from the timing measurements in the first and third rows.}
\label{tab:experiment-conventions}
\begin{tabular}{>{\raggedright\arraybackslash}p{.18\linewidth}>{\raggedright\arraybackslash}p{.20\linewidth}>{\raggedright\arraybackslash}p{.27\linewidth}>{\raggedright\arraybackslash}p{.23\linewidth}}
\toprule
Study & Output & Accuracy requirement & Timed phases \\
\midrule
Code capacity & All $4^k$ masses & Vector relative $\ell_1$ error
$\leq10^{-10}$ & Reused evaluation \\
Circuit precision & Two observable masses & Vector error, sector
errors, and ML separation & Accuracy diagnostic \\
Circuit batches & Two observable masses and ML label & Each sector's
relative error $\leq10^{-8}$ and separated winner intervals & Startup,
compilation, evaluation, own checks, recording \\
\bottomrule
\end{tabular}
\end{table}

\paragraph{Resource accounting.}
Planning searches for a rank decomposition or a TN contraction path.
Compilation builds reusable maps and numerical routines, which are then applied to the weights and queries during evaluation.
Returned values are accepted after their numerical checks.
Failed checks, timeouts, and unattempted queries are recorded separately.
Numeric-array forecasts count live array storage.
Process resident set
size (RSS) also includes interpreter, allocator, and compiled-map costs;
address-space limits constrain virtual memory.
One GiB is $2^{30}$ bytes, one MiB is $2^{20}$ bytes, and one MB is $10^6$ bytes.
Operation counts use each representation's
stated convention and guide its plan selection.
Workers run on pinned cores of a shared host, with one numerical thread
and repeated paired measurements.
The software environment is given in
\cref{app:evidence-index}.

\subsection{A signed-arithmetic cancellation diagnostic}
\label{app:binary-stability}

A fixed independent-X-noise example on the $[[162,2,9]]$ toric code
shows how cancellation can affect an ML decision.
To specify the example,
index the two edge families by $h_{rc}=9r+c$ and $v_{rc}=81+9r+c$,
where $r,c\in\{0,\ldots,8\}$ and coordinate arithmetic is modulo nine.
The binary syndrome check at $(r,c)$ has support
$\{h_{rc},h_{r,c+1},v_{rc},v_{r-1,c}\}$; retain all checks except
$(8,8)$.
The two logical-parity rows have supports
$\{h_{r0}:r=0,\ldots,8\}$ and $\{v_{0c}:c=0,\ldots,8\}$, in that order.
The query syndrome is generated by X errors on the following edges:
\[
    \begin{split}
    \{&10,19,20,38,41,42,44,46,49,60,61,79,89,93,94,\\
    &95,96,102,107,109,123,126,127,128,146,152,158,159,161\}.
    \end{split}
\]
This $29$-error pattern was drawn at $p=.12$ and reevaluated at
$p=10^{-4}$ to probe small likelihoods.
In logical-label order
$00,01,10,11$, the likelihoods computed by nonnegative Rank DP are approximately
\[
    (1.30615\!\times\!10^{-99},\;1.21129\!\times\!10^{-96},\;
    3.29534\!\times\!10^{-96},\;1.03663\!\times\!10^{-99}).
\]
Their sum is $4.50898\times10^{-96}$, and the corresponding numerical intervals certify $10$ as the unique maximizing label.
Evaluating the signed SOP representation with Rank DP instead returns
\[
    (-2.21366\!\times\!10^{-94},\;-1.84342\!\times\!10^{-92},\;
    -3.94904\!\times\!10^{-93},\;1.65716\!\times\!10^{-93}),
\]
which has three negative entries and selects $11$.
The nonnegative calculation treats the stored FP64 value of $p$ as exact and forms $1-p$ in its evaluation precision.
Its relative vector error bound is $6.27\times10^{-14}$; an independently evaluated nonnegative TN agrees to relative $\ell_1$ difference $1.90\times10^{-15}$.

At $p=.5$, the diagnostic also includes exactly tied sectors.
At $p=.49$, some validated intervals overlap; the output then records
the possible maximizing labels and a bound on the excess conditional failure probability
as in \cref{subsec:output-guarantees}.

\subsection{Numerical error in multiprecision Fourier evaluation}
\label{app:fourier-certificate}

The signed Fourier TN uses an absolute error bound obtained from the
sum of absolute input-weight monomials.
For a fixed query, write
$Z:=Z_{\vb{M}}(\vb{t};\vb{w})$ and let $\widehat Z$ denote its
computed value.
For the singleton Bernoulli DEM models, let
$w_{i0},w_{i1}\geq0$ be the two exact stored FP64 weights of block $i$,
and define
\begin{equation}
    B=\prod_{i=1}^{N}(w_{i0}+w_{i1}),
    \label{eq:fourier-expansion-norm}
\end{equation}
where $N$ is the number of blocks.
For each fixed detector and logical
character, expanding the product of factors $w_{i0}\pm w_{i1}$ gives
signed input-weight monomials whose absolute sum is
$\sum_{\vb{b}\in\FF_2^N}\prod_iw_{i b_i}=B$.
The detector-character normalization $2^{-m}$ averages these expansions
over $2^m$ characters, where $m$ is the number of detector constraints.
Any normalized inverse logical Walsh transform is
another signed average.
The triangle inequality therefore bounds the
absolute expansion sum of each final scalar by $B$.
Keeping the stored block sums in $B$ accounts for their exact
normalization.

The fixed pairwise contraction tree combines disjoint sets of original factors.
Expanding the floating-point calculation therefore perturbs each signed monomial by a product of elementary rounding factors.
Let $\varepsilon_{\mathrm{mach}}=2^{-b}$ be the unit roundoff for arithmetic with $b$ significand bits.
If $D$ bounds the number of rounding factors for every monomial, correctly rounded scalar addition, subtraction, and multiplication give
\begin{equation}
    |\widehat Z-Z|\leq\gamma_D B,
    \qquad \gamma_D=\frac{D2^{-b}}{1-D2^{-b}},
    \qquad D2^{-b}<1.
    \label{eq:fourier-absolute-bound}
\end{equation}
Indeed, the standard product bound limits each monomial's coefficient
perturbation by $\gamma_D$; summing absolute errors gives
\cref{eq:fourier-absolute-bound}.
This argument covers cancellation and
explains why a tiny result can require high precision even when the
absolute error bound is small.

For an unsliced contraction, let $U_j$ be the union of the index sets of the two input tensors at step $j$, and let $R_j$ be the index set of the resulting tensor.
The operation count for the fixed path includes
$2^{|U_j|}$ scalar multiplications and
$2^{|U_j|}-2^{|R_j|}$ scalar additions, so that
\[
    A=\sum_j\bigl(2^{|U_j|+1}-2^{|R_j|}\bigr),
    \qquad D=2A+2N+2q+10,
\]
where $q$ is the number of logical output bits.
The term $2A$ bounds the number of rounding factors from the contraction routines, including summation over indices present in only one operand and possible accumulation from zero.
The count $A$ already includes multiplications and additions separately.
The term $2N$ covers constructing both $w_{i0}+w_{i1}$ and $w_{i0}-w_{i1}$ for every block; $2q$ allows a sum or difference and a halving at each inverse Walsh--Hadamard stage along a monomial's evaluation.
The final $10$ provides an additional margin in the bound.

The sliced scalar executor uses the counts attached to its compiled
kernel list.
Within one slice, it charges one addition per input element
of each one-operand reduction and one multiplication plus one addition
per matrix-product term, including possible accumulation from zero.
Let $C$ be the sum of these kernel counts and $S=2^s$ the number of
assignments to the $s$ sliced binary indices.
A bound on the number of rounding factors is
\[
    D=S(C+1)+2N+10.
\]
The term $SC$ counts all slice contractions, and the additional $S$
counts their sequential outer sum from zero.
Here, the requested logical
label is included in the affine constraints, so no final logical Walsh
transform is needed.
Every unary factor associated with an original constraint retains its exact $\pm1/2$ normalization even when its index is sliced, and every binary slice assignment is evaluated once.
The numerical enclosure is formed after the complete slice sum.
Slicing changes
the execution cost and rounding count without changing the normalized
mathematical sum or $B$.

The implementation calculates $B$ rationally from the stored weights
and exports a checked power-of-two upper bound $B_+\geq B$ (with zero
permitted when $B=0$).
It forms the enclosure
\[
    [\max(0,\widehat Z-\gamma_DB_+),\;
    \widehat Z+\gamma_DB_+]
\]
with exact rational endpoint arithmetic.
A scalar is accepted at a requested relative tolerance only when the entire enclosure supports that tolerance.
The main color-code comparisons use fixed MP128 and require its error bound to establish the requested tolerance.
Earlier adaptive calculations try $53,80,112,128,160$ bits and include every attempt in the reported cost.

The arithmetic assumptions are exact conversion of finite FP64
inputs at $b\geq53$, round-to-nearest correctly rounded \texttt{mpmath}
scalar operations, and object-array kernels that use those operations.
Arbitrary exponents avoid finite-range underflow and overflow in this
model.
The bound in \cref{eq:fourier-absolute-bound} relies on these
arithmetic assumptions.
The distance-$3$, seven- and nine-round
postselection examples additionally use $160$-bit evaluations to check
agreement with the likelihood intervals from nonnegative Rank DP.

\subsection{Reed-Muller structural and code-capacity experiments}
\label{app:rm-capacity}

We use the construction in \cref{eq:rm-stabilizer-spaces,eq:rm-sector-matrix},
with qubits ordered by their nonzero binary labels $1,\ldots,2^m-1$
and monomial rows ordered by degree and then lexicographically.
Both logical representatives are all-ones vectors, retained at each Steane concatenation level.
For an error $(\vb{x},\vb{z})$, the implementation labels the logical sectors by the two commutation parities $\lambda=(\vb{1}^{\sfT}\vb{x},\vb{1}^{\sfT}\vb{z})$.
For a syndrome representative $(\vb{x}_s,\vb{z}_s)$, the residual logical label has coordinates $\ell=\lambda+(\vb{1}^{\sfT}\vb{x}_s,\vb{1}^{\sfT}\vb{z}_s)$ over $\FF_2$.
This syndrome-dependent relabelling preserves the likelihood values and maps their maximizing labels and ties accordingly.

\paragraph{Cut-ranks of the selected rank decomposition.}
The matrix $\vb{H}_X$ in \cref{eq:rm-sector-matrix} has the nonzero $m$-bit labels as columns.
Let $\vb{H}^+:=[\vb{H}_X^{\sfT},\vb{1}]^{\sfT}$ append an all-ones row to $\vb{H}_X$, so that $C_Z^\perp=\rowsp(\vb{H}^+)$.
For a set $U$ of physical qubits, let $\overline{U}$ be its complement and set $Y_U:=\bigcup_{i\in U}B_i$, where $B_i=\{i,n_m+i\}\in\calB_m$.
By \cref{eq:stabilizer-cut-connectivity}, the cut-rank of the reduced Tanner graph is
\begin{equation}
    \rho_{T(\vb{A}_m)}(Y_U)
    =\rank\vb{H}_{X,U}+\rank\vb{H}_{X,\overline{U}}
     +\rank\vb{H}^+_U+\rank\vb{H}^+_{\overline{U}}-(2m+1).
    \label{eq:rm-numerical-cut-rank}
\end{equation}
Thus, the structural sweep reduces to binary rank calculations in
dimensions $m$ and $m+1$.

The rank-decomposition tree recursively partitions the labels by successive binary coordinates.
Each nonroot node therefore represents an affine subspace with a fixed binary prefix, with zero omitted when present.
For a nonzero affine coset of dimension $r\leq m-2$,
the two local ranks in \cref{eq:rm-numerical-cut-rank} are $r+1,r+1$
and the complementary ranks are $m,m+1$, giving $2r+2$.
For a prefix subspace with zero removed and $2\leq r\leq m-2$,
the local ranks are $r,r+1$, giving $2r+1$.
A singleton has cut rank two, and both sides of the top hyperplane split
have cut-rank $2m-2$.
Hence, this tree has width exactly $2m-2$, giving $\rw(T(\vb{A}_m);\calB_m)\leq2m-2$.
These identities independently check every node cut for $m=3,\ldots,17$;
the sweep contains $524{,}227$ cuts.
Direct full-matrix ranks provide a
further check through $m=7$.

The concatenated codes $\calQ_{4,1}=[[105,1,9]]$ and $\calQ_{4,2}=[[735,1,27]]$ provide additional structural checks.
The selected hierarchical rank decompositions of $T(\vb{A}_{4,1})$ and $T(\vb{A}_{4,2})$ over their physical-qubit partitions both have width eight.

\paragraph{Finite-size treewidth bounds for the original TN graph.}
Let $q$ be the largest odd integer at most $m$,
$M_q=2\binom{q}{(q-1)/2}$, and $\Delta_q=(q+1)/2$.
The separator argument in \cref{eq:rm-middle-treewidth} gives the integer
lower bound $\lceil3M_q/(16\Delta_q)\rceil-1$.
For a complementary finite-size bound, fix
$a=\lfloor m/2\rfloor$ label coordinates.
Retain the $2^a-1$ nonempty
monomial checks on these coordinates, the all-ones logical check in the
same Pauli component, and the $2^{m-a}$ variable labels containing all
these coordinates.
The induced graph is $K_{2^a,2^{m-a}}$.
All selected monomials have degree at most $m-2$, so they are present in
the specified matrix.
Therefore, the plotted bound is
\begin{equation}
    \tw(G_m)\geq
    \max\left\{2^{\lfloor m/2\rfloor},
       \left\lceil\frac{3M_q}{16\Delta_q}\right\rceil-1\right\}.
    \label{eq:rm-numerical-treewidth-bound}
\end{equation}
The same biclique occurs inside one inner block of $G_{m,L}$, because
the root logical row restricts to all ones on that block.
These bounds concern the original graph with its specified check basis
and logical rows.
The timing controls below use algebraically
preprocessed representations.

\paragraph{Nonnegative TN controls.}
The candidate representations are categorical XOR chains and stabilizer-coset networks in the monomial and reduced row-echelon form (RREF) bases.
Each receives one greedy and two seeded random-greedy searches, capped at five seconds each.
For each representation, a feasible plan with the smallest estimated operation count is tested on the zero syndrome and on IID sampled syndromes.
The representation with the lowest median evaluation time is selected for the comparison.

A conditional parity-chain control first solves $\vb{H}_Z\vb{x}=\vb{s}_Z$ and $\vb{1}^{\sfT}\vb{x}=\lambda_x$.
Given one solution $\vb{x}_0$, all solutions are $\vb{x}=\vb{x}_0+\vb{H}_X^{\sfT}\alpha$ for $\alpha\in\FF_2^m$.
For each $\alpha$, a parity chain with $m+1$ state bits sums the joint weights $p_i(x_i,z_i)$ over $\vb{z}$ subject to $\vb{H}_X\vb{z}=\vb{s}_X$ and $\vb{1}^{\sfT}\vb{z}=\lambda_z$.
Summing over $\alpha$ for each logical label gives the four-sector vector in $O(n_m2^{2m+1})$ arithmetic operations.
The selected control is the RREF coset TN at $m=3$ and the conditional parity chain at $m=4,\ldots,10$.
The confirmations reuse these saved paths and implementations.

\paragraph{FP64 evaluation and accuracy.}

Both timed methods use FP64, allowing negligible sector masses to round to zero.
Independent reference intervals obtained from nonnegative Rank DP in native extended-precision arithmetic certify relative-$\ell_1$ error at most $10^{-10}$ for every complete four-sector vector.
Comparisons with the reference intervals use exact rational arithmetic.
The reference intervals also determine the possible ML winners.
\Cref{tab:rm-extension} records the number of queries whose intervals
separate a unique winner.

The confirmation covers $m=3,\ldots,10$ under depolarizing noise $p=.03$.
For each size, the seed $20260925000+100m$ generates $12$ IID
physical-error draws, retaining repeated syndromes, and each query
receives three evaluations per method.
Both methods reuse the same
stored FP64 weights.
Method order alternates by query on one pinned CPU
core with one numerical thread on the shared host.
Timings measure reused evaluation; planning, compilation, warmup,
reference generation, and external checks are accounted for separately.

The limits are $2$\,GiB of forecast numeric arrays and $4$\,GiB of process
address space.
Calls through $511$ qubits have a $60$-second limit;
the five $7$--$127$-qubit confirmations share a $300$-second worker limit,
and the $255$- and $511$-qubit confirmations each have an
$1800$-second worker limit.
At $1023$ qubits the limits are
$120$ seconds per call and $4500$ seconds for the worker.
All $12$ references, two warmups, and $72$ retained evaluations complete
at that size.

\begin{table}[!htbp]
\centering\small
\caption{CPU FP64 four-sector evaluations for the Reed-Muller family.
Each size uses $12$ IID queries and three repetitions. Times are medians
of per-query medians; TN/Rank DP is the median paired ratio.
The final column counts queries whose reference intervals certify one
strictly largest sector. At $n=1023$, all twelve queries have exactly two
tied maximizing sectors, as established by the symmetry argument below;
either label is an ML choice. All likelihood vectors meet the relative-$\ell_1$
target $10^{-10}$.
Reference generation and checking are accounted for separately from both
evaluation times.}
\label{tab:rm-extension}
\begin{tabular}{rrrrr}
\toprule
$n$ & Rank DP (s) & TN (s) & TN/Rank DP & Unique ML winners \\
\midrule
7 & 0.003055 & 0.0008087 & 0.268 & 12/12 \\
15 & 0.006916 & 0.001127 & 0.161 & 12/12 \\
31 & 0.01489 & 0.002572 & 0.172 & 10/12 \\
63 & 0.03237 & 0.007961 & 0.245 & 9/12 \\
127 & 0.07015 & 0.03726 & 0.531 & 8/12 \\
255 & 0.216 & 0.2832 & 1.305 & 2/12 \\
511 & 1.192 & 2.168 & 1.819 & 5/12 \\
1023 & 9.24 & 26.14 & 2.825 & 0/12 \\
\bottomrule
\end{tabular}
\end{table}

\paragraph{Exact twofold ties at $1023$ qubits.}
The reference intervals retain one pair of possible winners, either
$\{00,01\}$ or $\{10,11\}$, for each of the twelve $m=10$ queries.
An exact symmetry resolves these overlaps.
Fix any X pattern $\vb{x}$ compatible with the syndrome, and write $U=\supp\vb{x}$.
If a vector $\vb{v}$ satisfies
\begin{equation}
    \supp\vb{v}\subseteq U,\qquad \vb{H}_X\vb{v}=0,\qquad
    \vb{1}^{\sfT}\vb{v}=1,
    \label{eq:rm-sector-involution}
\end{equation}
then $\vb{z}\mapsto\vb{z}+\vb{v}$ preserves the syndrome and flips the second logical bit.
On the support of $\vb{v}$, the Pauli error changes between X and Y.
The depolarizing input has exactly equal stored X and Y weights, so this
map preserves the probability of every error configuration.
It is an
involution and hence pairs equal total masses in the two sectors.

For this family, $\ker\vb{H}_Z=C_X\oplus\Span\{\vb{1}\}$ has dimension $m+1$.
Thus, each syndrome has $2^{m+1}$ compatible X patterns, obtained from one representative by adding an X stabilizer and optionally $\vb{1}$.
For every one of the $12\times2^{11}=24{,}576$ patterns, exact binary
elimination produces an odd-weight null vector supported on $U$.
Equivalently,
$\rank[\vb{H}_{X,U};\vb{1}_U^{\sfT}]>\rank\vb{H}_{X,U}$, where the semicolon denotes vertical concatenation.
Independent checks verify the full kernel dimension, every witness's
support, and both parity conditions in \cref{eq:rm-sector-involution}.
Consequently,
\[
    P(s,\lambda_x,0)=P(s,\lambda_x,1)\qquad(\lambda_x=0,1)
\]
holds exactly for all twelve queries.
Together with the reference
intervals separating the two pairs, this certifies a twofold ML tie in
each case.
The zero unique-winner count in \cref{tab:rm-extension}
therefore has an algebraic explanation.

\paragraph{Symbolic contraction profiles of the original graph.}
The symbolic profiles retain one binary index per original-graph edge,
with unsplit COPY/parity tensors and the specified constraint bases.
One greedy and two seeded random-greedy searches receive five seconds
each.
The table records the symbolic operation and entry counts of
completed plans.
\begin{table}[!htbp]
\centering\small
\caption{Symbolic costs of retained pairwise contraction plans for the
original Reed-Muller TN graph with the specified constraint basis.
The columns give base-two logarithms of the
largest input tensor, largest generated intermediate, and scalar
operations for one sector, evaluated on the selected plans.}
\label{tab:rm-original-tn}
\begin{tabular}{rrrrr}
\toprule
$m$ & $n$ & Input entries & Intermediate entries & Scalar operations \\
\midrule
3 & 7 & 7 & 12 & 17.97 \\
4 & 15 & 15 & 32 & 39.71 \\
5 & 31 & 31 & 102 & 123.05 \\
6 & 63 & 63 & 343 & 373.01 \\
7 & 127 & 127 & 935 & 1397.00 \\
8 & 255 & 255 & --- & --- \\
\bottomrule
\end{tabular}
\end{table}
At $m=8$, the search reaches its time limit with the path unresolved;
the input tensor size remains known exactly.
Unsplit logical-parity
tensors require $2^n$ entries, motivating the algebraic preprocessing
used by the timed TN controls.

\subsection{Code-capacity comparisons and resource estimates}
\label{app:capacity-revision}

\paragraph{Code constructions.}
The main capacity comparison uses complete $4^k$ logical-sector vectors.
The open rectangular three-dimensional construction has
\begin{equation}
    n=abc+[(a-1)b+a(b-1)](c-1),\qquad d_X=ab,\quad d_Z=c.
\end{equation}
Take a cubic cell complex with vertices
$(i,j,k)\in\{0,\ldots,a-1\}\times\{0,\ldots,b-1\}\times\{0,\ldots,c\}$
and remove the two end caps $k=0,c$ as a relative boundary.
Qubits are the surviving edges: all longitudinal edges and transverse
edges at $k=1,\ldots,c-1$.
X checks are incidences at surviving vertices;
Z checks are boundaries of surviving faces, restricted to these edges.
A longitudinal string gives logical Z and a transverse sheet of
longitudinal edges gives logical X.
The $ab$ disjoint equivalent strings
and $c$ disjoint equivalent sheets give the matching distance bounds.
This construction agrees with the open-boundary ATNdecoding
implementation; its cubic geometry is also described by the rectification
below.
The tested rectangles set $a=b=3$ and $c=5,7,9$.

The algebraic comparisons use the quantum Golay $[[23,1,7]]$ code and a
doubled CSS $[[49,1,5]]$ code \cite{bravyi_2015_doubled,jain_albert_2024_transversal}.
The former uses cyclic shifts of the nonzero quadratic residues together
with zero; the latter applies the doubling construction first to Steane-$7$
and a trivial qubit, then to the $17$-qubit color code and the resulting
$15$-qubit code.
The $[[161,1,21]]$ concatenation uses Golay as the outer
code and Steane as each inner code.
The trivariate bicycle example TB30
is the $[[30,4,5]]$ construction of \cite{voss_2024_multivariate}, with
$l=3,m=5$, commuting cyclic shifts $x,y$, $z=xy$, and
$A=x+z^4$, $B=x+y^2+z^2$.
It requires $256$ logical outputs.
Construction
checks verify commutation, ranks, and logical operators, with independent enumeration where feasible.
The quoted code-family distances follow from the cited constructions.
The additional planar comparison uses a rotated distance-$13$ surface code.

\paragraph{Cohorts and timing.}
The three-dimensional and surface-$13$ comparisons fix the
selected rank decompositions and contraction paths before drawing $16$ fresh physical errors per
channel.
Each complete vector is timed three times per method, alternating
method order.
Every physical draw is retained, including repeated syndromes.
The Golay, doubled, and concatenated comparisons use $12$ draws per channel
and three repetitions; TB30 uses six draws per channel and two repetitions.
Channels are depolarizing
$p=.03$ and $(p_I,p_X,p_Y,p_Z)=(.945,.005,.045,.005)$.

Within each query, the reported method time is its median over repetitions.
The reported method time is the median of these query medians; the ratio
is the median of the ratios formed within each query.
Whiskers show the observed range over queries.
Sampling is IID within each channel; the two channels reuse the same
seeded random stream.
Each channel is analyzed as its own timing cohort.

Both methods use CPU FP64 and the same stored weights.
An independently checked nonnegative Rank DP calculation bounds the full likelihood vector, and every accepted result meets relative $\ell_1$ error $10^{-10}$.
Unique maximizing labels are compared when certified.
The vector criterion divides the sum of absolute sector errors by the total syndrome likelihood, as in \cref{eq:experiment-vector-accuracy}.
Timing follows the reused-evaluation convention of \cref{app:protocols}.

\begin{table}[!htbp]
\centering\small
\caption{Code-capacity timings for the correlated channel $(p_I,p_X,p_Y,p_Z)=(.945,.005,.045,.005)$.
Each code uses the same plans, implementations, and query counts as its depolarizing cohort.
The additional surface-$13$ comparison uses 16 queries and a cotengra TN implementation.
Brackets give observed ranges of paired query ratios.}
\label{tab:capacity-correlated}
\begin{tabular}{lrrr}
\toprule
Code & Rank DP (s) & TN (s) & Paired TN/Rank DP [range] \\
\midrule
$3\times3\times5$ patch & 0.6235 & 1.9400 & 3.105 [3.012, 3.187] \\
$3\times3\times7$ patch & 1.2295 & 4.0350 & 3.307 [3.115, 3.540] \\
$3\times3\times9$ patch & 1.6266 & 4.2821 & 2.629 [2.508, 2.741] \\
Surface & 0.3193 & 0.0320 & 0.099 [0.098, 0.112] \\
Golay & 0.0515 & 0.4549 & 8.831 [8.659, 8.938] \\
Doubled CSS & 0.0345 & 0.1943 & 5.625 [5.542, 5.675] \\
Golay--Steane & 0.1339 & 0.2321 & 1.754 [1.699, 1.799] \\
TB30 & 3.7759 & 3.2498 & 0.862 [0.850, 0.880] \\
\bottomrule
\end{tabular}
\end{table}

\paragraph{Tree construction and plan selection.}
The candidate rank decompositions include balanced and linear trees in the
constructor's physical-qubit order and its reverse, layer groupings for
the rectangles, and inner/outer groupings for concatenated codes.
Additional trees come from categorical and coset TN contraction paths:
retain one leaf for each physical-qubit block, remove auxiliary factor
leaves, and suppress nodes with a single remaining child.
Both the path's
child order and its reversal are considered.
For each tree, Rank DP computes the state spaces, scalar operation count, and live numeric-array storage.
For the rectangles, selection minimizes scalar work among candidates
within the $8$\,GiB array cap, breaking ties by memory and then candidate
label.
All three selected trees arise from coset contraction paths;
the plans are fixed before the IID timing samples are drawn.

\paragraph{TN representations and structural accounting.}
The main capacity paths are unsliced.
The TN search includes
categorical and coset TNs, equivalent stabilizer bases, greedy and seeded
random-greedy paths, and runtime checks of selected alternatives.
For the
doubled code, lower-weight and degree-balanced bases reduce the TN forecast
from roughly ten billion to forty million operations.
For the
concatenated code, exact inner-Steane contractions supply conditional
four-state weights to the outer-Golay TN.
Nonuniform-noise tests check all
sectors and their X/Z ordering.

The rectangle, Golay, doubled, concatenated, and TB30 comparisons use
opt\_einsum.
The surface comparison uses a rank decomposition of width $15$ and an
unsliced cotengra coset executor.

The selected rank decompositions of the rectangles' reduced Tanner graphs have width $18$ with respect to the physical-qubit partitions.
Their predicted scalar operation counts are $77{,}089{,}612$, $144{,}236{,}940$, and $207{,}208{,}780$ for the four-sector output.
Their live numeric-array forecasts are $18{,}877{,}760$, $18{,}879{,}104$, and $18{,}880{,}448$ bytes.
Linear trees ordered by layer have width $17$ but require more total join work.
TB30's
Rank DP/TN array forecasts are $4{,}424{,}640/228{,}741{,}008$ bytes.
These forecasts count live numeric arrays.
Process RSS additionally
includes interpreter, map, allocator, and backend storage.

\paragraph{Cases without completed timing comparisons.}
Quantum Hamming $[[15,7,3]]$ and Reed-Muller $[[16,6,4]]$ were screened
with all $16{,}384$ and $4096$ sectors.
The retained Rank DP/TN work ratios were
about $13.5$ and $4.67$; the current Rank DP implementation evaluates each logical sector separately.
These cases were excluded from timing comparisons on the basis of their resource estimates.
The $[[16,6,4]]$ code is the symmetric Reed-Muller construction;
\cref{subsec:rm-code-families} specifies the punctured family used above.
Other planar and concatenated screens, unconfirmed QR31 pilots, and all
stopped attempts retain their construction, forecast, or timeout status
in the experiment records.

\paragraph{Cubic geometry and the Vasmer--Browne construction.}
\label{app:cubic-capacity}
For $a=b=c=d$, the open surface code has
$n=3d^3-4d^2+2d$, with $d_X=d^2$ and $d_Z=d$.
Under rectification, edge qubits become vertices, X-star checks become
green octahedra, and Z-face checks become red--blue square faces.
This identifies the individual code with the $\mathrm{SC}_g$ component
of the Vasmer--Browne construction~\cite{vasmer_browne_2019_surface}.
Constructor checks verify commutation, check ranks, logical operators,
and the geometric string/sheet distances.
The cubic follow-up uses CPU FP64 depolarizing weights at $p=.03$ and the same full-vector target $10^{-10}$.
At $d=3$, a fresh 12-draw IID cohort with three repetitions gives median Rank DP/TN times $0.04815/0.00926$\,s and a median paired TN/Rank DP ratio $0.191$.
Both methods return all four sectors.
The TN evaluation uses cotengra with the saved basis and path.

The $d=4$ case, $[[136,1,4]]$, is a two-query, one-repeat IID pilot.
One accepted Rank DP evaluation completes in $89.56$\,s.
The first TN warmup exceeds its 600-second limit; the 1400-second overall limit interrupts the second reference evaluation and terminates the remaining worker processes.
Thus, one Rank DP output is accepted, but the paired runtime comparison remains incomplete.

At $d=5$, $[[285,1,5]]$, only binary ranks and tensor contraction dimensions are calculated.
The selected rank decomposition has width $37$ and a live numeric-array forecast of $9216$\,GiB.
The selected TN plans require $2^{70}$, $2^{62}$, $2^{58}$ slices under the 8, 24, and 64\,GiB caps, respectively, all beyond the limit of $65{,}536$ slices.
These estimates exceed the resource limits, so no likelihood evaluation is attempted.

The cubic studies use the same timing convention.
Independent checks
verify the constructors, query draws, output enclosures, and timing
aggregation.

\subsection{Certified color-circuit decoding and precision checks}
\label{app:circuit-mld-revision}

The main circuit table uses the complete stored models for color
$(5,5)$ and $(7,3)$, with gate, reset, and measurement rates all $.001$ and without an
additional idle channel.
The quiet syndrome and the first $31$ distinct nonzero sampled syndromes form the fixed query set described in \cref{app:protocols}.
Each method returns the two likelihoods $P^{\cl}(s,L_0=0)$ and $P^{\cl}(s,L_0=1)$, certifies each to relative error $10^{-8}$, and certifies a unique maximizing label.
Rank DP uses FP64; the signed detector-Fourier TN uses $128$ significand
bits in mpmath.
This fixes the requested two-sector output and arithmetic
policy throughout the comparison.

Each method and case receives three independent planning attempts capped
at $300$\,s each.
Previously constructed candidates, including the color
$(7,3)$ plans, are imported and reprofiled for the two-sector task.
Scalar Fourier paths are converted by removing the fixed logical-query
unary and retaining the logical-character index.
After the imported candidates, each attempt tries Fourier paths from
detector min-fill, chronological, and reverse-chronological elimination
orders, then seeded random-greedy searches, within its planning time limit.
Rows are adjacent when they share an event.
Min-fill eliminates the row
that adds the fewest edges among
its current neighbors, breaking ties by degree and then row index.
For Rank DP, a Fourier path supplies a candidate: retain its DEM-event leaves,
remove character-unary and auxiliary leaves, and suppress nodes with a
single remaining child.
Before these paths, the Rank DP search also
tries min-fill on the combined detector and observable rows.
Each event
is attached to its earliest eliminated incident row; binary joins combine
the event leaves along the resulting elimination tree.
Candidate profiles
use the requested two-sector output.
Selection minimizes feasible structural
work, then forecast memory and structural hash.

Three adjacent cold batch pairs alternate method order.
A worker starts,
loads the stored plan, compiles, evaluates all $32$ queries, performs its
own checks, and records the outputs.
These phases are included in batch
time.
Search and independent validation are excluded and recorded
separately.
The completed archive contains $384$ vectors and $768$
component intervals; an independent audit checks every one.
Earlier scalar-query studies are described in \cref{app:circuit-ledger}.

The precision diagnostic fixes query indices $0$--$3$ before new evaluation:
the quiet pattern and the first three stored sampled patterns in each case.
Native FP64 and $96$-bit Fourier TN use the identical saved path,
model, and slicing plan as the $128$-bit comparison.
Fresh Rank DP calls use
the saved rank decomposition.
The completed MP128 batch values provide the third precision point.
Reference comparisons are performed after the timed workers finish.

The independent accuracy check uses stored Rank DP likelihood intervals to enclose candidate
error relative to the exact stored-weight result.
It separately reports
whole-vector relative $\ell_1$ error at $10^{-10}$, each sector's relative
error at $10^{-8}$, and unique ML certification.
An error upper bound below
the target proves success; a lower bound above it proves failure; an
overlapping interval is unresolved.
Each method's own conservative
roundoff-error bound is recorded separately.
A conservative bound above
the target leaves that standalone accuracy check unresolved.
In particular, a small losing sector can fail component accuracy while
the full vector and the ML decision remain accurate.
A nonnegative TN
sums nonnegative weights; the available search outcomes for these
representations are given in \cref{app:circuit-ledger}.

\begin{table}[!htbp]
\centering\small
\caption{Four fixed precision-diagnostic queries per circuit. The first two count reference-certified whole-vector and both-sector accuracy. The last three count the method's own conservative certificates, for vector, both sectors, and unique ML respectively. Each count is out of four. The reference and standalone checks use their respective error bounds.}
\label{tab:circuit-precision-checks}
\begin{tabular}{llccccc}
\toprule
Circuit & Method & Ref. vector & Ref. sectors & Own vector & Own sectors & Own ML \\
\midrule
Color (5,5) & Fourier FP64 & 4 & 2 & 0 & 0 & 4 \\
Color (5,5) & Fourier MP96 & 4 & 4 & 4 & 4 & 4 \\
Color (5,5) & Fourier MP128 & 4 & 4 & 4 & 4 & 4 \\
Color (5,5) & Rank DP FP64 & 4 & 4 & 4 & 4 & 4 \\
Color (7,3) & Fourier FP64 & 4 & 0 & 0 & 0 & 4 \\
Color (7,3) & Fourier MP96 & 4 & 4 & 4 & 4 & 4 \\
Color (7,3) & Fourier MP128 & 4 & 4 & 4 & 4 & 4 \\
Color (7,3) & Rank DP FP64 & 4 & 4 & 4 & 4 & 4 \\
\bottomrule
\end{tabular}
\end{table}

\paragraph{Rank decompositions and resource limits.}
\Cref{tab:circuit-structure} gives the Rank DP structure for the two
circuits.
The work count sums leaf outcomes and compatible child-state
pairs over the two requested sectors.
Each child-state pair contributes
a product to one parent entry.
The width of the selected rank decomposition determines the largest DP table, whereas the total work depends on the compatible state pairs at every join.

\begin{table}[!htbp]
\centering\small
\caption{Circuit models and selected Rank DP plans. Array forecasts
count live numeric storage for FP64 evaluation.}
\label{tab:circuit-structure}
\begin{tabular}{lrrrrr}
\toprule
Circuit & Events & Detectors & Width & Work terms & Array bytes \\
\midrule
Color (5,5) & 1104 & 45 & 22 & $4{,}226{,}848{,}012$ & $86{,}001{,}104$ \\
Color (7,3) & 782 & 54 & 16 & $180{,}083{,}080$ & $5{,}026{,}208$ \\
\bottomrule
\end{tabular}
\end{table}

Each batch has a $7200$-second wall-clock limit.
Rank DP's numeric
array, address-space, and RSS limits are $10$, $12$, and $13$\,GiB;
the corresponding Fourier limits are $18$, $24$, and $25$\,GiB.
The plans allow at most $4096$ slices and $10^{12}$ operations under each method's counting convention.
These limits determine whether a plan is accepted for evaluation and when a worker is terminated.

\paragraph{Sufficient precision for the fixed query set.}
The stored reference intervals also give sufficient precision bounds without further contractions.
Let $L$ be the smaller of the two certified lower bounds on sector likelihoods, and let $G$ be a certified lower bound on the difference between the winning likelihood and the other likelihood.
At precision $b$, write $E_b=\gamma_D B_+$ for the Fourier absolute error bound, using the values of $D$ and $B_+$ for the stored contraction plan.
If
\begin{equation}
    L>2E_b,\qquad \frac{E_b}{L-2E_b}\leq10^{-8},\qquad G>4E_b,
\end{equation}
then its own likelihood intervals must meet the sector target and certify
the winner under the stated arithmetic assumptions.
Indeed,
$\widehat Z\geq L-E_b$, so the relative enclosure radius is at most
$E_b/(L-2E_b)$; the returned winner gap is at least $G-4E_b$.
For all $32$ queries these sufficient tests hold at $95$ and $104$ significand bits for color $(5,5)$ and $(7,3)$, respectively.
These operation-count bounds supply conservative sufficient precisions for the saved plans.
The timing table retains its common fixed MP128 policy.

\subsection{Circuit representation searches and outcomes}
\label{app:circuit-ledger}

The circuit comparisons consider categorical, nonnegative coset, and signed Fourier representations.
The searches below use the stated outputs and planning time limits.

Earlier scalar screens cover surface $(5,2)$, $(5,4)$, $(7,2)$ and color
$(3,7)$, $(5,4)$, $(7,3)$.
Their requested output is one all-quiet
observable-sector mass.
The current two-sector timing comparison instead
selects Rank DP and Fourier plans with three $300$-second attempts per
method and case, including imported candidates converted to the requested
output as described in \cref{app:circuit-mld-revision}.

For the two color circuits, a separate search over nonnegative TNs tries one
greedy categorical path and one greedy coset path, each with a
$10$-second planning limit.
The numeric-array forecast caps are $8$, $24$, and $64$\,GiB; planning workers have a $10$\,GiB RSS limit.
\Cref{tab:circuit-positive-search} records the actual planning outcomes.

\begin{table}[!htbp]
\centering\small
\caption{Nonnegative TN search outcomes for the task returning both sectors
of one measured observable.
Each row has a $10$-second planning limit and one greedy candidate.}
\label{tab:circuit-positive-search}
\begin{tabular}{lll}
\toprule
Circuit & Representation & Planning outcome \\
\midrule
Color (5,5) & Categorical & Time limit during path profiling \\
Color (5,5) & Nonnegative coset & Time limit during path profiling \\
Color (7,3) & Categorical & Recursion limit in path profiling \\
Color (7,3) & Nonnegative coset & Time limit during path profiling \\
\bottomrule
\end{tabular}
\end{table}

None of these four searches produces a usable contraction plan within its limits.
The completed runtime comparison
uses the Rank DP and Fourier plans and accuracy requirements in
\cref{app:circuit-mld-revision}.
Broader representation searches and
execution optimization of both methods remain directions for future work.

\subsection{Circuit noise learning, derivatives, and statistical uncertainty}
\label{app:circuit-learning-revision}

\paragraph{Physical rates and the DEM map.}
The color $(7,3)$ experiment demonstrates Rank DP as an exact
evaluator for likelihood-based calibration~\cite{cao_2026_differentiable}:
it fits the three physical gate, reset, and measurement rates
$\theta=(g,r,m)$ from syndrome likelihoods through their fixed map to DEM priors.
The circuit has $581$ physical noise-channel instances, of which $563$
have nonzero detector/observable signatures, and $782$ DEM events.
A signature is the binary vector of detector and observable flips
produced by a physical fault.
Independent events with the same signature combine into one DEM event, whose factor $1-2p_j$ is the product of the corresponding factors for the individual events.

For these circuits, the exact map from physical rates to DEM event probabilities is
\begin{equation}
    p_j(\theta)=\tfrac12\left[1-\exp\!\left(\sum_{a=1}^4
    C_{ja}f_a(\theta)\right)\right],\qquad
    \vb{f}(\theta)=
    \begin{pmatrix}
    \log(1-4g/3)\\ \log(1-16g/15)\\ \log(1-2r)\\ \log(1-2m)
    \end{pmatrix}.
    \label{eq:calibration-rate-map}
\end{equation}
The four columns of $\vb{C}$ correspond to one-qubit depolarization,
two-qubit depolarization, reset flips, and measurement flips.
For a depolarizing channel with a binary signature image of dimension
$h$, each nonzero signature contributes $2^{1-h}$ to its column.
To see this, a nontrivial character of that image changes sign on $2^{h-1}$ signatures.
Assigning the factor $(1-4g/3)^{2^{1-h}}$ to each nonzero signature, or $(1-16g/15)^{2^{1-h}}$ for two qubits, reproduces the channel's character expectation.
Reset and measurement flips contribute coefficient one.
Multiplying these factors for independent channels gives \cref{eq:calibration-rate-map}.
The implementation checks that each signature image is a binary subspace and that its nonzero signatures have equal multiplicities among the physical fault outcomes.
The event supports remain fixed as the rates vary.
The formula applies on $0<g<3/4$ and $0<r,m<1/2$.

\paragraph{Forward sensitivities.}
At an internal node $u$ with children $u_1,u_2$, differentiating the Rank DP update in \cref{eq:block-join} gives
\begin{equation}
    \begin{aligned}
        \partial_{\theta_j}F_{Y_u}(\vb{h})
        =\sum_{\substack{\vb{h}_1\in I_{Y_{u_1}},\,\vb{h}_2\in I_{Y_{u_2}}\\
                        \vb{h}_1+\vb{h}_2=\vb{h}}}
        \Bigl[&\bigl(\partial_{\theta_j}F_{Y_{u_1}}(\vb{h}_1)\bigr)F_{Y_{u_2}}(\vb{h}_2)\\
              &+F_{Y_{u_1}}(\vb{h}_1)\bigl(\partial_{\theta_j}F_{Y_{u_2}}(\vb{h}_2)\bigr)\Bigr].
    \end{aligned}
    \label{eq:calibration-sensitivity}
\end{equation}
Here, $j=1,2,3$ indexes the gate, reset, and measurement rates, and leaf derivatives follow from \cref{eq:calibration-rate-map}.
The likelihood and derivative arrays use the same exact power-of-two scaling.
The score is $\partial_{\theta_j}\log Z=(\partial_{\theta_j}Z)/Z$.
The mean NLL in \cref{eq:learning-objective} therefore has gradient
\[
    \partial_{\theta_j}\calL_N(\theta)
    =-\frac1N\sum_{i=1}^N
    \frac{\partial_{\theta_j}Z_{\vb{D}}(s_i;\theta)}
      {Z_{\vb{D}}(s_i;\theta)}.
\]
The likelihood is evaluated by nonnegative Rank DP with the roundoff-error bound in \cref{prop:roundoff-error}.
This bound does not apply to the signed derivatives, which are checked by the independent comparisons below.
The compiled kernels disable fast-math and fused multiply-add.

\paragraph{Training and held-out evaluation.}
The generating rates are $(.0012,.0007,.0016)$ and initialization is
$(.0005,.0018,.0006)$.
Each of three independent datasets contains
$4096$ training, $512$ validation, and $2048$ test circuit shots.
The optimizer variables are logits $\eta_a$, with
$\theta_a=.25\,\sigmoid(\eta_a)$.
Adam uses $64$ updates, batch size $64$, learning rate $.03$,
$(\beta_1,\beta_2)=(.9,.999)$, and $\epsilon=10^{-8}$.
Independent streams generate the three splits and minibatch schedules.
Each batch samples training positions without replacement; successive
batches are independent samples from the fixed training pool.

Validation NLL selects among updates $0,8,16,\ldots,64$, including
initialization.
Training and selection use detector outcomes.
The main plot's mean relative DEM-prior error is a diagnostic against
the known generating model; it weights all $782$ events equally and
shows all $65$ iterates.
The distinct objectives explain why its minimum
can occur at a different update from the selected validation checkpoint.
Test likelihoods are evaluated after selection.
For the paired per-shot
improvement $D_i=\mathrm{NLL}_{\rm init}(s_i)-\mathrm{NLL}_{\rm fit}(s_i)$,
the approximate $95\%$ interval is
$\overline{D}\pm1.96s_D/\sqrt{2048}$, where $s_D$ is the sample standard
deviation.
Each interval quantifies test-sample uncertainty conditional
on its selected fitted model.

\begin{table}[!htbp]
\centering\small
\caption{Validation-selected physical rates and held-out NLL improvement
for color $(7,3)$. Rates are in units of $10^{-3}$; brackets give
conditional approximate $95\%$ paired test intervals in nats per syndrome.}
\label{tab:color-circuit-learning}
\begin{tabular}{crrrrr}
\toprule
Dataset & Update & Gate & Reset & Measurement & NLL improvement [interval] \\
\midrule
1 & 48 & 1.171 & 1.553 & 1.191 & 0.1002 [0.0786, 0.1218] \\
2 & 64 & 1.126 & 1.188 & 1.248 & 0.0925 [0.0731, 0.1120] \\
3 & 64 & 1.133 & 0.995 & 1.409 & 0.0953 [0.0753, 0.1153] \\
\bottomrule
\end{tabular}
\end{table}

\paragraph{Independent validation.}
Fresh exact Stim conversions check the rate map at four interior points,
with maximum absolute event-prior difference $6.73\times10^{-16}$.
They also check all $195$ probability vectors from the saved iterates.
Small test instances are checked by complete DEM-event enumeration and independent TN differentiation.
On the color $(7,3)$ circuit, likelihoods agree with an independent calculation of the full likelihood vector and with Fourier TN evaluation.
Central finite differences at physical-rate steps $10^{-7}$ and $3\times10^{-8}$ check all three derivatives on three fixed syndromes, using new Stim conversions and Fourier TN evaluations at each perturbed rate.
The largest recorded Euclidean gradient difference is $4.97\times10^{-6}$.
The acceptance tolerance is
$3\times10^{-7}\max(1,\|\nabla_\theta\log P_\theta(s)\|_2)$.
The experiment audit replays circuit sampling, minibatch generation,
optimizer arithmetic, selection, and paired test statistics.
Each fit uses a pinned CPU core and a $2$\,GiB address-space limit.
Including held-out evaluation, the three runs take $69.9$, $67.3$, and
$67.2$ minutes.

\paragraph{Local identifiability and parameter uncertainty.}
The selected reset rates are $1.4$--$2.2$ times their generating value.
To examine sensitivity, let $\vb{D}$ be the detector-event incidence matrix
and $p_j$ the generating DEM priors.
The detector parity characters and
means are
\[
    t_a=\prod_j(1-2p_j)^{D_{aj}},\qquad \mu_a=(1-t_a)/2.
\]
Define $t_{a\oplus b}$ using the XOR of detector rows $a$ and $b$.
The single-shot covariance of their binary outcomes is then
\begin{equation}
    \Sigma_{ab}=\frac{t_{a\oplus b}-t_at_b}{4}.
    \label{eq:calibration-moment-covariance}
\end{equation}
The $54\times3$ Jacobian $J=\partial\mu/\partial\log\theta$ is evaluated
at the generating rates.
Its singular values are $.16289$, $.03611$, and
$.01046$, giving condition number $15.57$.
Its full column rank gives local identifiability through detector means.
The positive-definite $\Sigma$ also accounts for correlations between
detectors.
An optimally weighted linearized moment estimator based on
$N$ independent shots has local log-rate covariance
\begin{equation}
    V_{\rm mom}=\frac1N(J^{\sf T}\Sigma^{-1}J)^{-1}.
    \label{eq:calibration-local-resolution}
\end{equation}
For $N=4096$, its gate/reset/measurement standard errors are
approximately $4.7\%$, $29.3\%$, and $7.5\%$ in relative-rate units.
These values describe the local resolution of the specified moment
estimator at the generating model.
They quantify the unequal sensitivity
of the three physical rates, with reset noise the least resolved.
Fresh circuit conversions and finite differences independently check
$J$, $\Sigma$, and the resulting resolution calculation.

\paragraph{Decoding with the fitted noise models.}
After selecting the three fitted models by validation NLL, we draw $1024$ fresh IID circuit shots at the generating rates, with seed $2026092707$.
All multiplicities are retained, giving $346$ distinct syndromes.
For each syndrome, the initial, generating, and three fitted models
return both observable-sector masses and a uniquely certified ML label.
The compiled nonnegative Rank DP evaluates each sector with the observable row appended to the detector constraints.
Its FP64 rounding-factor bound is $\delta=14397$, and interval endpoints are computed with exact rational arithmetic.
Four fixed syndromes per model are also checked by an independent calculation of the full likelihood vector.
For each fitted model, the primary estimand is the reduction in failure probability relative to the initial model, estimated using paired conditional failure probabilities under the generating-model posterior.
The sample size and $3600$-second limit per model are fixed before sampling.

All five models choose the same label on every sampled syndrome and
have zero realized failures in the $1024$ shots.
Their common sample
mean conditional failure probability is $9.76314\times10^{-6}$.
Each paired difference is exactly zero on this cohort; applying
\cref{eq:decoder-bernstein} gives a per-comparison $95\%$ interval
contained in $[-.0200,.0200]$, or approximately $\pm2.00$ percentage
points.
This interval bounds the difference in overall failure probabilities under the generating model.
Conditional on the sampled syndromes, the common mean conditional failure probability corresponds to approximately $.0100$ expected failures in $1024$ shots.
Thus, the improvement in held-out likelihood does not change the ML decisions on this sample.
The independent audit replays all shots, verifies every winner
inequality, and recomputes the paired statistics.

\subsection{Postselection events and numerical intervals}
\label{app:postselection-details}

Postselection requires separate bounds on the acceptance probability and the probability of failure with acceptance.
For a model $\nu\in\{\cc,\cl\}$, suppose $[L_A,U_A]$ encloses $P^\nu(A)$ and $[L_B,U_B]$ encloses $P^\nu(A\cap B)$, with $L_A>0$ and $L_B\geq0$.
Then,
\begin{equation}
    \frac{L_B}{U_A}\leq P^\nu(B\mid A)\leq\frac{U_B}{L_A}.
    \label{eq:postselection-interval}
\end{equation}

\paragraph{Static grid and sampling checks.}
Static postselection uses all-zero syndrome, identity recovery, and the
disjoint union of all nontrivial logical sectors as $B$.
The grid crosses
surface-$3$, surface-$5$, and the $[[17,1,5]]$ color code with seven rates
$p\in\{10^{-5},3\times10^{-5},10^{-4},3\times10^{-4},10^{-3},
3\times10^{-3},10^{-2}\}$ and six channel choices.
In public $I,X,Y,Z$
order, these are
\begin{equation*}
    \begin{aligned}
    \text{depolarizing}:\quad &(1-p,p/3,p/3,p/3),\\
    \text{Z-biased}:\quad &(1-p,p/22,p/22,10p/11),\\
    \text{correlated}:\quad &(1-(2-c)p,(1-c)p,cp,(1-c)p),
    \end{aligned}
\end{equation*}
with $c\in\{0,.05,.5,.9\}$ for the correlated family.
The bias is
$p_Z/(p_X+p_Y)=10$.
Here, $p$ is the total nonidentity probability for the
first two families and each binary marginal for the correlated family,
whose total nonidentity probability is $(2-c)p$.
All $3\times7\times6=126$
cases are evaluated by nonnegative Rank DP on a CPU, with likelihood intervals certifying relative error at most $10^{-10}$.
At depolarizing $p=10^{-5}$, \cref{tab:postselection-enclosures} gives the three codes' conditional failure probabilities and their numerical enclosures.
The displayed intervals round the computed endpoints outwards.

\begin{table}[!htbp]
\centering\small
\caption{Zero-syndrome postselection at depolarizing $p=10^{-5}$,
with identity recovery.
$A$ is acceptance, and $B$ is the event of a nontrivial logical error.
The interval for $P(B\mid A)$ follows from separate bounds on $P(A)$ and $P(A\cap B)$.}
\label{tab:postselection-enclosures}
\begin{tabular}{lrl}
\toprule
Code & $P(A)$ (rounded) & Numerical enclosure of $P(B\mid A)$ \\
\midrule
Surface $[[9,1,3]]$ & $.9999100036$ &
$[8.8891555612839,\;8.8891555612842]\times10^{-16}$ \\
Surface $[[25,1,5]]$ & $.9997500301$ &
$[6.5846913682157,\;6.5846913682165]\times10^{-26}$ \\
Color $[[17,1,5]]$ & $.9998300136$ &
$[6.2966111212538,\;6.2966111212570]\times10^{-26}$ \\
\bottomrule
\end{tabular}
\end{table}

A separate predetermined sampling check uses one million draws for each
of four depolarizing cases: surface-$3$ at $p=.03$ and $.10$, and surface-$5$
and color-$17$ at $p=.10$.
The base seed is 2026091905, with case-index
spawn keys $0,1,2,3$ in that order.
Intervals are equal-tailed
Clopper--Pearson $95\%$ intervals per estimand, using all draws for
unconditional probabilities and accepted draws for the conditional
probability.
The surface-$3$, $p=.10$ sample has 530 accepted logical failures
versus expectation 477.59; its unconditional and conditional intervals
both lie above the exact value.
The two intervals use the same 530
events.
The full sample is retained, and independent physical-error
enumeration agrees with the Rank DP calculation.

\paragraph{Circuit events and partial acceptance.}
The circuit examples below use rotated Z memory with gate/reset/measurement
rates $.001$, no extra idle noise, and identity recovery.
The failure event
is the final measured memory-observable flip $L_0=1$.
\Cref{tab:partial-postselection} compares all-quiet acceptance with first-round
acceptance, which marginalizes the later detectors.
For a selected
set $J$ of constrained detectors, acceptance is
$A_J=\{s_j=0\text{ for all }j\in J\}$.
Its probability is computed using
the rows $\vb{D}_J$ alone; appending the observable row and fixing its value
to one gives $P(A_J,L_0=1)$.
This directly sums over every outcome of
the unconstrained detectors.

\begin{table}[!htbp]
\centering\small
\begin{tabular}{rlrr}
\toprule
Rounds & Acceptance condition & $P^{\cl}(A)$ & $P^{\cl}(L_0=1\mid A)$ \\
\midrule
7 & All detectors quiet & .7230370 & $5.49453\times10^{-8}$ \\
7 & First round quiet & .9717330 & $.0263165$ \\
9 & All detectors quiet & .6617884 & $5.75289\times10^{-8}$ \\
9 & First round quiet & .9717330 & $.0333414$ \\
\bottomrule
\end{tabular}
\caption{Distance-$3$ rotated Z-memory postselection with gate/reset/measurement
rates $.001$ and no extra idle noise. Identity recovery is fixed. Later
detectors are marginalized when only the first round is constrained; the
failure event remains the final measured memory-observable flip.}
\label{tab:partial-postselection}
\end{table}

Rank DP gives separate numerical bounds for the probabilities under both acceptance conditions.
\Cref{app:circuit-mld-revision} compares the numerical accuracy of individual likelihoods, complete likelihood vectors, and ML decisions.

\subsection{Decoder assessment and independent checks}
\label{app:decoder-details}

\paragraph{Decoder settings and cohorts.}
Quaternary BP+OSD provides a related correlation-aware
approach~\cite{kung_2023_quaternary}.
At each qubit, the joint Pauli channel gives the conditional priors
\begin{align*}
    P(x=1\mid z=0)&=\frac{p_X}{p_I+p_X},&
    P(x=1\mid z=1)&=\frac{p_Y}{p_Z+p_Y},\\
    P(z=1\mid x=0)&=\frac{p_Z}{p_I+p_Z},&
    P(z=1\mid x=1)&=\frac{p_Y}{p_X+p_Y}.
\end{align*}
Marginal BP+OSD~\cite{panteleev_osd} first solves the two binary syndrome problems with
priors $p_X+p_Y$ and $p_Z+p_Y$.
Starting from that correction, the
conditional heuristic uses the current Z estimate to update X with
the first pair of priors, then updates Z with the second pair; the
opposite orientation updates Z first.
Each orientation restarts from
the same marginal correction.
All syndrome-valid
intermediates are retained; the highest joint physical-error probability
wins, with lexicographic correction bits breaking ties.
Low effort uses one
conditional sweep per orientation, ten BP iterations, and order-zero OSD
(OSD-0).
High effort uses four sweeps, fifty BP iterations, and
combination-sweep OSD (OSD-CS) of order two.
BP uses the sum-product algorithm, parallel scheduling, and one thread.
These settings
jointly specify the iteration count, number of sweeps, and OSD search.

Both binary marginals are $.05$, with $p_Y=.0025,.025,.045$ as in
\cref{subsec:decoder-gap}.
The distance-$3$ population calculation sums all
256 syndromes under the true joint distribution.
After a separate
256-shot pilot, the distance-$5$ study freezes three fresh 8192-shot
cohorts, using seeds 2026091811, 2026091812, and 2026091813 in noise-point
order.
The primary contrast is marginal-low minus conditional-low at
$p_Y=.045$, evaluated at the fixed sample size.

\paragraph{Sampling and numerical intervals.}
For $N$ IID physical draws $(s_i,\ell_i)$, the conditional failure probability is $R_\delta(s_i)$, the excess conditional failure probability is $\Delta_\delta(s_i)$ from \cref{eq:decoder-regret}, and the realized failure indicator is $\mathbf{1}\{\delta(s_i)\ne\ell_i\}$.
Each lies in $[0,1]$; a paired difference between two decoders lies in $[-1,1]$.
Let $V_i$ denote the reported per-shot quantity, $s_V^2$ its unbiased sample variance, and $R=b-a$ the width of its known support $[a,b]$.
Thus, $R=1$ for conditional failure probabilities, excess conditional failure probabilities, and failure indicators, and $R=2$ for paired differences.
Applying the empirical-Bernstein inequality
of Maurer and Pontil~\cite[Theorem~4]{MaurerPontil2009} to both signs and
taking a union bound gives the two-sided statistical radius
\begin{equation}
    \rho_{\rm stat}=
    \sqrt{\frac{2s_V^2\log(4/\alpha)}{N}}
    +\frac{7R\log(4/\alpha)}{3(N-1)}.
    \label{eq:decoder-bernstein}
\end{equation}
The reported interval is
$[\overline{V}-\rho_{\rm stat}-\rho_{\rm num},
  \overline{V}+\rho_{\rm stat}+\rho_{\rm num}]\cap[a,b]$,
with $\alpha=.05$.
If $\epsilon$ is the largest certified posterior $\ell_1$ error bound over the cohort's distinct syndromes, the implementation conservatively uses $\rho_{\rm num}=2\epsilon$ for conditional failure probabilities and excess conditional failure probabilities and $4\epsilon$ for paired differences of conditional failure probabilities.
Realized failure indicators and their paired differences use $\rho_{\rm num}=0$.
Each interval has the stated coverage for its individual estimand.
Paired
differences retain the same physical draws, and repeated syndromes retain
their multiplicities.

\paragraph{Population and sampled results.}
At distance $3$ and $p_Y=.045$, the joint Bayes, marginal BP+OSD,
and conditional-low population failure probabilities are
$.01757028$, $.06969064$, and $.04010803$.
Conditional-low reduces the failure probability by $2.9583$ percentage points relative to marginal BP+OSD.
Conditional-high has $.0003846522$ greater population failure probability
than low, with a numerical interval excluding zero.
Thus, the low setting has the smaller overall failure probability for this channel.
Exact factorized inference and marginal BP+OSD have the same reported
failure probabilities at both correlated points in the distance-$3$
population and distance-$5$ cohorts.

At the strongest distance-$5$ correlation, actual failures are
$386/8192$ for marginal, $118/8192$ for conditional-low, $140/8192$ for
conditional-high, and $30/8192$ for joint Bayes.
The marginal-minus-low
paired failure reduction is $.03271484$ with interval
$[.02390427,.04152542]$.
The low-minus-high interval for the difference
in mean conditional failure probabilities is $[-.00927801,.00192004]$
and includes zero.
At the independent-noise point, the practical decoder's
observed mean excess conditional failure probability is $.0000723871$,
with interval $[0,.00153497]$; all distance-$3$ independent-control gaps
are zero.
Averages of conditional failure probabilities and observed
failure fractions estimate the same population quantity on the same
draws.
The conditional-probability estimator integrates out the remaining
logical-outcome randomness and reduces variance by conditioning on the
syndrome.

All nine population, pilot, and confirmatory calculations are completed, with all sampled draws retained.
The independent audit
checks 44\,656 stored corrections, their syndrome and logical maps,
retained candidate probabilities, paired statistics, and archived source files.
It enumerates all 262\,144 physical Pauli assignments for each distance-$3$
noise point as a floating-point cross-check at tolerance $10^{-12}$.
The population intervals use the numerical error bounds for nonnegative Rank DP and exact rational accumulation.
These comparisons quantify decoder failure probabilities under the fixed settings above.

\subsection{Computational environment and validation}
\label{app:evidence-index}

The principal CPU studies use a shared Intel Xeon Platinum 8268 host
running Linux, with workers pinned to CPU cores and BLAS/OpenMP thread
counts fixed to one.
The software environment uses Python 3.11.15 and
NumPy 2.4.3.
Circuit generation and sampling use Stim 1.15.0~\cite{stim}.
TN path search and execution use cotengra 0.7.5~\cite{cotengra} and
opt\_einsum 3.4.0~\cite{opt_einsum}; the selected representations and
executors are specified in
\cref{app:rm-capacity,app:capacity-revision,app:circuit-mld-revision}.
BP+OSD uses \texttt{ldpc} 2.4.1.
Multiprecision Fourier calculations
use \texttt{mpmath} 1.3.0 with the
arithmetic assumptions in \cref{app:fourier-certificate}.
The native extended-precision reference format used for the Reed-Muller experiments has $64$ significand
bits and normal exponent range $[-16382,16383]$, stored in $16$ bytes
on this platform.

Independent validation checks model and query identities, likelihood
enclosures, and the reported statistics.
The experiment-specific
subsections give the algebraic, enumeration, sampling, and calibration
checks, including their numerical tolerances.
\FloatBarrier

\end{document}